\documentclass{article}
\usepackage{amsmath}
\usepackage{amsfonts}
\usepackage{amssymb}
\usepackage{braket}
\usepackage[utf8]{inputenc}
\usepackage[T1]{fontenc}
\usepackage[colorlinks=true]{hyperref}
\hypersetup{%
  colorlinks = true,
  linkcolor  = black
}
\usepackage{graphicx}
\usepackage[dvipsnames]{xcolor}
\usepackage{authblk}
\usepackage{soul}
\usepackage{cite}

\usepackage{amsthm}
\usepackage{comment}

\usepackage{enumitem}

\newtheorem{lemma}{Lemma}
\newtheorem{proposition}{Proposition}
\newtheorem{corol}{Corollary}

\newtheorem{theorem}{Theorem}
\newtheorem{definition}{Definition}
\newtheorem{example}{Example}

\def\be{\begin{equation}}
\def\ee{\end{equation}}
\def\ba{\begin{eqnarray}}
\def\ea{\end{eqnarray}}

\newcommand\q{\quad}
\def\Nl{{\mathchoice
{\setbox0=\hbox{$\displaystyle\rm N$}\hbox{\hbox to0pt
{\kern0.4\wd0\vrule height0.9\ht0\hss}\box0}}
{\setbox0=\hbox{$\textstyle\rm N$}\hbox{\hbox to0pt
{\kern0.4\wd0\vrule height0.9\ht0\hss}\box0}}
{\setbox0=\hbox{$\scriptstyle\rm N$}\hbox{\hbox to0pt
{\kern0.4\wd0\vrule height0.9\ht0\hss}\box0}}
{\setbox0=\hbox{$\scriptscriptstyle\rm N$}\hbox{\hbox to0pt
{\kern0.4\wd0\vrule height0.9\ht0\hss}\box0}}}}
\def\Zl{{\mathchoice
{\setbox0=\hbox{$\displaystyle\rm Z$}\hbox{\hbox to0pt
{\kern0.4\wd0\vrule height0.9\ht0\hss}\box0}}
{\setbox0=\hbox{$\textstyle\rm Z$}\hbox{\hbox to0pt
{\kern0.4\wd0\vrule height0.9\ht0\hss}\box0}}
{\setbox0=\hbox{$\scriptstyle\rm Z$}\hbox{\hbox to0pt
{\kern0.4\wd0\vrule height0.9\ht0\hss}\box0}}
{\setbox0=\hbox{$\scriptscriptstyle\rm Z$}\hbox{\hbox to0pt
{\kern0.4\wd0\vrule height0.9\ht0\hss}\box0}}}}
\def\Ql{{\mathchoice
{\setbox0=\hbox{$\displaystyle\rm Q$}\hbox{\hbox to0pt
{\kern0.4\wd0\vrule height0.9\ht0\hss}\box0}}
{\setbox0=\hbox{$\textstyle\rm Q$}\hbox{\hbox to0pt
{\kern0.4\wd0\vrule height0.9\ht0\hss}\box0}}
{\setbox0=\hbox{$\scriptstyle\rm Q$}\hbox{\hbox to0pt
{\kern0.4\wd0\vrule height0.9\ht0\hss}\box0}}
{\setbox0=\hbox{$\scriptscriptstyle\rm Q$}\hbox{\hbox to0pt
{\kern0.4\wd0\vrule height0.9\ht0\hss}\box0}}}}
\def\Rl{{\mathchoice
{\setbox0=\hbox{$\displaystyle\rm R$}\hbox{\hbox to0pt
{\kern0.4\wd0\vrule height0.9\ht0\hss}\box0}}
{\setbox0=\hbox{$\textstyle\rm R$}\hbox{\hbox to0pt
{\kern0.4\wd0\vrule height0.9\ht0\hss}\box0}}
{\setbox0=\hbox{$\scriptstyle\rm R$}\hbox{\hbox to0pt
{\kern0.4\wd0\vrule height0.9\ht0\hss}\box0}}
{\setbox0=\hbox{$\scriptscriptstyle\rm R$}\hbox{\hbox to0pt
{\kern0.4\wd0\vrule height0.9\ht0\hss}\box0}}}}
\def\Cl{{\mathchoice
{\setbox0=\hbox{$\displaystyle\rm C$}\hbox{\hbox to0pt
{\kern0.4\wd0\vrule height0.9\ht0\hss}\box0}}
{\setbox0=\hbox{$\textstyle\rm C$}\hbox{\hbox to0pt
{\kern0.4\wd0\vrule height0.9\ht0\hss}\box0}}
{\setbox0=\hbox{$\scriptstyle\rm C$}\hbox{\hbox to0pt
{\kern0.4\wd0\vrule height0.9\ht0\hss}\box0}}
{\setbox0=\hbox{$\scriptscriptstyle\rm C$}\hbox{\hbox to0pt
{\kern0.4\wd0\vrule height0.9\ht0\hss}\box0}}}}
\def\Hl{{\mathchoice
{\setbox0=\hbox{$\displaystyle\rm H$}\hbox{\hbox to0pt
{\kern0.4\wd0\vrule height0.9\ht0\hss}\box0}}
{\setbox0=\hbox{$\textstyle\rm H$}\hbox{\hbox to0pt
{\kern0.4\wd0\vrule height0.9\ht0\hss}\box0}}
{\setbox0=\hbox{$\scriptstyle\rm H$}\hbox{\hbox to0pt
{\kern0.4\wd0\vrule height0.9\ht0\hss}\box0}}
{\setbox0=\hbox{$\scriptscriptstyle\rm H$}\hbox{\hbox to0pt
{\kern0.4\wd0\vrule height0.9\ht0\hss}\box0}}}}
\def\Ol{{\mathchoice
{\setbox0=\hbox{$\displaystyle\rm O$}\hbox{\hbox to0pt
{\kern0.4\wd0\vrule height0.9\ht0\hss}\box0}}
{\setbox0=\hbox{$\textstyle\rm O$}\hbox{\hbox to0pt
{\kern0.4\wd0\vrule height0.9\ht0\hss}\box0}}
{\setbox0=\hbox{$\scriptstyle\rm O$}\hbox{\hbox to0pt
{\kern0.4\wd0\vrule height0.9\ht0\hss}\box0}}
{\setbox0=\hbox{$\scriptscriptstyle\rm O$}\hbox{\hbox to0pt
{\kern0.4\wd0\vrule height0.9\ht0\hss}\box0}}}}
\newcommand{\ca}{\mathcal A}

\newcommand{\cb}{\mathcal B}

\newcommand{\cg}{\mathcal G}
\newcommand{\ch}{\mathcal H}
\newcommand{\ci}{\mathcal I}

\newcommand{\cl}{\mathcal L}

\newcommand{\cn}{\mathcal N}

\newcommand{\calr}{\mathcal R}

\newcommand{\ct}{\mathcal T}

\newcommand{\cv}{\mathcal V}

\newcommand{\leon}[1]{{\color{Maroon}[Leon: #1]}}
\newcommand{\ip}[2]{\left\langle\,#1\,|\,#2\,\right\rangle} 

\def\nn{\nonumber}

\newcommand{\eqa}{\begin{eqnarray}}
\newcommand{\neqa}{\end{eqnarray}}

\def\f{\frac}

\usepackage{bm}
\usepackage{bbm}
\def\C{{\mathbbm C}}

\newcommand{\lalg}[1]{\mathfrak{#1}}  
\newcommand{\SU}{\mathrm{SU}}
\newcommand{\SO}{\mathrm{SO}}
\newcommand{\U}{\mathrm{U}}
\newcommand{\su}{\lalg{su}}
\renewcommand{\sl}{\lalg{sl}}

\def\q{{\quad}}

\newcommand{\ketbra}[2] {
	| #1 \rangle \! \langle #2 |}
\def\Hom{{\rm Hom}}
\def\dim{{\rm dim}}

\def\I{{\mathbf 1}}
\def\A{{\sf A}}
\def\B{{\sf B}}
\def\C{{\sf C}}
\def\D{{\sf D}}
\def\spann{{\rm span}}
\def\phys{{\rm phys}}
\def\kin{{\rm kin}}
\def\sl2c{{\mathfrak{sl}_2 \Cl}} 
\def\su2{{\mathfrak{su}_2}}

\definecolor{darkgreen}{rgb}{0.0, 0.5, 0.13}

\usepackage{geometry}
\begin{document}

\title{Perspective-neutral approach to quantum frame covariance for general symmetry groups}

\author[1,2]{Anne-Catherine de la Hamette\thanks{\texttt{annecatherine.delahamette@univie.ac.at}}}
\author[3]{Thomas D. Galley\thanks{\texttt{tgalley1@perimeterinstitute.ca}}}
\author[4,5]{Philipp A.\ H\"ohn\thanks{\texttt{philipp.hoehn@oist.jp}}}
\author[6]{Leon Loveridge\thanks{\texttt{leon.d.loveridge@usn.no}}}
\author[1,2,3]{Markus P.\ M\"uller\thanks{\texttt{markus.mueller@oeaw.ac.at}}}
\affil[1]{\small Vienna Center for Quantum Science and Technology (VCQ), Faculty of Physics, University of Vienna, Boltzmanngasse 5, 1090 Vienna, Austria}
\affil[2]{\small Institute for Quantum Optics and Quantum Information, Austrian Academy of Sciences,\newline Boltzmanngasse 3, 1090 Vienna, Austria}
\affil[3]{\small Perimeter Institute for Theoretical Physics, 31 Caroline Street North, Waterloo ON N2L 2Y5, Canada}
\affil[4]{\small Okinawa Institute of Science and Technology Graduate University, Onna, Okinawa 904 0495, Japan}
\affil[5]{\small Department of Physics and Astronomy, University College London, London, United Kingdom}
\affil[6]{\small Quantum Technology Group, Department of Science and Industry Systems, University of South-Eastern Norway, 3616 Kongsberg, Norway}

\date{October 26, 2021}

\maketitle

\begin{abstract}
In the absence of external relata, internal quantum reference frames (QRFs) appear widely in the literature on quantum gravity, gauge theories and quantum foundations. Here, we extend the perspective-neutral approach to QRF covariance to general unimodular Lie groups. This is a framework that links internal QRF perspectives via a manifestly gauge-invariant Hilbert space in the form of ``quantum coordinate transformations'', and we clarify how it is a quantum extension of special covariance.
   We model the QRF orientations as coherent states which give rise to a covariant positive operator-valued measure, furnishing a consistent probability interpretation and encompassing non-ideal QRFs whose orientations are not perfectly distinguishable. We generalize the construction of relational observables, establish a variety of their algebraic properties and equip them with a transparent conditional probability interpretation. We import the distinction between gauge transformations and physical symmetries from gauge theories and identify the latter as QRF reorientations. 
The ``quantum coordinate maps'' into an internal QRF perspective are constructed via a conditioning on the QRF's orientation, generalizing the Page-Wootters formalism and a symmetry reduction procedure to general groups. We also find two types of QRF transformations: the gauge induced ``quantum coordinate transformations'' as passive changes of description, which we show are always unitary, and symmetry induced active changes of relational observables from one QRF to another. We then reveal novel physical effects: (i) QRFs whose orientations admit non-trivial isotropy groups can only resolve isotropy-group-invariant properties of other subsystems; (ii) when the QRF does \emph{not} admit symmetries, its internal perspective Hilbert space is not fixed and ``rotates'' through the kinematical subsystem Hilbert space as the QRF changes orientation. We also invoke the symmetries to extend a recent observation on the quantum relativity of subsystems to general groups, which explains the QRF dependence of entanglement and other physical properties. Finally, we compare with other approaches to QRF covariance and illustrate our findings in various examples. 
\end{abstract}

\newpage

\tableofcontents

\section{Introduction}

How does one describe a composite quantum system ``from the inside''? Such a description being tied to a choice of reference frame, this asks for the association of a reference frame with one of its subsystems: an \emph{internal} quantum reference frame. Since a composite system typically does not feature a distinguished choice of subsystem that may serve as a quantum reference frame, this leads to a multiple choice problem and a collection of quantum reference frames that may be in relative superposition. How, then, does one ``jump into the perspective'' of an internal subsystem to describe the remaining ones relative to it? In particular, given the multitude of quantum frame choices, how can one change internal perspective on the composite system by changing the quantum reference frame? And how does the description of physical properties of the composite system depend on the choice of internal quantum frame?

These questions are of relevance for composite quantum systems (incl.\ in field theory) in the presence of a symmetry principle whenever an idealized \emph{external} reference frame is either unavailable or possibly fictitious, and one thus has to resort to ``dynamical rods and clocks''. For example, questions of this kind arise in the foundations of quantum theory when invoking the idea that reference frames are in practice always associated with physical systems and should thereby ultimately be subject to the laws of quantum theory \cite{Aharonov:1967zza,Aharonov:1967zz,Aharonov:1984zz,Araki:1960zz,Yanase161,angeloPhysicsQuantumReference2011a,smithQuantumReferenceFrames2016,miyadera2016approximating,loveridge2017relativity,loveridgeSymmetryReferenceFrames2018a,loveridge2011measurement,busch2011position,Loveridge:2019phw,loveridge2020relational,Hoehn:2014vua}. They also arise in quantum gravity and cosmology \cite{DeWitt:1967yk,rovelliQuantumGravity2004,thiemannModernCanonicalQuantum2008,Tambornino:2011vg,Rovelli:1989jn,Rovelli:1990jm,Rovelli:1990ph,Rovelli:1990pi,dittrichPartialCompleteObservables2007,Dittrich:2005kc,Czech:2018kvg,Hardy:2018kbp,Hardy:2019cef,Donnelly:2015hta} and gauge theories \cite{Rovelli:2013fga,Rovelli:2020mpk,Donnelly:2015hta,Carrozza:2021sbk}, where the diffeomorphism/gauge symmetry forces one to resort to dynamical reference degrees of freedom to build gauge-invariant and thus physically meaningful quantities.\footnote{In a more operational context, quantum reference frames have also been extensively studied in 
quantum information theory \cite{Bartlett:2007zz,bartlett2009quantum,gour2008resource,gour2009measuring,Palmer:2013zza,smithCommunicatingSharedReference2019} and quantum thermodynamics \cite{aaberg2014catalytic,lostaglio2015description,lostaglio2015quantum,lostaglio2019coherence,marvian2019no,cwiklinski2015limitations,woods2019autonomous,woods2019resource}.  However, the problem of quantum frame covariance as studied here is concerned with a different class of foundational questions; see~\cite{Krumm:2020fws} for more discussion on the relation between these approaches.}

This has inspired a recent wave of efforts to establish a notion of quantum frame covariance of physical properties \cite{giacominiQuantumMechanicsCovariance2019,Bojowald:2010qw,Bojowald:2010xp,Hohn:2011us,Hoehn:2017gst,Vanrietvelde:2018pgb,Vanrietvelde:2018dit,hoehnHowSwitchRelational2018,Hoehn:2018whn,Hoehn:2019owq,Hoehn:2020epv,Hoehn:2021wet,Krumm:2020fws,Hoehn:2021flk,periodic,castro-ruizTimeReferenceFrames2019,Giacomini:2021gei,yang2020switching,giacominiRelativisticQuantumReference2019,streiter2021relativistic,hamette2020quantum,Ballesteros:2020lgl,Giacomini:2020ahk,Giacomini:2021aof,Barbado:2020snx,Mikusch:2021kro,Savi:2020qdl,Chataignier:2019kof,Chataignier:2020fys}.  In particular, this has led to the revelation of a quantum frame dependence of correlations and superpositions \cite{giacominiQuantumMechanicsCovariance2019,Vanrietvelde:2018pgb,hamette2020quantum}, the comparison of clock readings \cite{hoehnHowSwitchRelational2018,Hoehn:2020epv} (see also \cite{Bojowald:2010xp,Hohn:2011us}), temporal (non-)locality \cite{castro-ruizTimeReferenceFrames2019,Hoehn:2019owq}, certain resources \cite{Savi:2020qdl} and even singularity resolution in quantum cosmology \cite{Gielen:2020abd,Gielen:2021igw}. Much of this quantum frame dependence is explained by the observation that the very notion of subsystem is contingent on the choice of quantum frame, leading to the notion of a quantum relativity of subsystems \cite{Hoehn:2021wet}.\footnote{The dependence of singularity resolution on the choice of internal time in \cite{Gielen:2020abd,Gielen:2021igw} is not explained by a relativity of subsystems, but an internal-time-dependent choice of gauge-invariant Hilbert space.}

The focus of the present paper is the so-called \emph{perspective-neutral approach to quantum frame covariance} \cite{Hoehn:2017gst,Vanrietvelde:2018pgb,Vanrietvelde:2018dit,hoehnHowSwitchRelational2018,Hoehn:2018whn,Hoehn:2019owq,Hoehn:2020epv,Hoehn:2021wet,Krumm:2020fws,Hoehn:2021flk,periodic,castro-ruizTimeReferenceFrames2019,Giacomini:2021gei,yang2020switching} (see also \cite{Bojowald:2010qw,Bojowald:2010xp,Hohn:2011us} for its semiclassical precursor). What distinguishes this approach is that it not only works with the internal quantum reference frame perspectives, but also links them via a  perspective-neutral structure that encodes all these internal perspectives. The ensuing quantum reference frame changes thereby assume the form of ``quantum coordinate changes'', in close analogy to special relativity, as we elucidate shortly in Sec.~\ref{ssec_SR}. This is achieved by taking manifest gauge-invariance as the starting point, which renders the approach fully relational by construction and independent of external (or fictitious) reference structures. This makes it applicable to the foundations of quantum theory, gauge theories and gravity alike.

Our main achievement in this paper is to extend the perspective-neutral approach to quantum reference frames for arbitrary unimodular Lie groups (including finite ones). This approach had previously been developed only for Abelian groups \cite{Vanrietvelde:2018pgb,hoehnHowSwitchRelational2018,Hoehn:2018whn,Hoehn:2019owq,Hoehn:2020epv,Hoehn:2021wet,Krumm:2020fws,Hoehn:2021flk,periodic,castro-ruizTimeReferenceFrames2019,Giacomini:2021gei,yang2020switching} and the Euclidean group \cite{Vanrietvelde:2018dit}. Unimodular Lie groups are characterized by their Haar measure being both left- and right-invariant and encompass most finite-dimensional groups of physical interest, such as the Euclidean, Galilei or Poincar\'e spacetime groups, $\U(1)$ phase symmetry, temporal reparametrization invariance in finite dimensions, and most gauge groups. They do not, however, encompass the infinite-dimensional diffeomorphism group which, in its canonical incarnation, is not a proper Lie group \cite{Isham:1984sb,Isham:1984rz,thiemannModernCanonicalQuantum2008}; our results thus cannot be directly applied in full quantum gravity yet.

The only previous work on quantum frame covariance that encompasses general symmetry groups is \cite{hamette2020quantum}, which, however, uses a different formalism and does not discuss observables. While we recover the form of the quantum frame transformations for states derived in \cite{hamette2020quantum} in the form of quantum coordinate transformations, we shall see that the perspective-neutral approach extended here always yields unitary frame changes, while the approach of \cite{hamette2020quantum} only does so for ideal frames. The two approaches are thus not in general equivalent. Nevertheless, one can view the present work, in a certain sense, as connecting the framework of \cite{hamette2020quantum} with the perspective-neutral machinery.

Along the way, we establish a number of other results of which we provide a synopsis here.

\subsection{Summary of the results}

To formalize the perspective-neutral approach for general symmetry groups, we invoke group averaging procedures in the form of Refined Algebraic Quantization \cite{Giulini:1998kf,Giulini:1998rk,thiemannModernCanonicalQuantum2008}, which is intimately intertwined with Dirac's constraint quantization program of gauge systems \cite{diracLecturesQuantumMechanics1964,Henneaux:1992ig}, to construct the gauge-invariant (and perspective-neutral) states and observables. In contrast to most previous work on the perspective-neutral approach, we shall here directly work at the level of group representations, rather than at the algebraic level of their generators, which are the constraints of the theory.

\begin{itemize}
    \item \emph{We model the orientation states of quantum reference frames for general symmetry groups as coherent states that constitute a covariant positive operator-valued measure} (POVM) \cite{holevoProbabilisticStatisticalAspects1982,buschOperationalQuantumPhysics}, extending previous constructions \cite{loveridge2017relativity,loveridgeSymmetryReferenceFrames2018a,Loveridge:2019phw,Hoehn:2019owq,Hoehn:2020epv,periodic,Smith:2017pwx,Smith:2019imm}. In particular, this equips the frame orientation states with a consistent probabilistic interpretation and encompasses what we call \emph{ideal, complete} and \emph{incomplete reference frames}, where the latter admit non-trivial isotropy groups of the frame orientation states and all frames that are not ideal are characterized by orientation states that are not perfectly distinguishable.
    \item Relational observables play a key role in quantum gravity/cosmology and gauge systems more generally \cite{Rovelli:1989jn,Rovelli:1990jm,Rovelli:1990ph,Rovelli:1990pi,rovelliQuantumGravity2004,dittrichPartialCompleteObservables2007,Dittrich:2005kc,Dittrich:2006ee,Dittrich:2007jx,thiemannModernCanonicalQuantum2008,Tambornino:2011vg,Carrozza:2021sbk}, however, their quantization is notoriously difficult. Here, we extend the quantum construction procedure of relational observables developed in \cite{Hoehn:2019owq,Hoehn:2020epv,Hoehn:2021wet,Krumm:2020fws,Hoehn:2021flk,periodic,Chataignier:2019kof,Chataignier:2020fys} to general unimodular groups. This also generalizes the construction of \cite{Bartlett:2007zz} to non-ideal frames, non-compact groups and arbitrary orientations and of \cite{loveridgeSymmetryReferenceFrames2018a,miyadera2016approximating,loveridge2017relativity} to arbitrary frame orientations. \emph{Relational observables take the form of frame-orientation-conditional gauge transformations of the operator intended to be described relative to the frame}; in the language of gauge theories, they constitute a form of frame-dressed observables. This takes the form of a particular incoherent group averaging procedure (the so-called $G$-twirl). We establish a variety of their algebraic properties and, invoking the covariant frame POVM and generalizing the Page-Wootters formalism (see below), also \emph{equip these relational observables with a clear conditional probability interpretation}. This extends the construction of \cite{Hoehn:2019owq,Hoehn:2020epv} (see also \cite{Chataignier:2020fys}) to general groups. 
    \item Each relational observable family, describing some other subsystem relative to the frame for all its possible orientations, defines an orbit in the algebra. The dimension of this orbit depends on the choice of non-invariant subsystem observable and the isotropy group of the reference frame.
    We note that \emph{relational observables relative to a fixed choice of quantum reference frame stratify the algebra of gauge-invariant observables} and different frame choices thus lead to different stratifications.
    \item We draw connections with gauge theories and import their distinction between (I) \emph{gauge transformations} and (II) \emph{symmetries}  \cite{Donnelly:2016auv,Geiller:2019bti,Carrozza:2021sbk} into the setting of quantum reference frames. This distinction is particularly important in the context of so-called edge modes, which are degrees of freedom that appear on finite boundaries and account for the fact that such boundaries break \emph{a priori} gauge invariance \cite{Donnelly:2016auv,Donnelly:2011hn,Donnelly:2014fua,Donnelly:2014gva,Geiller:2019bti,Carrozza:2021sbk,Freidel:2020xyx,Wieland:2021vef,Wieland:2017cmf,Riello:2021lfl}. Gauge transformations and symmetries commute as they act on ``opposite sides'' of these edge modes.
    Recently, edge modes have been identified in classical gauge theories as dynamical reference frames for the local gauge group in the same sense as the quantum reference frames here \cite{Carrozza:2021sbk}; specifically, upon quantization, edge modes are quantum field theory versions of quantum reference frames. They are external reference frames for a subregion delimited by the finite boundary and describe how that subregion relates to its complement. Importantly, symmetries (II) have been identified in \cite{Carrozza:2021sbk} as edge mode frame reorientations. \\
    Here, we show that one can similarly introduce symmetries as quantum frame reorientations that act from the right on the frame orientation states if the gauge transformations act from the left (or vice versa). They are important for relational observables, which are gauge- but not symmetry-invariant: as we show, the \emph{frame reorientations constitute precisely the transformations that translate along the relational observable orbits} in the algebra. We also exploit the notion of symmetries for arriving at other insights (see below).
    \item Quantum reference frames giving rise to both gauge transformations and symmetries require systems of coherent states that admit a group action from both the left and right. For compact groups, we fully characterize such frames.
    \item We generalize the ``trinity of relational quantum dynamics'', an equivalence of three \emph{a priori} different formulations of relational dynamics \cite{Hoehn:2019owq,Hoehn:2020epv,periodic} (see also \cite{Hoehn:2021flk}), to quantum reference frames for general unimodular Lie groups. This involves generalizing the Page-Wootters formalism \cite{pageEvolutionEvolutionDynamics1983,pageClockTimeEntropy1994,giovannettiQuantumTime2015,Smith:2017pwx,Smith:2019imm,Hoehn:2019owq,Hoehn:2020epv,Hoehn:2021wet,castro-ruizTimeReferenceFrames2019} and a quantum symmetry reduction procedure \cite{Vanrietvelde:2018pgb,Vanrietvelde:2018dit,hoehnHowSwitchRelational2018,Hoehn:2018whn,Hoehn:2019owq,Hoehn:2020epv,Hoehn:2021wet,Giacomini:2021gei}, both of which define ``quantum coordinate maps'' from the perspective-neutral Hilbert space into the internal frame perspectives, to general symmetry groups. \emph{These quantum coordinate maps amount to a particular gauge-fixing of the frame orientation and provide the answer to the opening question of ``how to jump into the perspective of an internal subsystem?''}. The result yields a relational ``Schr\"odinger'' and ``Heisenberg picture'' relative to the quantum frame, respectively. We put these names in quotation marks as they refer to whether states or observables transform unitarily under reorientations of the quantum frame, which only in the case of a temporal frame corresponds to time evolution. Since the quantum coordinate maps are invertible, these Schr\"odinger and Heisenberg descriptions relative to the internal frame perspectives are unitarily equivalent to the perspective-neutral picture of relational observables and gauge-invariant states, establishing a ``trinity of relational quantum physics''.
    \item When a non-trivial isotropy group arises for the frame orientations, we demonstrate that \emph{the quantum frame can only resolve those properties of the remaining subsystems that are also invariant under its isotropy group}. This holds both at the level of relational observables and in the internal frame perspective. Through gauge invariance, the quantum frame thereby enforces properties on the remaining subsystems that they may not have originally featured. This can be interpreted as a certain form of coarse-graining of the kinematical subsystem information.
    \item We reveal a striking novel physical effect for non-Abelian theories: \emph{when the quantum reference frame does \emph{not} admit symmetries (II), the reduced Hilbert space constituting the ``quantum coordinate description'' of the remaining subsystems is not fixed, but actually depends on the reference frame orientation}. As the frame changes orientation, the internal perspective Hilbert space ``rotates'' through the kinematical subsystem Hilbert space. The same applies to the reduced observable algebras, but the perspective-neutral Hilbert space and algebra remain fixed. This extends the notion of quantum subsystem relativity of \cite{Hoehn:2021wet}: not only is a gauge-invariant notion of subsystem frame-dependent, but it is even frame-orientation-dependent when the frame admits no symmetries.
    \item Since we have two types of group actions, (I) gauge transformations and (II) symmetries as frame reorientations, we also find \emph{two distinct types of quantum reference frame transformations}. Using the above mentioned ``quantum coordinate maps'', one can construct the ``quantum coordinate transformations'' that take one from one internal quantum frame perspective into another one. These amount to \emph{frame-relation-conditional gauge transformations} and, since the quantum coordinate maps are invertible, are always unitary, regardless of whether the frame is ideal or even incomplete. This has to do with the gauge-induced degeneracy in the description of perspective-neutral (physical) states in terms of kinematical data. This result differs from the quantum frame changes in \cite{hamette2020quantum}, which are only unitary for ideal frames. These quantum coordinate changes amount to \emph{passive changes of description} of the same invariant states and observables. All previously reported quantum reference frame transformations are of such passive form.\\
    We then also show how \emph{frame-relation-conditional symmetries} give rise to a new active type of quantum reference frame transformations. These are transformations at the perspective-neutral level that transform from the relational observables relative to one frame to those relative to another. That is, they map \emph{between} the different stratifications of the perspective-neutral observable algebra in terms of relational observables defined by the two frames. The super-operator implementing this transformation is a particular incoherent group averaging ($G$-twirl) over the symmetry (rather than gauge) group action.
    \item It has been shown in \cite{Hoehn:2021wet} that a gauge-invariant notion of subsystems is contingent on the choice of quantum frame. We \emph{extend this quantum relativity of subsystems to general symmetry groups} and arrive at this result in an arguably simpler manner than in \cite{Hoehn:2021wet}, invoking the newly introduced notion of symmetries (II) and specifically the quantum reference frame transformations in the form of frame-relation-conditional symmetries. Together with \cite{Hoehn:2021wet}, this provides a transparent explanation for the above mentioned quantum frame dependence of physical properties such as correlations and superposition.
    \item We elucidate why the perspective-neutral approach to quantum frame covariance, further developed here, is not in general equivalent to the purely perspective-dependent approaches \cite{giacominiQuantumMechanicsCovariance2019,giacominiRelativisticQuantumReference2019,streiter2021relativistic,hamette2020quantum,Ballesteros:2020lgl,Barbado:2020snx,Mikusch:2021kro}. The latter start with a \emph{fixed} subsystem Hilbert space in a given internal quantum frame perspective that is not in general compatible with gauge-invariance under the symmetry group. In contrast to here, this leads in general situations to non-unitary quantum reference frame transformations \cite{hamette2020quantum}.
    
\end{itemize}

Given the many tools and concepts that this paper borrows from gauge theories and the fact that it encompasses general (unimodular) gauge groups, it lays out a path towards the application of quantum reference frames in quantum gauge field theories, especially together with the classical observation that dynamical reference frames are incarnated, for example, as edge modes in the field theory context \cite{Carrozza:2021sbk}. Since edge modes only appear on the boundary of a finite subregion and not everywhere in its bulk, this would also constitute a somewhat more controlled step into the field theory context, where one does not have to quantize the frame field ``everywhere''.

Altogether, the perspective-neutral approach, together with its subsystem relativity, can also be seen as a manifestation of the view on gauge theories proposed in \cite{Rovelli:2013fga,Rovelli:2020mpk} (see also \cite{Gomes:2019xhu,Gomesgauge}): more than mathematical redundancy, the presence of gauge symmetry encodes the fact that physical quantities in our universe are relational. Gauge-dependent quantities, such as the frame orientations, while not predictable, admit a physical interpretation as the relata that compose to form gauge-invariant and predictable relational quantities; they provide the ``handles'' through which subsystems can couple in different ways. Indeed, in the perspective-neutral approach, we exploit the gauge induced redundancy to establish the linking structure of all internal (hence relative) perspectives in the first place and describe all the subsystems relative to one another in various orientations while maintaining manifest gauge-invariance. 
In particular, also the gauge-dependent frame orientations are instrumental in formulating physically meaningful and gauge-invariant questions, such as ``what is the value of this subsystem observable, given that the frame is in orientation such and such?'', as encoded in relational observables. In fact, the gauge induced redundancy is \emph{the} key to why the perspective-neutral approach works and why it yields unitary quantum reference frame transformations.

\subsection{Organization of the paper}
To motivate the general framework of the perspective-neutral approach, we begin by drawing an analogy with special relativity and its concept of covariance in Sec.~\ref{ssec_SR}. In Sec.~\ref{sec_QRF}, we explain what we mean by quantum reference frames for general symmetry groups, distinguishing between complete, ideal and incomplete frames. We show how to model their orientation states in terms of generalized coherent states, which also gives rise to a covariant positive operator-valued measure. This in turn equips the frame orientations with a consistent probabilistic interpretation. Next, in Sec.~\ref{sec_ph}, we establish the perspective-neutral Hilbert space, which is the physical Hilbert space of gauge-invariant states, before discussing gauge-invariant observables in Sec.~\ref{sec_diracobs} and specifically constructing relational observables in Sec.~\ref{ssec_relobs} via frame-orientation-conditional gauge transformations of the quantity of interest, similar to the special relativity example in Sec.~\ref{ssec_SR} below. We establish their key algebraic properties and show how a reference frame subject to a non-trivial frame orientation isotropy group can only resolve properties of the remaining degrees of freedom that are invariant under this isotropy group too. We also demonstrate a novel striking physical effect present only in non-Abelian theories: when the quantum frame does not admit a symmetry group action (in contrast to the special relativity example above), the physical system Hilbert space is not just dependent on the choice of reference frame, but also on the orientation state of that reference frame.

In Sec.~\ref{sec:phys_cond}, we show that the physical inner product on the physical Hilbert space can be re-written as a conditional inner product, a fact which will prove useful in extending the Page-Wootters formalism to general groups. In Sec.~\ref{ssec_LR} we consider the family of quantum reference frames which admit of both a symmetry and a gauge group $G$. We provide a full classification of all such quantum reference frames for compact $G$. In Sec.~\ref{ssec_symgauge} we further discuss the difference between gauge and symmetry, showing that for reference systems which admit both gauge and symmetry groups, the symmetry group is a generator of orbits of relational observables in the space of physical observables.

In Sec.~\ref{sec:gen_PW}, we extend the Page-Wootters formalism to general groups, providing a `relational Schr{\"o}dinger  picture' for general groups, where the states of the system (but not the observables) depend on the frame orientation. We continue this line of reasoning in Sec.~\ref{ssec_condprob} by providing a conditional probability interpretation for physical states. This allows us to show that relational observables between a frame and system ask the precise question: ``what is the probability that a certain quantity on the system takes a given value when the frame is in a given orientation?''. In Sec.~\ref{sec_Heisenberg} we adopt a complementary approach, invoking a symmetry reduction procedure that results in a unitarily equivalent `relational Heisenberg picture', where the observables of a system (rather than the states) are frame orientation dependent. The unitary equivalence of these two \emph{a priori} distinct reductions was previously only known for Abelian groups.

In Sec.~\ref{sec: QRF transformations as quantum coordinate changes}, we study the case of two (or more) possible quantum reference frames, deriving in Sec.~\ref{sec:QRF_change_S} changes of quantum reference frame for general groups as ``quantum coordinate transformations'' in the `relational Schr{\" o}dinger picture'. This recovers as a special case the known change of quantum reference frame for ideal frames. In Sec.~\ref{sec:QRF_change_H} we derive the corresponding change of quantum reference frame in the  `relational Heisenberg picture'. In Sec.~\ref{sec_relobschanges}, we explore transformations between relational observables relative to one frame and relational observables relative to a second frame. Here only, we restrict our attention to the regular representation and show that they are in this case related by physical symmetries  which are conditional on the relative orientations of the two frames. In Sec.~\ref{sec:subsys_rel} we extend the known relativity of subsystems to ideal reference frames for general groups and give a novel perspective on it in terms of symmetries. In Sec.~\ref{sec_examples}, we provide a number of examples of non-ideal quantum reference frames. We begin in Sec.~\ref{sec: U(1) example} with an example for the group $\U(1)$, showing how to explicitly determine the physical Hilbert space as well as the physical system Hilbert spaces, which provide an example for quantum subsystem relativity.  In Sec.~\ref{sec:j_1_coherent_states}, we introduce a system of coherent states with an  $\SU(2)$  gauge symmetry but not an $\SU(2)$ physical symmetry. This allows us in Sec.~\ref{sec: SU(2) example with 3 particles} and \ref{sec: SU(2) example with 4 particles} to give examples which, for the first time, explicitly exhibit reference frame orientation dependence of the physical system (internal perspective) Hilbert spaces.
In Sec.~\ref{sec:perspective_neutral_vs_dependent}, we compare the perspective-neutral approach of the present work to the purely perspective-dependent approach to quantum reference frames. Finally, we conclude in Sec.~\ref{sec:conclusion}, providing suggestions for future research directions. Certain technical details have been moved to various appendices.

\section{Analogy with special covariance: perspective-neutral structure in special relativity}\label{ssec_SR}

To understand the intuition behind the perspective-neutral approach, it is worthwhile to draw an analogy with special relativity and recall the concept of special covariance.\footnote{The same analogy can be drawn with general covariance, which, however, invokes the infinite-dimensional diffeomorphism group and thus necessitates field theory. Since the present paper focuses on quantum reference frames in quantum mechanics, rather than quantum field theory, the analogy with inertial observers and the finite-dimensional Lorentz group is here more appropriate for our purposes.} Colloquially, it asserts that all the laws of physics should be the same in every inertial reference frame. This leads to the laws of physics being formulated in terms of tensor equations. There are two (equivalent) ways in which one can think of tensors as (inertial frame) perspective-neutral structure that encodes and links all the internal frame perspectives. Let us now illustrate this in some detail, as this will help to understand the perspective-neutral approach in quantum theory later, where we will encounter entirely analogous structures and relations between them as here in special relativity. We begin with the way more familiar to physicists, where we identify only Lorentz-invariant quantities, i.e.\ Lorentz scalars, as inertial frame perspective-neutral objects that however can be expressed in different inertial coordinate systems (internal frame perspectives).\\~

\noindent\textbf{(A)} Suppose we are given a Lorentz observer $O$, whose orthonormal frame is specified by $e^\mu_A$, where $\mu=0,\ldots,3$ is the spacetime index and $A=0,\ldots,3$ labels the vectors setting up the tetrad with $e_0^\mu$ the timelike vector tangent to $O$'s worldline. Carrying two types of indices, the frame $e_A^\mu$ also carries \emph{two} types of representations of the Lorentz group: 
\begin{itemize}
    \item[(I)] $\Lambda^\mu{}_\nu$, acting on the spacetime index, is a passive transformation that acts on every spacetime tensor, incl.\ $O$'s frame vectors $e^\mu_A$, and changes the \emph{description} of it in terms of some inertial coordinates $x^\mu$; suggestively for later, let us call these ``gauge transformations''.
    \item[(II)] $\Lambda^A{}_B$, acting on the frame index, is an active transformation that acts only on $O$'s frame $e^\mu_A$, changing its \emph{orientation} (e.g.\ boosting or rotating it); suggestively for later, let us call these ``symmetries''.
\end{itemize}
It is clear that the two Lorentz group actions (I) and (II) \emph{commute} on the frame.
In fact, there is even a third representation of the Lorentz group since the reference frame $e_A^\mu$ is itself group-valued: $\Lambda^\mu{}_A:=e^\mu_A$ is also a Lorentz transformation, since $\eta_{AB}=\Lambda^\mu{}_A\Lambda^\nu{}_B\eta_{\mu\nu}$ is thanks to the frame orthonormality also of the form $\rm{diag}(-1,1,1,1)$. Hence, through (I) and (II) the Lorentz group acts regularly on itself. In the context of this article, such frames are referred to as \emph{ideal} reference frames and we shall see its quantum analog later (where we also consider non-ideal frames).

The energy-momentum tensor measured by $O$ is $T_{AB}=e^\mu_A T_{\mu\nu} e^{\nu}_B$, where $T_{\mu\nu}$ is the spacetime energy-momentum tensor; e.g.\ energy- and momentum-flux measured by $O$ are $E=e_0^\mu T_{\mu\nu}e_0^\nu$ and $p_i=e_0^\mu T_{\mu\nu}e^\nu_i$, $i=1,2,3$, respectively.
Now every component of $T_{AB}$ is invariant under Lorentz (gauge) transformations $\Lambda^\mu{}_\nu$ (I); as a spacetime scalar, every component $T_{AB}$ is a tensor and  perspective-neutral. Indeed, all inertial observers will agree that $O$ will measure the energy $E$ in her/his frame, and we have not yet invoked any coordinate system $x^\mu$ constituting the \emph{description} of the spacetime scalars $T_{AB}$ from the internal perspective of some inertial frame. $T_{AB}$ thus defines a set of observables encoding in a (gauge-)invariant\,---\,and hence internal frame perspective-neutral\,---\,manner the energy-momentum properties relative to $O$'s frame. Suggestively for later, let us call such contracted objects ``relational observables''. For later, we also note that, since $e^\mu_A$ is both the frame orientation \emph{and} a Lorentz group element, the relational observable $T_{AB}$ can be viewed as a frame-orientation-conditional Lorentz transformation of the non-invariant $T_{\mu\nu}$, or a frame dressed observable. (In the quantum theory, we will see that relational observables are quantum-frame-orientation-conditional gauge transformations too.)
These relational observables are, however, not invariant under symmetries (II): $T_{AB}$ transforms as a rank-two tensor under reorientations $\Lambda^A{}_B$ of $O$'s frame. This changes the physical situation; for example, a boosted $O$ will measure a distinct energy-flux than the unboosted $O$. 

Let us now consider the questions of reference frame changes and frame covariance. To this end, we introduce a second inertial observer $O'$, specified by the frame ${e'}_{A'}^\mu$. Since we have two types of Lorentz transformations, (I) gauge transformations and (II) symmetries, we will also have two types of frame transformations. We will begin with the active frame transformations, which proceed at the perspective-neutral level and change relational observables from one frame to another; these are induced by the symmetries (II). We restrict attention to the energy-momentum relational observables. The transformation from $T_{AB}$, the energy-momentum relational observables associated with the first observer $O$, to ${T'}_{A'B'}={e'}_{A'}^\mu T_{\mu\nu}{e'}_{B'}^\nu$, the corresponding relational observables associated with the second observer $O'$ reads
\ba 
T_{AB}=\Lambda_A{}^{A'}T'_{A'B'}\Lambda^{B'}{}_B=\Lambda_A{}^{A'}{e'}_{A'}^\mu T_{\mu\nu}{e'}_{B'}^\nu\Lambda^{B'}{}_B\,,\label{lorentz}
\ea 
where 
\ba 
\Lambda^{A'}{}_A:=e^\mu_Ae'^{A'}_\mu\,.\label{lorentz2}
\ea 
This is a \emph{relation-conditional reorientation} of $O'$'s frame $e'^\mu_{A'}$: $\Lambda^{A'}{}_A$ acts only on $O'$'s frame on the right hand side and it depends on the relative orientation of the two frames. Indeed, the Lorentz symmetries $\Lambda^{A'}{}_A$ are the relational observables describing $O$'s frame components relative to $O'$'s frame.  Notice that so far no choice of coordinate system has been invoked and everything is independent of inertial coordinates. Eq.~\eqref{lorentz} is an active change from one set of relational observables to another at the inertial frame perspective-neutral level.

Next, we explore reference frame transformations induced by the gauge transformations (I). These are passive transformations and will only change the \emph{coordinate descriptions} of relational observables, corresponding to internal frame perspectives. First, we note that ``jumping into the internal frame perspective'' of an inertial observer means choosing a coordinate system that amounts to a gauge-fixing of the Lorentz group action (I). The relational observables $T_{AB}$ can then be expressed in any such coordinate system and, while all inertial observers will agree on the numbers $T_{AB}$, say $E=e^\mu_0 T_{\mu\nu}e^\nu_0$, their \emph{description} in terms of spacetime tensor components $e^\mu_A,T_{\mu\nu}$ will be different. For example, the inertial coordinate system corresponding to $O$ is the gauge-fixing such that $e^\mu_A=\delta^\mu_A$. In this coordinate system, the invariant relational observable $T_{AB}=e^\mu_A T_{\mu\nu} e^\nu_B$ looks particularly simple: it coincides with the form of the non-invariant spacetime tensor $T_{\mu\nu}$. Let $\Lambda^{\mu'}{}_\mu$ denote the Lorentz (gauge-)transformation between $O'$'s coordinate system, in which now $e'^{\mu'}_{A'}=\delta^{\mu'}{}_{A'}$, and $O$'s original coordinate system. Then we can write the \emph{same} relational observables in $O'$'s perspective as 
\ba 
T_{AB}=e^\mu_A T_{\mu\nu} e^\nu_B=e^{\mu}_A\Lambda_\mu{}^{\mu'}T_{\mu'\nu'}\Lambda^{\nu'}{}_{\nu}e^\nu_B\,,
\ea 
where $T_{\mu'\nu'}$ is the energy-momentum tensor in $O'$'s coordinates. In particular,  $\Lambda^{\mu'}{}_\mu$ is generally not the identity matrix and so $T_{AB}$ will have a less simple form in $O'$'s coordinates, appearing as a non-trivial Lorentz transformation of the spacetime tensor $T_{\mu'\nu'}$.
Since $T_{\mu'\nu'}$ coincides as matrix with $T'_{A'B'}$, just like $T_{\mu\nu}$ coincides with $T_{AB}$, it is clear that as matrix $\Lambda^{\mu'}{}_\mu$ coincides with $\Lambda^{A'}{}_{A}$ given in Eq.~\eqref{lorentz2}. However, in contrast to the latter, $\Lambda^{\mu'}{}_\mu$ acts on the spacetime index and is thus a \emph{frame-relation-conditional} gauge transformation. It is the coordinate transformation from $O'$'s internal perspective into the one of $O$.

In conclusion, in special relativity, \emph{passive/active changes of inertial frame appear as relation-conditional Lorentz gauge/symmetry transformations}.\\~

\noindent\textbf{(B)} Finally, let us now discuss the alternative, somewhat more abstract way in which one can think about tensors as perspective-neutral structure and special covariance. Abstractly, tensors are multilinear maps from copies of the tangent (and cotangent) space into the real numbers. As such, they encode the physics as experienced in any local spacetime reference frame at once: an abstract tensor has to be contracted with a vector frame, in order to determine the numbers that a corresponding observer would find in a measurement of the quantities embodied by the tensor. In this sense, tensors abstractly constitute a description of the local physics \emph{before} a choice of reference frame has been made; they are reference frame perspective-neutral structures. 

Let us return to the example of the energy-momentum tensor $T$ and tetrad $e_A$, corresponding to $O$'s inertial frame. We invoke Penrose's abstract index notation in which $T_{\mu\nu}$ does now \emph{not} denote the component of $T$ in some arbitrary coordinate system, but merely the fact that it has to be contracted with two tangent vectors to produce a number, i.e.\ that (at event $x$ in Minkowski space) it is a map $T:T_x\mathbb{R}^{1,3}\times T_x\mathbb{R}^{1,3}\rightarrow\mathbb{R}$. Similarly, $e^\mu_A$ is not the component of $e_A$ in some coordinate system, but highlights that this object is a vector. As such, the ``index contraction'' $T_{AB}=e_A^\mu T_{\mu\nu}e^\nu_B$ means now nothing else than $T_{AB}=T(e_A,e_B)\in\mathbb{R}$. However, as before, these are the numbers that $O$ would measure when probing the energy-momentum tensor $T$. Accordingly, we can view the tetrad vectors $e^\mu_A$ as the coordinate maps from the abstract, perspective-neutral tensor description into the internal perspective of $O$. The inverse coordinate maps, mapping $O$'s internal perspective back into the perspective-neutral level are the inverse tetrads $e_\mu^A$, satisfying $e_\mu^Ae^\mu_B=\delta^A{}_B$ and so $T_{\mu\nu}=e_\mu^AT_{AB}e^B_\nu$, again understood in abstract notation.

Similarly, we can view a second observer $O'$'s tetrad $e'^\mu_{A'}$ as the coordinate map from the perspective-neutral description into $O'$'s internal perspective. It is then clear that the reference frame transformation between $O$'s and $O'$'s internal perspectives is again a \emph{relation-conditional} Lorentz transformation $\Lambda^{A'}{}_A=e_A^\mu e_\mu^{A'}$ of the form appearing in Eqs.~\eqref{lorentz} and~\eqref{lorentz2}, except that we now interpret all ingredients in terms of abstract index notation. Note that this Lorentz transformation indeed has the compositional form of a coordinate transformation, linking the internal perspectives of $O$ and $O'$ \emph{via} the perspective-neutral structure (assuming the linking role of the manifold in a coordinate transformation).
Indeed, since 
\ba 
T_{AB}=\Lambda_A{}^{A'}T_{A'B'}\Lambda^{B'}{}_B=e_A^\mu e_\mu^{A'}T_{A'B'}e^{B'}_\nu e^\nu_B= e_A^\mu T_{\mu\nu} e^\nu_B \,,
\ea 
the transformation first maps $T_{A'B'}$ back into the abstract, perspective-neutral object $T_{\mu\nu}$, then mapping it into $O$'s internal perspective using the corresponding coordinate map. This is a more abstract, but also arguably more immediate incarnation of special covariance.\\~

We will see completely analogous structures, maps and transformations as above in the quantum theory (including the fact that a relational observable appears simpler in the internal perspective of its associated frame, than in other internal frame perspectives). The quantum frame covariance established by the perspective-neutral approach can thus be viewed as a direct quantum analog of special covariance. (A field theory extension would establish a quantum version of general covariance.)

\section{What are quantum reference frames for general symmetry groups?}\label{sec_QRF}

Suppose we are given some (gauge) symmetry group $G$ that acts on a system of interest. 
In a nutshell, a quantum reference frame for $G$, or $G$-frame, is a subsystem which ideally is ``as non-invariant as possible'' under the action of $G$. The idea is to use its configurations as a \emph{dynamical coordinate system} to parametrize (or gauge fix) the orbits of $G$ in Hilbert space or the observable algebra; clearly, this requires the quantum frame degrees of freedom to be non-invariant in some manner. The goal is then to describe in a (gauge-) invariant way how the remaining degrees of freedom transform under $G$ relative to this dynamical frame. That is, the non-invariant reference degrees of freedom serve to construct invariant observables out of the remaining degrees of freedom.

Let us now be more precise. For definiteness, we shall henceforth assume that $G$ is an arbitrary \emph{unimodular Lie group}, i.e.\ a Lie group whose left- and right-invariant Haar measures coincide. Our construction will encompass both compact and non-compact groups. The reason for restricting to unimodular groups will become clear in Sec.~\ref{sec_ph} below; essentially, non-unimodular groups will complicate the construction of the physical Hilbert space of gauge-invariant states. However, most groups of physical interest are unimodular; in particular, all Abelian, compact~\cite{Freitag} or semi-simple groups are unimodular~\cite[Sec.\ IV]{Glasner}, as are the Lorentz~\cite{Schroeck} and Galilei groups~\cite{Cassinelli}. Unimodularity is thus not a severe restriction for most physical purposes.
While much of our construction below would also  work for unimodular topological groups, we shall focus on Lie groups as their generators have physical significance as observables of the theory. For example, in gauge theories these generators constitute the constraints or charges of the theory, depending on whether $G$ acts as a gauge or symmetry group.

Despite the fact that the generators often have physical significance, our construction does not rely on any assumptions of connectedness of the group. In particular, it also applies to \emph{finite groups} which are zero-dimensional Lie groups. The counting measure is both a left and right invariant Haar measure making the group unimodular. Hence by suitable replacement of $\int_G dg$ by $\sum_{g \in G}$ below, one can apply the results of this paper to discrete groups (for quantum reference frames for finite groups, see also \cite{hamette2020quantum,Krumm:2020fws,Hoehn:2021flk}).

In much of the sequel, we shall consider a quantum reference frame $R$ together with a system $S$ of interest to be described relative to $R$. The total kinematical Hilbert space, which can also be interpreted as the Hilbert space of $R$ and $S$ described relative to some (possibly fictitious) \emph{external} frame, then takes the product form $\ch_{\rm kin}=\ch_R\otimes\ch_S$. We shall assume that $\ch_{\rm kin}$ carries a unitary tensor product representation of the unimodular Lie group $G$, i.e.\ for all $g\in G$ and $\ket{\psi_{\rm kin}}$ we have $U_{RS}(g)\,\ket{\psi_{\rm kin}}=U_R(g)\otimes U_S(g)\,\ket{\psi_{\rm kin}}$, where $g\mapsto U_R(g)$ and $g\mapsto U_S(g)$ are strongly continuous unitary representations. In gauge theory terminology, the globally acting representation $U_{RS}$ constitutes the action  of the gauge group; later, we will also encounter representations of the same group $G$ that only act on the frame and which are called symmetries. 
We shall not restrict the unitary representation $U_S$ on $\ch_S$ further (e.g., it may be reducible). Even though symmetries in quantum theory are generally implemented via \emph{projective representations}, this setup does not restrict generality: every projective representation of $G$ comes from a unitary representation of a central extension $G'$ of $G$~\cite{Schottenloher}. Thus, we may restrict our attention to unitary representations without loss of generality.

On the reference frame Hilbert space $\ch_R$, on the other hand, we shall impose a few further assumptions. Firstly, while the representation $U_R$ of $G$ need not be irreducible, we shall require the existence of a system of coherent states, denoted $\{U_R,\ket{\phi(g})_R\}$, such that $G$ acts transitively on the coherent state system. That is, for all $g,g'\in G$ we have $\ket{\phi(g'g)}_R= U_R(g')\ket{\phi(g)}_R$. In particular, we have $\ket{\phi(g)}_R=U_R(g)\ket{\phi(e)}$, i.e.\ the coherent state system is the orbit of $\ket{\phi(e)}$ under the representation, where $e$ denotes the identity in $G$. Note that this defines an equivariant map $\tilde E:G\rightarrow\ch_R$, $\tilde E(g)=\ket{\phi(g)}_R$, as it satisfies $\tilde E(gg')=U_{R}(g)\,\ket{\phi(g')}_R$. The ``seed'' state $\ket{\phi(e)}_R$ of the coherent state system is some fixed vector in $\ch_R$ such that the equivariance condition and transitivity hold true (we will impose further conditions below); specifically, it is convention which of the states in a coherent state set one calls $\ket{\phi(e)}_R$. This definition of generalized coherent states largely follows the one given in \cite{Perelomov}.

We shall interpret the coherent states $\ket{\phi(g)}_R$ as the different possible orientation states of the $G$-frame and thus often call $g$ the quantum frame orientation.
We emphasize that the coherent states will generally not be orthogonal, that is $\braket{\phi(g)|\phi(g')}\nsim\delta(g,g')$, where $\delta(g,g')$ is the delta distribution on the unimodular $G$ with respect to its (left and right) Haar measure $dg$. The frame orientations will thus generally not be perfectly distinguishable.

Nevertheless, the frame orientations can be used to completely parametrize the $G$ orbits in Hilbert space or the observable algebra, provided $G$ acts not only transitively on the coherent state system, but also freely, i.e.\ the isotropy group for each frame orientation is trivial. In this case, $U_R$ defines a regular action on the coherent state system, so that the latter constitutes a principle homogeneous space for $G$ on which the equivariant map $\tilde E$ above is invertible. Frames $R$ with such properties have been called \emph{complete reference frames} in \cite{Bartlett:2007zz}.

Note that a regular action of $G$ does \emph{not} mean a regular representation, which is the special case when $\ch_R\simeq L^2(G,dg)$,\footnote{The left regular representation is defined in $L^2(G)$ by $(U_R(g)f)(h) = f(g^{-1}h)$, with the right regular representation $(U_R(g)f)(h) = f(hg)$. We will often refer to the carrier space $L^2(G)$ as the regular representation, with the defining action left implicit.} i.e.,\ when the group provides the configuration manifold on which $R$ is quantized. In this case, the group $G$ acts regularly on itself and we have orthogonal, i.e.\ perfectly distinguishable frame orientation states $\braket{g|g'}=\delta(g,g')$. For this reason, this special class of complete frames is  called \emph{ideal reference frames}. 

By contrast, when $G$ acts transitively, but not freely on the coherent state system, i.e.\ when frame orientation states admit a non-trivial isotropy group, the equivariant map $\tilde E$ is not invertible on the coherent state system and the frame orientations cannot provide a complete parametrization of $G$ orbits. Such frames $R$ will henceforth be called \emph{incomplete reference frames}, as they cannot resolve the orbits. Our construction below encompasses complete and incomplete quantum reference frames. One could also define \emph{overcomplete reference frames} as those on whose orientation states $G$ acts freely, but not transitively, however, we shall not consider them further in this manuscript. 

In order for the frame orientation states to admit a consistent probabilistic interpretation, we assume that the representation $U_R(g)$ and the seed coherent state $\ket{\phi(e)}$ are chosen such that they give rise to a resolution of the identity
\ba
\int_G dg\,\ket{\phi(g)}\!\bra{\phi(g)}_R=n\,\mathbf{1}_R \,,\label{resid}
\ea
where $n>0$ is a constant. This will also permit us to expand arbitrary states of $R$ in the coherent states. Since the Haar measure $dg$ is only defined up to a multiplicative constant, we can (and will)  absorb the normalization factor into $dg$ such that $n=1$. In contrast to a different convention that is often found in the literature, here we will thus in general have $\int_G dg\neq 1$ for compact groups $G$.

Owing to this normalization, the frame orientation states give rise to a \emph{covariant positive operator-valued measure} (POVM)~\cite{holevoProbabilisticStatisticalAspects1982,buschOperationalQuantumPhysics} furnishing the probabilistic interpretation. Indeed, for every $Y\subset\cb(G)$, where $\cb(G)$ denotes the Borel sets on the group $G$, we can define the positive effect operators 
\ba 
E_Y:=\int _Y dg\,\ket{\phi(g)}\!\bra{\phi(g)}_R\geq0 \, ,\label{covPOVM}
\ea 
which are normalized, i.e.\ $E_G=\I_R$, and satisfy the covariance property~\cite{holevoProbabilisticStatisticalAspects1982,buschOperationalQuantumPhysics}
\ba 
E_{h.Y}=U_R(h) E_Y U_R(h)^\dagger\,.
\ea 
To see the latter property,
write
 $U_R(h) E_Y U_R(h)^\dagger = \int_Y dg\,\ket{\phi(hg)}\!\bra{\phi(hg)}_R = \int_G dg \, \chi_Y (g) \ket{\phi(hg)}\!\bra{\phi(hg)}_R$, where $\chi_Y$ is the characteristic function of the subset $Y$. Then set $hg = g'$ and use the left invariance of $dg$ to write the last integral as $\int_G dg'\, \chi_Y(h^{-1}g') \ket{\phi(g') }\!\bra{\phi(g')}_R = \int_{hY} dg'\,\ket{\phi(g')}\!\bra{\phi(g')}_R = E_{h.Y}$. Covariant POVMs have been previously used to model quantum reference frames for the translation group or $\rm{U}(1)$, e.g.\ to model quantum clocks~\cite{Brunetti:2009eq,Smith:2017pwx,Smith:2019imm,Hoehn:2019owq,Hoehn:2020epv,Hoehn:2021wet,periodic}; here we generalize this to quantum frames for general unimodular groups. They were also introduced in, e.g., \cite{loveridgeSymmetryReferenceFrames2018a}, working for general locally compact groups in order to ``relativise'' non-invariant self-adjoint operators/POVMs.

Assumption~\eqref{resid} imposes non-trivial constraints on the representation and the coherent states. For the special case of compact groups $G$, these constraints can be described as follows. The proof is given in Appendix~\ref{SecProofEx1}.
\begin{example}
\label{Ex1Compact}
Suppose that $G$ is compact and that $\ch$ is finite-dimensional. Then we can decompose the representation and Hilbert space as follows
\be
   U_R(g)=\bigoplus_q U_R^{(q)}(g)\otimes\mathbf{1}^{(q)}\,,\quad
   \mathcal{H}=\bigoplus_q \mathcal{M}^{(q)}\otimes \mathcal{N}^{(q)}\,,
\ee
where the $U_R^{(q)}$ are mutually inequivalent irreducible representations.
Then Assumption~\eqref{resid} is satisfied if and only if
\begin{equation}
   \dim\,\mathcal{M}^{(q)}\geq\dim\,\mathcal{N}^{(q)}\quad\mbox{for all }q\,,
   \label{eqNecessaryId}
\end{equation}
and if the seed coherent state $\ket{\phi(e)}$ has the form
\begin{equation}
   \ket{\phi(e)}=\frac 1{\sqrt{\dim\,\ch}}\cdot\sum_q \sqrt{\dim\,\mathcal{M}^{(q)}}\sum_{j=1}^{\dim\,\mathcal{N}^{(q)}} \ket{j}_{\mathcal{M}}\otimes \ket{j}_{\mathcal{N}}\,,
   \label{eqNecessaryId2}
\end{equation}
where $\{\ket{j}_{\mathcal{M}}\}$ is an arbitrary orthonormal system of $\mathcal{M}^{(q)}$ (similarly $\{\ket{j}_{\mathcal{N}}\}$).
\end{example}
In the special case where the representation is irreducible, there is only a single $q$ and $\dim\,\mathcal{N}^{(q)}=1$. Thus, we always obtain a resolution of the identity, and every state can be chosen as $\ket{\phi(e)}$.

Assumption~\eqref{resid} also imposes certain conditions on the isotropy subgroup of non-compact groups $G$:
\begin{lemma}
\label{lemCompact}
 Convergence of the integral in~\eqref{resid} implies that the \textbf{isotropy subgroup}~\cite{Perelomov} of the coherent state system must be \textbf{compact}. This group $H$ is defined as
\be
   H=\{g\in G\,\,|\,\,|\phi(g)\rangle\langle\phi(g)|=|\phi(e)\rangle\langle\phi(e)|\}\,.
\ee
Consider the quotient space $X:=G/H$. We will see that the group averaging integral can be split into two integrals: one over $H$, and another one over $X$. To this end, note that every coset $x\in X$ can be written $x=g_x H$ for some suitable $g_x\in G$ (in fact, every $g_x\in x$ will work). Every such coset has the property that $\ket{\phi(g)}\!\bra{\phi(g)}=\ket{\phi(g')}\!\bra{\phi(g')}$ for all $g,g'\in x$, hence we can define $\ket{\phi(x)}:=\ket{\phi(g_x)}$, and $\ket{\phi(x)}\bra{\phi(x)}$ does not depend on the choice of $g_x$. Since we always assume that $G$ is unimodular, the existence of a left-invariant measure on the coset space\footnote{While the coset space $X$ is generally not a group, by a left-invariant measure on $X$ we mean a measure $d_Xx$ such that $d_Xgx=d_Xx$ for all $g\in G$.} $X:=G/H$ is equivalent to the unimodularity of the subgroup $H$~\cite{Nachbin}. Thus, let us \textbf{once and for all assume that $H$ is unimodular}. Then we can informally write
\be
\int_G dg\ket{\phi(g)}\!\bra{\phi(g)}=\int_X\left(\int_H \ket{\phi(g_x h)}\!\bra{\phi(g_x h)}d_H h\right)d_X x={\rm Vol}(H)\int_X \ket{\phi(x)}\!\bra{\phi(x)} d_X x\,,
\ee
where $d_H$ and $d_X$ are corresponding suitably normalized invariant measures on $H$ and on $X$, respectively, and ${\rm Vol}(H):=\int_H d_H h$
(this manipulation is indeed rigorously correct for the integration of continuous functions with compact support; here we assume that we can do the same for integrals over operators). For this to be finite, $H$ must have finite volume, i.e.\ be compact. If so, our assumption~\eqref{resid} becomes (via $dx:=d_X x$)
\be
   \int_X dx\, \ket{\phi(x)}\!\bra{\phi(x)}=n\cdot\mathbf{1}_R\,,
\ee
and we will choose the invariant measure $dx$ on $X$ such that the normalization constant is $n=1$.
\end{lemma}
Compactness of the isotropy group $H$ excludes some important cases, like certain \emph{periodic clocks}, where $G=(\mathbb{R},+)$ is the group of time translations, and $H=(\mathbb{Z},+)$ \cite{periodic}. On the other hand, it includes interesting cases of coherent state systems with non-trivial isotropy group such as the Weyl-Heisenberg group, as we will see in Example~\ref{ExWeylHeisenberg}.

In the general case, our coherent states $\ket{\phi(g)}$ will not be vectors in the Hilbert space $\ch_R$, but they will be distributions in the sense of rigged Hilbert spaces (for the general theory, see \cite{gel2016generalized,schwartz1950theorie}). This is the case in the following example (where $\ket{x}$ is an antilinear continuous functional on the Schwartz space $\mathcal{S} \subset L^2(\mathbb{R})$). We elaborate more on this example in Appendix~\ref{ss:gr}.
\begin{example}\label{ex_2}
Let $G=(\mathbb{R},+)$ be the translation group in one dimension, acting on $\ch_R=L^2(\mathbb{R})$, via $T_a\psi(x)=\psi(x-a)$. Then a system of coherent states is given by the position eigenstates $\ket{x}$ such that $\ket{\psi}=\int_{\mathbb{R}} \psi(x)\ket{x}\,dx$: we have the resolution of the identity $\mathbf{1}_R=\int_{\mathbb{R}} \ket{x}\!\bra{x}\, dx$.
\end{example}

We noted above that the coherent states modelling the frame orientation states are generally not perfectly distinguishable, in contrast to Example~\ref{ex_2}. A famous case in which this happens is described in the following example.

\begin{example}[The Weyl-Heisenberg group, part 1 of 2]
\label{ExWeylHeisenberg}
A single particle in one dimension is described by a wave function $\psi\in L^2(\mathbb{R})$, with suitably defined position and momentum operators
\be
   (\hat x \psi)(x)=x\cdot\psi(x)\,,\qquad (\hat P \psi)(x)=-i\frac { d} {{ d} x} \psi(x)\,.
\ee
These operators also generate translations of position or momentum: $\exp(-i x_0 \hat P)\psi(x)=\psi(x-x_0)$, and $\exp(i p_0 \hat x)\hat\psi(p)=\hat \psi(p-p_0)$, where $\hat \psi$ is the Fourier transform of $\psi$. However, translations of position and momentum do not commute. Specifically, for $\alpha\in\mathbb{C}$, define $x,p\in\mathbb{R}$ via $\alpha=(x+ip)/\sqrt{2}$, and then define the \emph{displacement operators}~\cite{Perelomov}
\be
   D(\alpha)=\exp(-ix\hat p - i p \hat x) \,.
\ee
These operators satisfy
\be
   D(\alpha)D(\beta)=e^{i {\rm Im}(\alpha\bar\beta)} D(\alpha+\beta) \,,
\ee
which motivates the definition of a group --- the Weyl-Heisenberg group $W_1$. Its group elements $g(t;\alpha)$ are characterized by a real number $t\in\mathbb{R}$ and a complex number $\alpha\in\mathbb{C}$. Writing $\alpha=\alpha_1+i\alpha_2$ with $\alpha_j\in\mathbb{R}$, the multiplication rule is
\be
   (s;\alpha)\cdot(t;\beta)=(s+t-\alpha_1\beta_2+\alpha_2\beta_1;\alpha+\beta)\,,
\ee
which arises from the identification $(s;\alpha)\leftrightarrow e^{is}D(\alpha)$. In this sense, the Weyl-Heisenberg group describes the interplay of position and momentum translations for a single Bosonic mode.\\
$\strut\qquad$The unitary representations of $W_1$ are characterized by a real number $\lambda$ such that \cite{Perelomov}
\be
   T^\lambda(s;0)=\exp(i\lambda s)\mathbf{\hat 1}\,.
\ee
For $\lambda=0$, there are many different irreducible representations, and they are all one-dimensional. These are irrelevant for our purpose. For every fixed $\lambda\neq 0$ (to which we focus our attention in the following), any two irreducible representations of $W_1$ are unitarily equivalent. One can directly check that, up to unitary equivalence,
\begin{eqnarray}
T^\lambda(0;\alpha)&=&D(\sqrt{\lambda}\,\alpha)\qquad (\lambda>0)\,,\nn\\
T^\lambda(0;\alpha)&=&D(\sqrt{|\lambda|}\,\bar\alpha)\qquad(\lambda<0)\,,
\end{eqnarray}
and necessarily $T^\lambda(s;\alpha)=T^\lambda(s;0)T^\lambda(0;\alpha)$. Now, since the Lebesgue measure on $\mathbb{R}\times\mathbb{C}\simeq\mathbb{R}^3$ is both left and right invariant, $W_1$ is a unimodular group. But suppose that we would like to define a coherent state system on our reference system $R$, as described in Lemma~\ref{lemCompact}. Then the isotropy group necessarily contains the subgroup $\{(t;0)\,\,|\,\, t\in\mathbb{R}\}\subset W_1$, which is not compact. Thus, at first sight, it seems as if the construction of this paper cannot be applied to $W_1$ due to the insight of Lemma~\ref{lemCompact}.\\
$\strut\qquad$However, this problem can be circumvented by considering, for instance, the representation $W'_1:=T^1(W_1)$ as a group in its own right. It is a group with elements of the form $(\Theta;\alpha)$, with $\Theta\in\mathbb{C}$ such that $|\Theta|=1$, and $\alpha\in\mathbb{C}$ arbitrary. The multiplication rule is
\be
   (\Theta;\alpha)\cdot(\Theta';\beta)=(\Theta\Theta' \exp(i(-\alpha_1\beta_2+\alpha_2\beta_1))\,;\alpha+\beta)\,.
\ee
The isotropy subgroup of $W'_1$ is now isomorphic to ${\rm U}(1)$, i.e.\ compact, and the methods of this paper can be applied to it. Specifically, we have $G=W'_1$, $H=\{(\Theta;0)\}$, and $X\equiv G/H=\mathbb{C}$.

It is well-known~\cite{Perelomov} that the following construction gives a coherent state system. Define $\ket{\alpha}:=D(\alpha)\ket{0}$, where $\ket{0}$ is the vacuum vector. Then
\be
   \mathbf{1}=\int_{\mathbb{C}}d\alpha\, \ket{\alpha}\!\bra{\alpha}\,,
\ee
where $d\alpha$ is the suitably normalized Lebesgue measure on $\mathbb{C}$. Thus, our approach covers the paradigmatic case of coherent states for the harmonic oscillator as a special case.

We will continue our discussion of the Weyl-Heisenberg group in Example~\ref{ExWH2} below.
\end{example}

\section{Perspective-neutral Hilbert space and relational observables}\label{sec_3}

In this section, we will  develop the necessary Hilbert space and algebraic structures underlying the perspective-neutral approach to quantum frame covariance. Exploiting these structures, we shall then explain how to `jump into an internal quantum reference frame perspective' and how to change quantum reference frame in the subsequent sections, before illustrating our findings in a few examples.

\subsection{Physical Hilbert space and inner product}\label{sec_ph}

We begin with the construction of the so-called physical Hilbert space of gauge-invariant states, i.e.\ the Hilbert space invariant under the gauge group action $U_{RS}$. This is the arena of the perspective-neutral approach \cite{Vanrietvelde:2018pgb,Vanrietvelde:2018dit,hoehnHowSwitchRelational2018,Hoehn:2018whn,Hoehn:2019owq,Hoehn:2020epv,castro-ruizTimeReferenceFrames2019,Giacomini:2021gei,Hoehn:2021wet,Hoehn:2021flk}. For unimodular groups $G$, the idea is essentially to decompose its representation $U_{RS}$ on $\ch_{\rm kin}$ into a direct sum (when $G$ is compact) or direct integral (when $G$ is non-compact) of irreducible representations and to then select the component carrying the trivial representation as a set out of which to construct the physical Hilbert space $\ch_{\rm phys}$. When $G$ is non-compact and one is dealing with a direct integral decomposition of $\ch_{\rm kin}$, states in the trivial representation will be distributions so that $\ch_{\rm phys}$ will \emph{not} be a linear subspace of $\ch_{\rm kin}$. For non-unimodular groups this strategy for constructing $\ch_{\rm phys}$ does not quite work and one has to proceed somewhat differently \cite{Giulini:1998kf}; specifically, one cannot find states invariant under the unitary action of $U_{RS}$ and instead has to construct a non-unitary representation from it which does admit invariant states. While interesting for gravity \cite{Giulini:1998kf,thiemannModernCanonicalQuantum2008}, incorporating this case would render our exposition less accessible, which is why we focus on unimodular groups as these encompass most cases of physical interest.

Before we delve into this, let us provide some intuition. Why should one proceed to the trivial representation and $\ch_{\rm phys}$? When one is dealing with a gauge system this is the space solving the constraints of the theory and the only physically meaningful one \cite{diracLecturesQuantumMechanics1964,Henneaux:1992ig}. More generally, interpreting the kinematical Hilbert space $\ch_{\rm kin}$ as the space of states distinguishable relative to an external frame (which is fictitious for gauge systems), the aim is to construct a Hilbert space of states that are distinguishable relative to an \emph{internal} and dynamical reference frame and without access to an external reference structure. This space can thus only contain states that carry purely relational information and this is the physical Hilbert space. As we shall see shortly, it is obtained via a \emph{coherent} group averaging procedure which removes external reference information. While there exist other external-frame-independent structures that are commonly used in quantum information theory \cite{Bartlett:2007zz}, e.g.\ \emph{incoherently} group-averaged states (which are necessarily mixed) or the multiplicity spaces $\cn^{(q)}$ in Example~\ref{Ex1Compact}, it was shown in \cite{Krumm:2020fws,Hoehn:2021flk} that the additional requirement of an internal description of $S$ relative to $R$ requires the physical Hilbert space. In \cite{Hoehn:2021flk} it was further demonstrated for finite groups that $\ch_{\rm phys}$ admits the information-theoretic characterization of being the maximal subspace of states that can be purified in an external-frame-independent manner. In contrast to incoherently group-averaged states, $\ch_{\rm phys}$ also admits an unambiguous \emph{invariant} partial trace\,---\,the relational trace \cite{Krumm:2020fws}.\footnote{The relational trace can be used to resolve \cite{Krumm:2020fws} the so-called `paradox of the third particle' \cite{angeloPhysicsQuantumReference2011a}.}

Key for our construction is that the (gauge) symmetry defined by $U_{RS}$ induces a \emph{redundancy in the description of the physical Hilbert space} $\ch_{\rm phys}$ in terms of the kinematical data on $\ch_{\rm kin}$ (physical states are equivalence classes of kinematical states). There are thus many different ways in which to describe the same invariant physical state and these correspond essentially to different gauges. The idea underlying the perspective-neutral approach \cite{Vanrietvelde:2018pgb,Vanrietvelde:2018dit,hoehnHowSwitchRelational2018,Hoehn:2018whn,Hoehn:2019owq,Hoehn:2020epv,castro-ruizTimeReferenceFrames2019,Giacomini:2021gei,Hoehn:2021wet,Hoehn:2021flk} is to associate these different descriptions with different internal quantum reference frame perspectives on the same physical situation (in our present setup, the system $S$ could contain other subsystems that may serve as quantum frames). In this sense, the physical Hilbert space encodes and links all internal frame perspectives, which is why it is also called the \emph{perspective-neutral Hilbert space}. This is the quantum analog of the tensorial structure in special relativity, which, as discussed in Sec.~\ref{ssec_SR}, is perspective-neutral in the same sense. We shall elaborate on this in Sec.~\ref{sec: QRF transformations as quantum coordinate changes} when discussing quantum reference frame changes. As we shall see, the redundancy will also be key for these transformations to be unitary even when dealing with non-ideal reference frames, i.e.\ when the frame orientation states are not perfectly distinguishable. In particular, this does not hold when applying the quantum frame transformations to kinematical states \cite{hamette2020quantum}. 

Let us now outline how to construct the physical Hilbert space. We will invoke the method of Refined Algebraic Quantization (RAQ) in the form of group averaging which has been developed for general Lie groups in some detail \cite{Giulini:1998kf,Giulini:1998rk,thiemannModernCanonicalQuantum2008} and applies to our case. Since the purpose of this manuscript is to explain what to do with $\ch_{\rm phys}$ in the context of quantum reference frames, rather than how precisely to construct it (and this is done elsewhere in the literature), we shall be somewhat informal in our exposition in the main body, employing a physicist's rather than more precise functional analytic notation. We shall just assume that the following group averaging procedures exist and refer to the above literature for a more rigorous account for when this is the case. We shall also illustrate the following construction in some examples and discuss, in Appendix~\ref{app_groupaverage}, what it might take in practice for a rigorous formulation. This construction is deeply intertwined with the program of constraint quantization of gauge systems going back to Dirac \cite{diracLecturesQuantumMechanics1964,Henneaux:1992ig}.

In a nutshell, we are looking for states $\ket{\psi_{\rm phys}}$ such that 
\ba 
U_{RS}(g)\,\ket{\psi_{\rm phys}}=\ket{\psi_{\rm phys}}\,,\q\q\q\forall\,g\in G\,.\label{physicalstatedef}
\ea
If the $U_{RS}(g)$ do not have a joint eigenvector of eigenvalue $\lambda=1$ (in particular, if at least one $U_{RS}(g)$ does not contain $\lambda=1$ in its discrete spectrum), then $\ket{\psi_{\rm phys}}$ is an improper eigenstate and will not be contained in $\ch_{\rm kin}$. Inspired by the theory of rigged Hilbert spaces, RAQ \cite{Giulini:1998kf,Giulini:1998rk,thiemannModernCanonicalQuantum2008} then proceeds to choose some dense linear subspace $\Phi\subset\ch_{\rm kin}$ which is left invariant by the unitary action of $U_{RS}$. (If the $U_{RS}(g)$ have a joint eigenvector with eigenvalue $\lambda=1$, as is typically the case with compact $G$, one may also have $\Phi=\ch_{\rm kin}$.) One then seeks solutions to Eq.~\eqref{physicalstatedef} in the algebraic dual $\Phi^*$ of $\Phi$, i.e.\ the set of \emph{all} (not just continuous) complex linear functionals on $\Phi$. Indeed, $\Phi^*$ admits a dual action of $U_{RS}$ defined by $\left(U_{RS}(g)\xi\right)[\psi]=\xi\big[U_{RS}(g^{-1})\psi\big]$ for $\xi\in\Phi^*$ and arbitrary $\psi\in\Phi$. In this case,  Eq.~\eqref{physicalstatedef} should be understood as $U_{RS}(g)\psi_{\rm phys}=\psi_{\rm phys}$ with $\psi_{\rm phys}\in\Phi^*$. Since $\Phi^*$ contains $\ch_{\rm kin}$, one has the triple  $\Phi\hookrightarrow\ch_{\rm kin}\hookrightarrow\Phi^*$ and solution to Eq.~\eqref{physicalstatedef} should be thought of as distributions on $\Phi$. For example, for the translation group $U(x)=e^{ix\hat p}$ on $L^2(\mathbb{R})$, the solution to Eq.~\eqref{physicalstatedef} is the zero-momentum eigenstate $\ket{p=0}$ which is a distribution.

RAQ then aims for an (anti-linear) `rigging map' $\eta:\Phi\rightarrow\Phi^*$ such that:
\begin{itemize}
    \item[(i)] The image of $\eta$ solves Eq.~\eqref{physicalstatedef}, i.e.\ $\eta(\phi)[U_{RS}(g)\psi]=\eta(\phi)[\psi]$, for all $\phi,\psi\in\Phi$ and $g\in G$.
    \item[(ii)] It gives rise to a well-defined inner product~(\ref{etaPIP}) on its image $\rm{Im}(\eta)\subset\Phi^*$. This requires reality $\eta(\phi)[\psi]=\overline{\eta(\psi)[\phi]}$, where the overline denotes complex conjugation, and positivity $\eta[\phi](\phi)\geq0$ for all $\phi,\psi\in\Phi$.
    \item[(iii)] $\eta$ commutes with invariant observables $O$ (that commute with $U_{RS}$ on $\Phi$) in the sense that $O(\eta(\phi))=\eta(O\phi)$, $\phi\in\Phi$.
\end{itemize}
$\ch_{\rm phys}$ is then the Cauchy completion of the image of $\eta$ under the physical inner product given by the rigging map itself (cf.\ (ii))
\ba 
\braket{\eta(\phi)|\eta(\psi)}_{\rm phys}:=\eta(\phi)[\psi]\,.\label{etaPIP}
\ea

How does one construct the rigging map $\eta$? Let us now argue informally and adopt a physicist's notation. One way to proceed is by defining the \emph{coherent group averaging} `projector' 
\ba
\Pi_{\rm phys}:=\int_G dg\,U_{RS}(g)\,\label{proj}
\ea
and setting \cite{Giulini:1998kf}
\ba 
\eta\ket{\psi_{\rm kin}}:=\bra{\psi_{\rm kin}}\Pi_{\rm phys}\,.\label{etagroup}
\ea 
If $G$ is unimodular,  we can informally argue that
\ba 
   \Pi_{\rm phys}^\dagger = \int_G dg\, U_{RS}(g^{-1})=\int_G dg\, U_{RS}(g)=\Pi_{\rm phys}\label{piphyssym},
\ea
thanks to $dg=d(g^{-1})$, where $\dag$ denotes the adjoint with respect to $\ch_{\rm kin}$. This enforces the reality condition for the resulting map $\eta$:
\[
   \eta(\phi)[\psi]=\langle\psi|\Pi_{\rm phys}|\psi\rangle=\overline{\langle\psi|\Pi_{\rm phys}^\dagger|\phi\rangle}=\overline{\eta(\psi)[\phi]}.
\]
For non-unimodular groups this informal calculation fails, i.e.\ $\Pi_{\rm phys}^\dag\neq\Pi_{\rm phys}$,\footnote{If $G$ is non-unimodular, we have $\Pi_{\rm phys}^\dagger = \int_G dg\, U_{RS}(g^{-1})=\int_G d g\, U_{RS}(g) \Delta_G (g^{-1})$, where $\Delta_G$ is the modular function that is only identically one on all of $G$ for unimodular groups.} which is related to the above comment that in this case one should look for physical states that are invariant under a \emph{non-unitary} representation of the group \cite{Giulini:1998kf}. While one can incorporate non-unimodular groups consistently in RAQ \cite{Giulini:1998kf}, this is the reason why henceforth we restrict to unimodular groups $G$, the symmetry of $\Pi_{\rm phys}$ simplifying the ensuing discussion.

For simplicity, we shall from now on slightly abuse notation,  writing
\ba
\ket{\psi_{\rm phys}}:=\Pi_{\rm phys}\,\ket{\psi_{\rm kin}}\,,
\ea
rather than the more precise Eq.~\eqref{etagroup}, minding that physical states may be distributions. Owing to the symmetry of $\Pi_{\rm phys}$, this will not cause problems. In the sequel, we shall just assume that this group averaging procedure converges and yields a non-trivial image, and discuss in Appendix~\ref{app_groupaverage} in more detail what this means in practice for $G=\mathbb{R}$.
Informally, $\Pi_{\rm phys}$ maps kinematical states to physical states because the image transforms trivially under $G$:
\ba
U_{RS}(g)\,\ket{\psi_{\rm phys}} = \int_G dg'\,U_{RS}(g g')\,\ket{\psi_{\rm kin}}=\int_G d(g g')\,U_{RS}(g g')\,\ket{\psi_{\rm kin}} = \ket{\psi_{\rm phys}}\,.\label{stateinv}
\ea
We have written `projector' in quotation marks above because $\Pi_{\rm phys}$ cannot literally be an idempotent linear map, unless $\Phi=\ch_{\rm kin}$ as in the case of compact groups: it is a map $\Pi_{\rm phys}:\Phi\to\Phi^*$, and so the expression $\Pi_{\rm phys}^2$ does not make sense if $\Phi\subsetneq\mathcal{H}_{\rm kin}$. Let us illustrate this with some informal calculations. We can view $\Pi_{\rm phys}$ as an improper projector when $G$ is non-compact. This means that its square does not satisfy the projector property $\Pi_{\rm phys}^2\neq\Pi_{\rm phys}$, but leads to a divergence because of integrating a constant over a non-compact gauge orbit:
\ba
\Pi_{\rm phys}^2\,\ket{\psi_{\rm kin}} = \int_G dg\,\ket{\psi_{\rm phys}} = \rm{Vol}(G)\,\ket{\psi_{\rm phys}}\,,\label{improj}
\ea
where $\rm{Vol}(G)=\int_Gdg$ is the volume of $G$. Owing to the last equation, we should normalize $\Pi_{\rm phys}$ in Eq.~\eqref{proj} by a pre-factor $1/\sqrt{\rm{Vol}(G)}$ in the case that $G$ is compact in order for $\Pi_{\rm phys}$ to satisfy the projector property. However, since we would like to cover both the compact and non-compact case at once, we refrain from doing so. 

We will instead ensure normalization of the resulting states $\ket{\psi_{\rm phys}}$ by introducing the physical inner product on the image of $\Pi_{\rm phys}$ along the lines of Eq.~\eqref{etaPIP}. Indeed, Eq.~\eqref{improj} is also the reason why physical states are \emph{not} elements of $\ch_{\rm kin}$ when $G$ is non-compact. Indeed, physical states are then not normalizable in the kinematical inner product because
\ba
\braket{\psi_{\rm phys}|\phi_{\rm phys}}_{\rm kin} = \braket{\psi_{\rm kin}|\,\Pi_{\rm phys}^2\,|\phi_{\rm kin}}_{\rm kin} = \rm{Vol}(G)\,\braket{\psi_{\rm kin}|\,\Pi_{\rm phys}\,|\phi_{\rm kin}}_{\rm kin}\,,
\ea
where $\braket{\cdot|\cdot}_{\rm kin}$ denotes the inner product on $\ch_{\rm kin}$. This is the primary reason why we need to introduce a new inner product on physical states. Note that here we have made use of unimodularity of $G$ which implies that $\Pi_{\rm phys}^\dag=\Pi_{\rm phys}$ is symmetric.

The new inner product on physical states, defined by Eq.~\eqref{etaPIP}, is now given by simply dropping the $\rm{Vol}(G)$ factor:
\ba
\braket{\psi_{\rm phys}|\phi_{\rm phys}}_{\rm phys}:=\braket{\psi_{\rm kin}|\,\Pi_{\rm phys}\,|\phi_{\rm kin}}_{\rm kin}\,,\label{PIP}
\ea
where $\ket{\phi_{\rm kin}}\in\ch_{\rm kin}$ is any representative of the equivalence class of kinematical states which maps under group averaging to the same physical state, and likewise for $\bra{\psi_{\rm kin}}$. Since $\Pi_{\rm phys}$ is symmetric for unimodular groups, the inner product satisfies the reality condition and the right-hand side only depends on the equivalence classes, i.e.\ on the physical states, and \emph{not} on the kinematical representatives of the equivalence classes. Indeed, suppose $\ket{\psi_{\rm kin}}\neq\ket{\psi'_{\rm kin}}$, but $\psi_{\rm kin}\sim\psi'_{\rm kin}$ in the sense that $\Pi_{\rm phys}\,\ket{\psi_{\rm kin}} = \Pi_{\rm phys}\,\ket{\psi'_{\rm kin}}$. Then 
\ba
\braket{\psi_{\rm kin}|\,\Pi_{\rm phys}\,|\phi_{\rm phys}}_{\rm kin} = \braket{\Pi_{\rm phys}\,\psi_{\rm kin}\,|\phi_{\rm kin}}_{\rm kin} = \braket{\Pi_{\rm phys}\,\psi'_{\rm kin}\,|\phi_{\rm kin}}_{\rm kin}=\braket{\psi'_{\rm kin}|\,\Pi_{\rm phys}\,|\phi_{\rm kin}}_{\rm kin}\,.
\ea
Positivity, on the other hand, i.e.\ why $\langle\psi_{\rm phys}|\psi_{\rm phys}\rangle_{\rm phys}\geq 0$ should hold, is a different question. It is currently not known whether group-averaging always yields a positive rigging map $\eta$. However, the uniqueness result in \cite{Giulini:1998kf} establishes that, when group averaging converges sufficiently quickly and a positive rigging map exists, the former defines the unique rigging map and physical inner product. This means that when group averaging does not give rise to a positive semi-definite inner product, a positive rigging map cannot exist and RAQ fails. To the best of our knowledge, examples of this have not yet been constructed in the literature, and our examples below will also yield a positive semi-definite physical inner product. In particular, note that in the special case of the regular representation of $G$ on $\ch_{\rm kin}$ (which in particular implies an ideal reference frame $R$ in our case), one can show that group averaging \emph{always} defines a positive rigging map \cite{Giulini:1998kf}.

The physical Hilbert space is now defined as the Cauchy-completion of the image of $\Pi_{\rm phys}$ with the physical inner product defined in Eq.~\eqref{PIP}. (This might require quotienting out zero-norm states.) We will discuss the observable properties (iii) of the group-averaging rigging map in the follow subsections.

Let us now consider examples. We will continue with the Weyl-Heisenberg group of Example~\ref{ExWeylHeisenberg}. For the translation group, see Appendix~\ref{app_groupaverage} and also \cite{Vanrietvelde:2018pgb}.

\begin{example}[The Weyl-Heisenberg group, part 2 of 2]
\label{ExWH2}
Recall Example~\ref{ExWeylHeisenberg}. Let us choose the reference system $R$ to carry the defining representation $T^1(W_1)$ of $W'_1$, and consider the different options that we have to define a representation on the remaining system $S$. If $S$ is simply a copy of $R$, carrying the same representation of $W'_1$ as $R$, then elements $\ket{\psi_{\rm phys}}$ of the physical Hilbert space would have to satisfy the invariance property $T^1(s;\alpha)\otimes T^1(s;\alpha)\ket{\psi_{\rm phys}}=\ket{\psi_{\rm phys}}$. In the special case $\alpha=0$, this would mean that $e^{2is}\ket{\psi_{\rm phys}}=\ket{\psi_{\rm phys}}$ which is impossible. Fortunately, there is a better choice: it turns out that $T^{-1}(W_1)$ is a representation of $W'_1$. Since $T^1(s;\alpha)\otimes T^{-1}(s;\alpha)=D(\alpha)\otimes D(\bar\alpha)$, the phase is annihilated and this problem disappears. Thus, let us choose the representation $T^{-1}(W_1)$ of $W'_1$ for $S$.\\
$\strut\qquad$Let us now determine the physical Hilbert space. A tedious but straightforward calculation shows that there are no $\ket{\psi_{\rm kin}}\in\ch_{\rm kin}$ with $D(\alpha)\otimes D(\bar\alpha)=\ket{\psi_{\rm kin}}$ for all $\alpha$, hence the elements of $\ch_{\rm phys}$ will all lie outside of $\ch_{\rm kin}$. Informally, up to an irrelevant constant factor, the projection onto the physical Hilbert space is given by
\be
   \Pi_{\rm phys}=\int_{\mathbb{C}}d\alpha\, D(\alpha)\otimes D(\bar \alpha)\,.
\ee
Setting again $\alpha=(x+ip)/\sqrt{2}=\alpha_1+i\alpha_2$ and using the representation in \cite{Potocek2015}, the action of $D(\alpha)$ can be expressed as
\begin{eqnarray}
D(\alpha):L^2(\mathbb{R})&\to& L^2(\mathbb{R}) \,,\nn\\
\psi(x)&\mapsto& e^{i\sqrt{2}\alpha_2(x-\alpha_1/\sqrt{2})}\psi(x-\sqrt{2}
\alpha_1)\,.
\end{eqnarray}
Then, defining $\ket{\psi_{\rm phys}}:=\Pi_{\rm phys}\ket{\psi_{\rm kin}}$, where $\ket{\psi_{\rm kin}}=\int_{\mathbb{R}} dx_1\int_{\mathbb{R}} dx_2 \psi(x_1,x_2)\ket{x_1}\otimes\ket{x_2}$, we get
\be
   \ket{\psi_{\rm phys}}=\pi\int_{\mathbb{R}} ds \int_{\mathbb{R}} dx \, \psi(x-s,x-s)\ket{x}\ket{x}
   =\pi\left(\int_{\mathbb{R}}ds\,\psi(s,s)\right)\int_{\mathbb{R}}dx\,\ket{x}\ket{x}\,.
\ee
This calculation shows that the physical Hilbert space $\ch_{\rm phys}$ is one-dimensional; it consists of all complex multiples of $\ket{f}:=\int_{\mathbb{R}}dx\,\ket{x}\ket{x}$. Certainly, $\ket{f}$ is not an element of the kinematical Hilbert space $\ch_{\rm kin}$; instead, it can be understood as a linear functional on a suitably defined subspace of $\mathcal{H}_{\rm kin}^*$, acting as
\be
   \langle\varphi|f\rangle=\int_{\mathbb{R}} dx\,\overline{\varphi(x,x)}\,.
\ee
Furthermore, even in this distributional sense, $\ket{\psi_{\rm phys}}$ cannot be well-defined for all $\ket{\psi_{\rm kin}}\in\ch_{\rm kin}$: instead, the domain of definition of $\Pi_{\rm phys}$ must be a set of $\ket{\psi_{\rm kin}}\in\ch_{\rm kin}$ for which $s\mapsto \psi(s,s)$ makes sense as a function (for example, for continuous $\psi$) and the integral $\int_{\mathbb{R}} ds\, \psi(s,s)$ converges. Such sets of states would be candidates for the dense subset $\Phi\subset\ch_{\rm kin}$.\\
$\strut\qquad$The inner product on the physical Hilbert space becomes
\be
   \langle\varphi_{\rm phys}|\psi_{\rm phys}\rangle_{\rm phys}=\langle\varphi|\psi_{\rm phys}\rangle=\pi\int_{\mathbb{R}}dx\,\overline{\varphi(x,x)} \int_{\mathbb{R}}dt\, \psi(t,t)\,,
\ee
and we see that $\langle\psi_{\rm phys}|\psi_{\rm phys}\rangle_{\rm phys}=\pi \left|\int_{\mathbb{R}} dx\, \psi(x,x)\right|^2\geq 0$.\\
$\strut\qquad$The one-dimensionality of the physical Hilbert space tells us that the set of relational states of two particles is, in this situation, basically trivial: intuitively, having translation symmetry of position \emph{and momentum} introduces such a strong notion of symmetry that almost no gauge-invariant physics remains. This can be best understood from a classical perspective: we start out with a four-dimensional kinematical phase space and have a two-dimensional gauge group acting on it. This amounts to two constraints and two gauge-fixing conditions that result in a unique gauge-invariant classical state and so a zero-dimensional reduced phase space \cite{Henneaux:1992ig}.\\
$\strut\qquad$The situation becomes more interesting if we consider a system $S$ that contains more than one mode. For example, let $S$ consist of three modes (it turns out that two is not enough to yield a non-trivial physical Hilbert space), carrying a representation
\be
   T(W'_1):=T^p(W_1)\otimes T^q(W_1)\otimes T^r(W_1) \,
\ee
of the Weyl-Heisenberg group $W_1$, where $p+q+r=-1$ and $p>0$, $q<0$ and $r<0$. It turns out that this can again be interpreted as a representation of the group $W'_1=T^1(W_1)$, and
\be
   T(\Theta;\alpha)=\Theta^{-1}D(\alpha)\otimes D(\bar\alpha)\otimes D(\bar\alpha)\,.
\ee
Hence, we now have
\be
   \Pi_{\rm phys}=\int_{\mathbb{C}}d\alpha\, D(\alpha)\otimes D(\alpha)\otimes D(\bar\alpha)\otimes D(\bar\alpha)\,,
\ee
and
\be
   \ket{\psi_{\rm phys}}=\pi\int_{\mathbb{R}^3}dx_1 dx_2 dx_3 \int_{\mathbb{R}}ds\,\psi(s,x_2-x_1+s,x_3-x_1+s,x_2-x_3+s)\ket{x_1}\ket{x_2}\ket{x_3}\ket{x_1+x_2-x_3}\,.
\ee
The functionals of this form span a non-trivial linear space $\ch_{\rm phys}$. By definition, the inner product on it satisfies
\begin{eqnarray*}
\langle\varphi_{\rm phys}|\psi_{\rm phys}\rangle_{\rm phys}&=&\pi\int_{\mathbb{R}^3}dx_1 dx_2 dx_3 \overline{\varphi(x_1,x_2,x_3,x_1+x_2-x_3)}\int_{\mathbb{R}}ds\,\psi(x_1-s,x_2-s,x_3-s,x_1+x_2-x_3-s)\\
&=& \pi\int_{\mathbb{R}}dt\int_{\mathbb{R}}du \left(\int_{\mathbb{R}}dx\,\overline{\varphi(x,x+t,x+u,x+t-u)}\right)\left(\int_{\mathbb{R}}dy\, \psi(y,y+t,y+u,y+t-u)\right),
\end{eqnarray*}
where the second line follows from the first via a series of suitable substitutions. This defines a non-negative inner product, since
\be
   \langle\psi_{\rm phys}|\psi_{\rm phys}\rangle_{\rm phys}=\pi \int_{\mathbb{R}}dt \int_{\mathbb{R}}du \left| \int_{\mathbb{R}}\psi(x,x+t,x+u,x+t-u)\right|^2\geq 0 \,.
\ee
It is thus natural to identify $\mathcal{H}_{\rm phys}$ with a subspace of $L^2(\mathbb{R}^2)$, containing the functions $\tilde\psi(t,u):=\int_{\mathbb{R}}dx\, \psi(x,x+t,x+u,x+t-u)$ for $\psi\in\Phi$, where $\Phi\subseteq\mathcal{H}_{\rm kin}$ is a suitably chosen subspace on which this map $\psi\mapsto\tilde\psi$ makes sense and converges.
\end{example}

\subsection{Dirac observables from group averaging}\label{sec_diracobs}

There is a useful alternative form in which we can write the physical inner product. To understand it, we first need to understand how to map kinematical operators from $\ca_{\rm kin}=\cl(\ch_{\rm kin})$ to gauge-invariant operators in $\ca_{\rm phys}:=\cl(\ch_{\rm phys})$, where $\cl(\ch)$ denotes the set of linear operators on a Hilbert space $\ch$.
We choose to work with general linear rather than bounded operators here, so as to accommodate unbounded operators such as, say, relative distances or momenta in the case of the translation group $G=(\mathbb{R},+)$. $\ca_{\rm phys}$ and $\ca_{\rm kin}$ are therefore not algebras without the further stipulation of some common domain. Nevertheless, with some abuse of terminology, we shall sometimes refer to these sets as observable algebras for simplicity.

We introduce the $G$-twirl, namely the \emph{incoherent group averaging} of kinematical operators $A\in\ca_{\rm kin}$ over the action $U_{RS}$ of the gauge group
\ba
\cg(A):=\int_G dg\,U_{RS}(g)\,A\,\left[U_{RS}(g)\right]^\dag\,,\label{Gtwirl}
\ea
to turn them into gauge-invariant operators. 

When $G$ is non-compact, the integral that defines $\cg(A)$ will not in general converge for all operators $A$. For example, it diverges if $A$ is already invariant, i.e.\ when $[A,U_{RS}(g)]=0$ for all $g\in G$. For a more rigorous discussion of convergence of the $G$-twirl, see Appendix~\ref{app_groupaverage}. Essentially, we define the domain of $\cg$ to be the subset of $\ca_{\rm kin}$ on which its converges. In the next subsection, we shall see a particular class of kinematical observables $A$ for which this integral should typically converge.

Before demonstrating invariance of the twirled operators, we introduce the weak equality `$\approx$' for notational simplicity, which indicates coincidence of operators on $\ch_{\rm phys}$ \cite{Hoehn:2019owq,Hoehn:2020epv}. 
Suppose we are given two kinematical operators $O_1,O_2\in\ca_{\rm kin}$ whose action on the dense subset $\eta(\Phi)\subset\ch_{\rm phys}$ is well-defined,\footnote{By this, we do not necessarily mean that the $O_i$ leave $\eta(\Phi)$ invariant; they may also map outside of this subset.} where $\eta$ is the group-averaging rigging map. We say that they are \emph{weakly equal}, denoted $O_1\approx O_2$, if and only if 
\ba\label{weakeq}
 O_1\,\ket{\psi_{\rm phys}} = O_2\,\ket{\psi_{\rm phys}}\, ,\q\forall\,\ket{\psi_{\rm phys}}\in\eta\left(\Phi\right)\subset\ch_{\rm phys}\,.
\ea
By contrast, borrowing the terminology of constrained systems \cite{diracLecturesQuantumMechanics1964,Henneaux:1992ig}, we call the usual algebraic equality `=' of operators within $\ca_{\rm kin}$ \emph{strong}. 

This enables us to distinguish weak and strong invariance of operators:
\begin{definition}
We say that an observable $A\in\ca_{\rm kin}$ is a \emph{strong Dirac observable} if it commutes strongly with the group actions, i.e.\ $[A,U_{RS}(g)]=0$ for all $g\in G$. By contrast, $A$ is called a \emph{weak Dirac observable} if it commutes weakly with the group action, i.e.\ $[A,U_{RS}(g)]\approx0$ for all $g\in G$.
\end{definition}
Clearly, a strong Dirac observable is also a weak Dirac observable, but the converse is not in general true.

We have:

\begin{lemma}\label{lem_dirobs}
For all $A$ for which Eq.~\eqref{Gtwirl} converges, $\cg(A)$ is a strong Dirac observable, i.e.\
\ba
[\cg(A),U_{RS}(g)]=0\,,\q\forall\,g\in G\,.\nn
\ea
\end{lemma}
\begin{proof}
We have for any $g\in G$, using the left-invariance of $dg$,
\ba
U_{RS}(g)\,\cg(A) &=&\int_G dg'\,U_{RS}(g g')\,A\,\left[U_{RS}(g')\right]^\dag\nn\\
&=&\int_G d(gg')\,U_{RS}(g g')\,A\,\left[U_{RS}(g')\right]^\dag\nn\\
&=&\int_G d\tilde g'\,U_{RS}(\tilde g')\,A\,\left[U_{RS}(g^{-1}\tilde g')\right]^\dag\nn\\
&=&\cg(A)\,U_{RS}(g)\,.\nn
\ea
\end{proof}

The following tells us that the coherent and incoherent group averaging of kinematical operators are weakly equal.
\begin{lemma}\label{lem_G}
For all $A$ for which Eq.~\eqref{Gtwirl} converges and which admit a well-defined action on $\eta(\Phi)$, the projection of $A$'s action onto the physical Hilbert space coincides weakly with its $G$-twirl:
\ba
\Pi_{\rm phys}\,A\approx \cg(A)\,.
\ea
\end{lemma}

\begin{proof}
The statement is informally verified by direct computation
\ba
\Pi_{\rm phys}\,A\,\Pi_{\rm phys} &=&\int_G dg\, U_{RS}(g)\,A\,\int_G dg'\,U_{RS}(g')=\int_G dg\,U_{RS}(g)\,A\,\left[U_{RS}(g)\right]^\dag\,\int_G dg'\,U_{RS}(g g')\nn\\
&=&\cg(A)\,\Pi_{\rm phys}\,,\nn
\ea
where in the last step we have made use of the left-invariance of the Haar measure. The claim then follows from acting with the last identity on $\Phi\subset\ch_{\rm kin}$ and the definition of the weak equality in Eq.~\eqref{weakeq}.
\end{proof}

Being an invariant operator, restricting the action of $\cg(A)$ to the physical Hilbert space defines an element of $\ca_{\rm phys}$, i.e.\ $\cg(A)\upharpoonright\ch_{\rm phys}\in\ca_{\rm phys}$. In particular, informally, $\cg(A)$ thus commutes with coherent group averaging, i.e.
\ba 
[\cg(A),\Pi_{\rm phys}]=0\,
\ea
and so with the rigging map $\eta$. Informally at least, (iii) in Sec.~\ref{sec_ph} is thus satisfied for $G$-twirled kinematical observables.

\subsection{Relational Dirac observables}\label{ssec_relobs}

In Sec.~\ref{sec_diracobs} we explained how to construct Dirac observables from kinematical operators that are not already Dirac observables. We can now make use of our assumptions on the frame $R$ from Sec.~\ref{sec_QRF} in order to extend an especially useful class of Dirac observables to $G$-frames: \emph{relational Dirac observables} \cite{Rovelli:1989jn,Rovelli:1990jm,Rovelli:1990ph,Rovelli:1990pi,rovelliQuantumGravity2004,dittrichPartialCompleteObservables2007,Dittrich:2005kc,Dittrich:2006ee,Dittrich:2007jx,thiemannModernCanonicalQuantum2008,Tambornino:2011vg,Bojowald:2010qw,Bojowald:2010xp,Hohn:2011us,hoehnHowSwitchRelational2018,Hoehn:2018whn,Hoehn:2019owq,Hoehn:2020epv,Hoehn:2021wet,Krumm:2020fws,Hoehn:2021flk,Chataignier:2019kof,Chataignier:2020fys,Dittrich:2016hvj,Dittrich:2015vfa,Carrozza:2021sbk}.

The basic idea is to use the frame orientation states $\ket{\phi(g)}_R$ to single out a point or an equivalence class of points on $G$-orbits in $\ca_{\rm kin}$ or $\ch_{\rm kin}$, which is essentially a gauge-fixing. When $G$ acts regularly on the coherent state system $\{U_{R},\ket{\phi(g)}_R\}$, we can single out a point on the $G$-orbits. When it only acts transitively, i.e.\ the frame orientation states admit an isotropy group $H$, we can only single out equivalence classes of points on the $G$-orbits. As we shall see shortly, in that case, the quantum reference frame can only resolve the associated equivalence classes of states and observables of the system $S$.

Let us consider a special class of non-invariant kinematical observables that we can then group average: observables aligned to frame $R$ in orientation $g\in G$ \cite{Krumm:2020fws,Hoehn:2021flk}. Let $f_S$ be an arbitrary system operator in $\ca_S:=\cl(\ch_S)$. The observable ${\ket{\phi(g)}\!\bra{\phi(g)}_R\otimes f_S}\in\ca_{\rm kin}$ is aligned to $R$ and can be regarded as a non-invariant description of the system property $f_S$ conditional on $R$ being in orientation $g\in G$. Indeed, this operator transforms covariantly under the group, i.e.\ $U_{RS}(g')\left(\ket{\phi(g)}\!\bra{\phi(g)}_R\otimes f_S\right)U_{RS}^\dag(g')=\ket{\phi(g'g)}\!\bra{\phi(g'g)}_R\otimes U_S(g')f_SU^\dag_S(g')$, and so is a good candidate for a kinematical observable on which group averaging might converge. Note that when the coherent state system admits a non-trivial isotropy subgroup $H$ and $g'=h\in H$,  we have $U_{RS}(h)\left(\ket{\phi(g)}\!\bra{\phi(g)}_R\otimes f_S\right)U_{RS}^\dag(h)=\ket{\phi(g)}\!\bra{\phi(g)}_R\otimes U_S(h)f_SU^\dag_S(h)$, hence we obtain the same frame orientation, but in a general distinct system observable and therefore, in turn, a distinct aligned observable. Kinematically, these observables are thus generally distinguishable, but this will no longer be the case upon group averaging. Invoking the $G$-twirl in Eq.~\eqref{Gtwirl}, we can map the aligned observable into a Dirac observable, generalizing the construction in \cite{Loveridge:2019phw,Hoehn:2019owq,Hoehn:2020epv,Hoehn:2021wet,Krumm:2020fws,Hoehn:2021flk,Chataignier:2019kof,Chataignier:2020fys} to unimodular Lie groups:
\ba
F_{f_S,R}(g)&:=&\cg\left(\ket{\phi(g)}\!\bra{\phi(g)}_R\otimes f_S\right)\nn\\
&\approx&\Pi_{\rm phys}\left(\ket{\phi(g)}\!\bra{\phi(g)}_R\otimes f_S\right)\,,\label{qrelobs}
\ea
where the second line holds on account of Lemma~\ref{lem_G}. Thanks to Lemma~\ref{lem_dirobs}, $F_{f_S,R}(g)$ is informally a strong Dirac observable. We emphasize that this is a \emph{frame-orientation-conditional $G$-twirl for the gauge group} to distinguish it from the relation-conditional $G$-twirl for the symmetry group action later in Sec.~\ref{sec_relobschanges}. Encoding in an invariant manner the system property $f_S$ conditioned on $R$ being in orientation $g$, we call it a \emph{relational Dirac observable}. In \cite{Hoehn:2019owq,Hoehn:2020epv,periodic,Chataignier:2019kof,Chataignier:2020fys} it was shown for Abelian groups (and especially for relational dynamics) that Eq.~\eqref{qrelobs} indeed furnishes a quantization of classical relational observables \cite{Dittrich:2005kc,dittrichPartialCompleteObservables2007,thiemannModernCanonicalQuantum2008}. The construction of $F_{f_S,R}(g)$ for $g=e$ was also given in \cite{Bartlett:2007zz} for the case of the left regular representation of compact Lie groups, i.e.\ ideal reference frames, but not called a relational observable. Viewed as a function of $f_S$, this construction also coincides for $g=e$ with the `relativisation' map appearing in   \cite{loveridgeSymmetryReferenceFrames2018a,miyadera2016approximating,loveridge2017relativity}, defined with respect to an arbitrary $G$-covariant POVM of the reference.

Before we discuss the physical meaning of relational observables and the isotropy group further, let us address the question of whether they are well-defined, i.e.\ whether the involved group averaging converges. To this end, note that we can rewrite Eq.~\eqref{qrelobs}, using an integration variable transformation and right invariance of $dg$, in the following form\footnote{This frame-orientation-conditional gauge transformation of the system observable $f_S$ is the quantum analog of the classical relational observable $\big[\tilde E^{-1}(R)g^{-1}\big]^{-1}\triangleright f $ in \cite[Eqs.~(68\,\&\,70)]{Carrozza:2021sbk}. Here, $\tilde E$ is an equivariant map from the gauge group $G$ into the frame configurations labeled by $R$ and $\triangleright$ denotes the gauge group action on some kinematical phase space function $f$, depending on the remaining degrees of freedom. When the frame is complete, this map is invertible (if incomplete only on the coset space). Since $R$ is dynamical, the classical relational observable is thus also a frame-orientation-conditional gauge transformation of $f$ (see also the example of special relativity in Sec.~\ref{ssec_SR}). Now recall from Sec.~\ref{sec_QRF} that the coherent states modeling the quantum frame orientations constitute an equivariant map $\tilde E$ from the group into the frame Hilbert space too. Hence, Eq.~\eqref{qrelobs2} is the genuine quantum analog of the classical expression above. }
\ba 
F_{f_S,R}(g)&=&\int_Gdg'\,\ket{\phi(g'g)}\!\bra{\phi(g'g)}_R\otimes U_S(g')\,f_S\,U^\dag_S(g')\nn\\
&=&\int_Gdg'\,\ket{\phi(g')}\!\bra{\phi(g')}_R\otimes U_S(g'g^{-1})\,f_S\,U_S^\dag(g'g^{-1})\label{qrelobs2}\\
&=&\int_GE(dg')\otimes U_S(g'g^{-1})\,f_S\,U_S^\dag(g'g^{-1})\,,\nn
\ea 
where $E(dg)=dg\,\ket{\phi(g)}\!\bra{\phi(g)}_R$ are the effect densities
giving rise to the normalized $G$-covariant POVM $\cb(G)\ni Y\mapsto E_Y=\int_YE(dg)$ in Eq.~\eqref{covPOVM}. The group averaging inside $F_{f_S,R}(g)$ can thus be understood as integration with respect to this POVM. From here, the integral is defined on a dense subset of bounded operators in $\ca_S$, and on the whole space if $G$ is Abelian, e.g.\ see \cite[Sec.\ 4.1]{loveridgeSymmetryReferenceFrames2018a}. Relational observables are thus often well-defined (though the isotropy group $H$ needs to be compact, see below and Lemma~\ref{lemCompact}).

The covariant POVM construction in Eq.~\eqref{qrelobs2} also equips the relational observables with a consistent probabilistic interpretation that we shall make more precise in Sec.~\ref{ssec_condprob}. We shall see there that the relational observables encode the question ``what is the probability that $f_S$ has value $f$, given that the frame has orientation $g$?'' in a gauge-invariant manner, however, \emph{only} when acting on physical states (owing to the gauge symmetry induced redundancy in their description). This yields a completely relational (and external-frame-independent) definition of such a probability. This is the quantum generalization of the fact that classically such observables encode the value of $f_S$ when $R$ is in configuration $g$ \cite{Dittrich:2005kc,dittrichPartialCompleteObservables2007,thiemannModernCanonicalQuantum2008,Tambornino:2011vg,hoehnHowSwitchRelational2018,Hoehn:2018whn,Hoehn:2019owq,Hoehn:2020epv}.

\begin{example}\label{ex_5}
Let us consider once more the translation group $G=(\mathbb{R},+)$ in one dimension, as in Example~\ref{ex_2} and Refs.~\cite{Vanrietvelde:2018pgb,Vanrietvelde:2018dit}, and let us suppose $S$ is a second particle, so that $\ch_{\rm kin}=L^2(\mathbb{R})\otimes L^2(\mathbb{R})$. We are thus dealing with the regular representation. Denote by $q_R,q_S$ the position operators of $R$ and $S$. The relational observable encoding the position of the system $S$, conditioned on the frame particle $R$ being in orientation (i.e.\ position) $y$, is the relative distance:
\ba 
F_{q_S,R}(y)&=&\int_\mathbb{R}dx\,\ket{x}\!\bra{x}_R\otimes e^{i(x-y)p_S}\,q_S\,e^{i(x-y)p_S}=\int_\mathbb{R}dx\ket{x}\!\bra{x}_R\otimes\left(q_S+x-y\right)\nn\\
&=&q_S-q_R+y\,.\label{reldist}
\ea 
\end{example}

Let us now address the impact of a non-trivial isotropy group $H$. Generalizing Lemma~\ref{lemCompact}, the $G$-twirl in Eq.~\eqref{qrelobs2} can in this case be split into a separate group averaging over the coset $X=G/H$ and the isotropy group $H$ (assumed unimodular and compact as before):
\ba 
F_{f_S,R}(g)&=&\int_Xd_Xx\ket{g_xg}\!\bra{g_xg}_R\otimes U_S(g_x)\left(\int_H d_Hh\,U_S(h)\,f_S\,U_S^\dag(h)\right)U_S^\dag(g_x)\,.
\ea 
That is, we first have to average the system observable $f_S$ over the isotropy group $H$, before averaging the result together with the frame orientations over the coset. Note that this defines an equivalence class $[f_S]$ of system observables in $\ca_S$: 
\ba \label{HfS}
f_S\sim f_S'\q\q\Leftrightarrow\q\q f_S^{\rm avg}:=\frac 1 {{\rm Vol}(H)}\int_H d_Hh\,U_S(h)\,f_S\,U_S^\dag(h)=\frac 1 {{\rm Vol}(H)}\int_H d_Hh\,U_S(h)\,f'_S\,U_S^\dag(h)\,.
\ea 
$f_S\sim f_S'$ are then indistinguishable relative to the frame $R$: they give rise to the \emph{same} relational observables. In particular, $f_S^{\rm avg}$ is invariant under the isotropy group and an element of $\ca_S$ since $H$ is compact. We can thus also write $[f_S]=f_S^{\rm avg}$. Physically, this means that the \emph{frame $R$ can only resolve those properties of $S$ that are invariant under the isotropy group of the frame}. In this sense, gauge invariance enforces the isotropy properties of the frame onto the system, even if the system did not originally have those properties. In the case of periodic clocks this means that only periodic system observables can evolve relative to the clock in a gauge-invariant manner \cite{periodic}.

For later purpose, we note that since the label $g$ can run over the entire group $G$, the object $F_{f_S,R}(g)$, for fixed $f_S\in\ca_S$, constitutes, in fact, a $\dim \,G$-parameter family of Dirac observables. That is, \emph{every relational observable family defines an orbit in $\ca_{\rm phys}$}; this orbit is of dimension $\dim\,X$ or less, since $F_{f_S,R}(gh)=F_{f_S,R}(g)$ for all $h\in H$.

The construction of relational observables preserves the identity. This result will be useful below when rewriting the physical inner product in Eq.~\eqref{PIP} and is an extension of \cite[Lemma 4.1]{Hoehn:2014fka}. 
\begin{lemma}\label{lem_Rd}
The frame orientation `projector' $\ket{\phi(g)}\!\bra{\phi(g)}_R\otimes \mathbf{1}_S$ $G$-twirls to the identity,
\ba
\cg\left(\ket{\phi(g)}\!\bra{\phi(g)}_R\otimes \mathbf{1}_S\right)=\mathbf{1}_{RS}\,,\q\forall\,g\in G\,,\nn
\ea
where $\mathbf{1}_{RS}$ is the identity operator on $\ch_{\rm kin}$. 
\end{lemma}

\begin{proof}
For arbitrary $g\in G$ we have
\ba
\cg\left(\ket{\phi(g)}\!\bra{\phi(g)}_R\otimes \mathbf{1}_S\right)&=&\int_G dg'\,U_{RS}(g')\left(\ket{\phi(g)}\!\bra{\phi(g)}_R\otimes \mathbf{1}_S\right)\left[U_{RS}(g')\right]^\dag=\int_G dg'\,\ket{\phi(g'g)}\!\bra{\phi(g'g)}_R\otimes \mathbf{1}_S\,.\nn
\ea
Since $G$ is unimodular, the resolution of the identity in Eq.~\eqref{resid} implies the claim.
\end{proof}

Next, we shall show that relational observables define an algebra homomorphism from a subset $\ca_S^{\rm phys}\subset\ca_S$ of the system observable algebra to the algebra $\ca_{\rm phys}$ of Dirac observables on $\ch_{\rm phys}$. To this end, we consider the following, \emph{a priori} group element dependent object:\footnote{ Expressions of the form $\ip{\varphi_R \otimes \mathbf{1}_S}{\Omega \varphi_R \otimes \mathbf{1}_S}$ can be understood as follows: the collection of operators $\sum_{i,j}A_i \otimes B_j$ is weakly dense in $B(\mathcal{H}_R \otimes \mathcal{H}_S)$ and therefore $\ip{\varphi_R \otimes \mathbf{1}_S}{\Omega \varphi_R \otimes \mathbf{1}_S}$ can be defined on pure tensors as $\ip{\varphi_R \otimes \mathbf{1}_S}{(A \otimes B) \varphi_R \otimes \mathbf{1}_S} = \ip{\varphi_R }{A \varphi_R }B$ and extended to the whole space by linearity and continuity. Alternatively the expression $\ip{\varphi_R \otimes \cdot}{\Omega \varphi_R \otimes \cdot}$ defines a real, bounded quadratic form on $\mathcal{H}_S \times \mathcal{H}_S$ and thus the unique
self-adjoint operator $\Omega_{\varphi_R}$. This extends to mixed states through the following construction: the map $\Gamma_{\sigma}: \mathcal{L}(\mathcal{H}_R \otimes \mathcal{H}_S) \to \mathcal{L}(\mathcal{H}_S)$ by $\rm{Tr} [\rho \Gamma_{\sigma}(\Omega)] = \rm{Tr} [ \sigma \otimes \rho \Omega ]$, to hold for all system states $\rho$, defines a completely positive normal conditional expectation. 
Setting $\sigma= \ket{\phi(g)}\bra{\phi(g)}$ and $\Omega = \Pi_{\rm phys}$, we see that $\Pi_S^{\rm phys}(g) = \Gamma_{\ket{\phi(g)}}  (\Pi_{\rm phys})$ (with a small abuse of notation), and thus the physical system Hilbert space defined by that projector acquires a $g$-dependence. From here, it is simple to show that $U(g')\Gamma_{\ket{\phi(g)}}(\Pi_{\rm phys})U(g')^{\dag} = \Gamma_{\ket{\phi(g'g)}}(\Pi_{\rm phys})$, i.e., $U(g')\Pi_S^{\rm phys}(g)U(g')^{\dag} = \Pi_S^{\rm phys}(g'g)$, from which Eq.~\eqref{projgind} follows. For $\Pi_S^{\rm phys}(g)$ to be independent of $g$, we require that $\Pi_S^{\rm phys} = \Pi_S^{\rm phys}(g)$ for all $g$, hence $[\Pi_S^{\rm phys},U(g)]=0$ for all $g$.}
\ba 
\Pi_S^{\rm phys}(g):=\left(\bra{\phi(g)}_R\otimes \mathbf{1}_S\right)\Pi_{\rm phys}\left(\ket{\phi(g)}_R\otimes\mathbf{1}_S\right)=\int_G dg'\braket{\phi(g)|\phi(g'g)}_R\,U_S(g')\,.\label{Sproj}
\ea

\begin{lemma}\label{lem_Sproj}
Let $G$ be unimodular, then $\Pi_S^{\rm phys}(g)$ is an orthogonal projector. This projector is $g$-independent if and only if
\ba 
\big[U_S(g),\Pi_S^{\rm phys}(e)\big]=0\,\q\q\forall\,g\in G\,,\label{projgind}
\ea 
which is equivalent to 
\ba 
\big[U_S(g),\Pi_S^{\rm phys}(g)\big]=0\,\q\q\forall\,g\in G\,.\label{projgindagain}
\ea
Furthermore, we necessarily have $\big[U_S(g),\Pi_S^{\rm phys}(e)\big]=0$ $\forall\,g\in H$, where $H$ is the isotropy subgroup, and for all $g$ in the center of $G$. 
\end{lemma}

\begin{proof}
The property $\left(\Pi_S^{\rm phys}(g)\right)^\dag=\Pi_S^{\rm phys}(g)$ is immediate. We continue with idempotence:
\ba 
\left(\Pi_S^{\rm phys}(g)\right)^2&=&\left(\bra{\phi(g)}_R\otimes \mathbf{1}_S\right)\left(\strut\Pi_{\rm phys}\left(\ket{\phi(g)}\!\bra{\phi(g)}_R\otimes\mathbf{1}_S\right)\Pi_{\rm phys}\right)\left(\ket{\phi(g)}_R\otimes\mathbf{1}_S\right)\nn\\
&=&\left(\bra{\phi(g)}_R\otimes \mathbf{1}_S\right)\Pi_{\rm phys}\left(\ket{\phi(g)}_R\otimes\mathbf{1}_S\right)=\Pi_S^{\rm phys}(g)\,,\nn
\ea 
where in going to the second line we have made use of Lemmas~\ref{lem_G} and~\ref{lem_Rd}.

Next, we note that
\ba 
\Pi_S^{\rm phys}(g)&=&\left(\bra{\phi(e)}_R U_R^\dag(g)\otimes\mathbf{1}_S\right)\Pi_{\rm phys}\left(U_R(g)\ket{\phi(e)}_R\otimes\mathbf{1}_S\right)\nn\\
&=&U_S(g)\,\Pi_S^{\rm phys}(e)\,U_S^\dag(g)\,,\label{projgind2}
\ea 
where in the last line we made use of $U_{RS}(g)\,\Pi_{\rm phys}=\Pi_{\rm phys}$ (cf.\ Eq.~\eqref{stateinv}). Hence, if condition \eqref{projgind} is fulfilled, we have $\Pi_S^{\rm phys}(g)=\Pi_S^{\rm phys}(e)$, $\forall\,g\in G$. Conjugating Eq.~\eqref{projgind2} with $U_S^\dag(g)(\cdot)U_S(g)$, the same conclusion follows if Eq.~\eqref{projgindagain} holds. Conversely, inserting $\Pi_S^{\rm phys}(g)=\Pi_S^{\rm phys}(e)$ for arbitrary $g\in G$ into the left hand side (resp.\ right hand side) of Eq.~\eqref{projgind2} implies Eq.~\eqref{projgind} (resp.\ Eq.~\eqref{projgindagain}).

Using Eq.~(\ref{Sproj}), the final statement follows from 
\ba 
U_S(h)\,\Pi_S^{\rm phys}(e)\,U_S^\dag(h)=\int_Gdg'\,\braket{\phi(h)|\phi(g'h)}_RU_S(g')=\Pi_S^{\rm phys}(e)\,,
\ea 
invoking a change of variables $hg'h^{-1}\to g'$ and left-right invariance of $dg'$.
\end{proof}

For general group representations, $\Pi_S^{\rm phys}(g)$ will project onto a $g$-dependent\,---\,i.e.\ frame orientation dependent\,---\,subspace $\ch_{S,g}^{\rm phys}\subset\ch_S$ of the kinematical system Hilbert space, which, generalizing \cite{Hoehn:2019owq,Hoehn:2020epv,periodic}, we call the \emph{physical system Hilbert space}. We shall see later that this is the reduced Hilbert space for $S$ of the extension of the Page-Wootters formalism \cite{pageEvolutionEvolutionDynamics1983,pageClockTimeEntropy1994,giovannettiQuantumTime2015,Smith:2017pwx,Smith:2019imm,Hoehn:2019owq,Hoehn:2020epv,periodic} to general groups. Eq.~\eqref{projgind2} shows us that as we change the orientation $g$ of the frame $R$, the physical system Hilbert space $\ch_{S,g}^{\rm phys}$ will `rotate' through $\ch_S$, unless Eq.~\eqref{projgind} is satisfied. We shall illustrate this effect in examples in Sec.~\ref{sec_examples}. The lemma also implies that the isotropy group $H$ of the frame orientation states is an isotropy group for the physical system Hilbert space $\ch_{S,g}^{\rm phys}$ inside $\ch_S$, as is the center of $G$.

Eq.~\eqref{Sproj} entails that the projector $\Pi_S^{\rm phys}(g)$ is $g$-independent, for example, when the coherent state system $\{U_R,\ket{\phi(g)}_R\}$ also admits a unitary action from the right $V_R$ of the group $G$:
\ba 
\braket{\phi(g)|\phi(g'g)}_R=\braket{\phi(g)|V_R(g^{-1})|\phi(g')}_R=\braket{\phi(e)|\phi(g')}_R\,.
\ea 
This is, in particular, the case for Abelian groups and also when $\ch_R$ carries the regular representation of $G$, i.e.\ when $\ch_R\simeq L^2(G,dg)$ and so $R$ constitutes an ideal reference frame. Notice that in the latter case of an ideal frame, we even have $\Pi_S^{\rm phys}=\mathbf{1}_S$ so that $\ch_S^{\rm phys}=\ch_S$, owing to $\braket{g|g'}_R=\delta(g,g')$.

Systems of coherent states $\{U_R,\ket{\phi(g)_R}\}$ which admit of a unitary action to the right $V_R$ of the group $G$ on the coherent states are called left-right (LR) systems of coherent states. We classify all systems of LR coherent states for compact $G$ in Sec.~\ref{ssec_LR}. This will become relevant when distinguishing physical symmetries as frame reorientations from gauge transformation below in Sec.~\ref{ssec_symgauge}.

We can use $\Pi_S^{\rm phys}$ to project $\ca_S$ to the algebra of observables $\ca_{S,g}^{\rm phys}:=\cl(\ch_{S,g}^{\rm phys})$ on the physical system Hilbert space:
\ba 
f_{S,g}^{\rm phys}:=\Pi_S^{\rm phys}(g)\,f_S\,\Pi_{S}^{\rm phys}(g)\,,\label{fsys}
\ea 
for arbitrary $f_S\in\ca_S$.\footnote{We label $f_{S,g}^{\rm phys}$ by an index $g$ rather than as a function of $g$ in order to distinguish the operator from the `Heisenberg picture' observables encountered later.} Depending on the group representation, $\ca_{S,g}^{\rm phys}$ may `rotate' through $\ca_S$ as we change the frame orientation $g$ of $R$. 

The role of the projector becomes clear through the following result, which is a generalization of \cite[Lemma 1]{Hoehn:2019owq}:

\begin{lemma}\label{lem_weakequiv}
Let $G$ be unimodular. The relational observables $F_{f_S,R}(g)$ and $F_{f_{S,g}^{\rm phys},R}(g)$ (as defined in Eq.~(\ref{qrelobs})) are weakly equal, i.e.
\ba 
F_{f_S,R}(g)\approx F_{f_{S,g}^{\rm phys},R}(g)\,,
\ea 
where
$f_{S,g}^{\rm phys}$ is given by Eq.~\eqref{fsys}. Hence, the relational observables form weak equivalence classes, where $F_{f_S,R}(g)$ and $F_{\tilde f_S,R}(g)$ are equivalent if $\Pi_S^{\rm phys}(g)\, f_S\,\Pi_S^{\rm phys}(g)=\Pi_S^{\rm phys}(g)\,\tilde f_S\,\Pi_S^{\rm phys}(g)$.
\end{lemma}

\begin{proof}
Using Lemma~\ref{lem_G}, direct computation yields
\ba
F_{f_{S,g}^{\rm phys},R}(g)&\approx&\Pi_{\rm phys}\left(\ket{\phi(g)}\!\bra{\phi(g)}_R\otimes \Pi_S^{\rm phys}(g)\,f_S\,\Pi_S^{\rm phys}(g)\right)\nn\\
&=&\Pi_{\rm phys}\left(\ket{\phi(g)}\!\bra{\phi(g)}_R\otimes \Pi_S^{\rm phys}(g)\right)\left(\mathbf{1}_R\otimes f_S\,\Pi_S^{\rm phys}(g)\right)\,.\nn
\ea 
We observe that
\ba 
\left(\ket{\phi(g)}\!\bra{\phi(g)}_R\otimes \Pi_S^{\rm phys}(g)\right)&=&\left(\ket{\phi(g)}\!\bra{\phi(g)}_R\otimes \mathbf{1}_S\right)\Pi_{\rm phys}\left(\ket{\phi(g)}\!\bra{\phi(g)}_R\otimes\mathbf{1}_S\right)\,.\label{bla1}
\ea 
Hence, using Lemmas~\ref{lem_G} and~\ref{lem_Rd}, we have
\ba 
F_{f_{S,g}^{\rm phys},R}(g)&\approx&\Pi_{\rm phys}\left(\ket{\phi(g)}\!\bra{\phi(g)}_R\otimes f_S\,\Pi_S^{\rm phys}(g)\right)\nn\\
&=&\Pi_{\rm phys}\left(\mathbf{1}_R\otimes f_S\right)\left(\ket{\phi(g)}\!\bra{\phi(g)}_R\otimes \Pi_S^{\rm phys}(g)\right)\nn\\
&\approx&\Pi_{\rm phys}\left(\ket{\phi(g)}\!\bra{\phi(g)}_R\otimes f_S\right)\nn\\
&\approx&  F_{f_S,R}(g)\,,\nn
  \ea 
  where in going from the second to the third line, we once more invoked Eq.~\eqref{bla1} and Lemmas~\ref{lem_G} and~\ref{lem_Rd}.
\end{proof}

Invoking these weak equivalence classes of relational observables, we are now in a position to demonstrate in which sense they define an algebra homomorphism.  The following theorem is a generalization of the homomorphisms established in \cite{Hoehn:2019owq,Hoehn:2020epv,periodic} for general representations of Abelian groups, and in \cite{Bartlett:2007zz} for the regular representation of compact Lie groups. For the regular representation, the equivalence classes are trivial so that $\ca_{S,g}^{\rm phys}=\ca_S$. The following result is also the quantum analog of the weak algebra homomorphism constructed in \cite{dittrichPartialCompleteObservables2007,Dittrich:2005kc} for classical relational observables.

\begin{theorem}\label{lem_homo}
For unimodular $G$, the relationalization map relative to reference system $R$ in orientation $g$
\ba
\mathbf{F}_{R,g}:\ca_{S,g}^{\rm phys}\rightarrow\cg\left(\ca_{\rm kin}\right) \, ,\nn\\
f_{S,g}^{\rm phys}\mapsto F_{f_{S,g}^{\rm phys},R}(g) \,, \nn
\ea
is a weak $*$-homomorphism, i.e.\ it preserves addition, multiplication, and the adjoint in the form:
\ba
F_{a_{S,g}^{\rm phys}+b_{S,g}^{\rm phys}\cdot c_{S,g}^{\rm phys},R}(g)&\approx& F_{a_{S,g}^{\rm phys},R}(g)+F_{b_{S,g}^{\rm phys},R}(g)\cdot F_{c_{S,g}^{\rm phys},R}(g)\,,\nn\\ F_{(a_{S,g}^{\rm phys})^\dag,R}(g)\upharpoonright \ch_{\rm phys}&=& \left(F_{a_{S,g}^{\rm phys},R}(g)\upharpoonright {\ch_{\rm phys}}\right)^*\,, \forall\, a_{S,g}^{\rm phys},b_{S,g}^{\rm phys},c_{S,g}^{\rm phys}\in\ca_{S,g}^{\rm phys}\,.\label{homo}
\ea
Here, $\dag,*$ denote the adjoints on $\ch_S$ and $\ch_{\rm phys}$, respectively, and $\upharpoonright\ch_{\rm phys}$ denotes restriction of the domain to $\ch_{\rm phys}$. (Note that Lemma~\ref{lem_G} implies $\Pi_{\rm phys}\,F\,\Pi_{\rm phys}=F\,\Pi_{\rm phys}$ for $F\in\mathcal{G}(\ca_{\rm kin})$, hence the $*$-adjoint in Eq.~(\ref{homo}) is well-defined.)
In the case that $U_R$ is the left regular representation of $G$\,---\,hence ideal reference frames\,---\,this is also a strong $*$-homomorphism, i.e.\ the weak equality $\approx$ in Eq.~\eqref{homo} can be replaced by the strong equality $=
$.
\end{theorem}
Note that we always have the trivial identity $F_{f_s^\dagger,R}(g)=\left(F_{f_s,R}(g)\right)^\dagger$ for $\dagger$ the adjoint on $\ca_{\rm kin}$; the above says that the analog of this is also true with respect to the physical inner product.
\begin{proof}
Verifying the homomorphism with respect to addition is trivial. As regards preservation of multiplicativity, using Eq.~(\ref{qrelobs}), we have
\ba
F_{b_{S,g}^{\rm phys}\cdot c_{S,g}^{\rm phys},R}(g) \approx\Pi_{\rm phys}\left(\ket{\phi(g)}\!\bra{\phi(g)}_R\otimes b_{S,g}^{\rm phys}\cdot c_{S,g}^{\rm phys}\right)\,.\nn
\ea
Since $c_{S,g}^{\rm phys}\in\ca_{S,g}^{\rm phys}$, we can rewrite this as
\ba 
F_{b_{S,g}^{\rm phys}\cdot c_{S,g}^{\rm phys},R}(g) &\approx&\Pi_{\rm phys}\left(\mathbf{1}_R\otimes b_{S,g}^{\rm phys}\right)\left(\ket{\phi(g)}\!\bra{\phi(g)}_R\otimes\Pi_S^{\rm phys}(g)\right)\left(\mathbf{1}_R\otimes c_{S,g}^{\rm phys}\right)\nn\\
&=&\Pi_{\rm phys}\left(\ket{\phi(g)}\!\bra{\phi(g)}_R\otimes b_{S,g}^{\rm phys}\right)\Pi_{\rm phys}\left(\ket{\phi(g)}\!\bra{\phi(g)}_R\otimes c_{S,g}^{\rm phys}\right)\nn\\
&\approx&F_{b_{S,g}^{\rm phys},R}(g)\cdot F_{c_{S,g}^{\rm phys},R}(g)\,,\nn
\ea
where we have invoked the identity in Eq.~\eqref{bla1} in going to the second line and made use of Lemma~\ref{lem_G} in the last line.

Lastly, preservation of the adjoint is checked directly in terms of the physical inner product \eqref{PIP}. Let $\ket{\psi_{\rm phys}},\ket{\phi_{\rm phys}}\in\ch_{\rm phys}$ be arbitrary. Then
\ba 
\braket{\phi_{\rm phys}|F_{(a_{S,g}^{\rm phys})^\dag,R}(g)|\psi_{\rm phys}}_{\rm phys}&=&\braket{\phi_{\rm phys}|\Pi_{\rm phys}\left(\ket{\phi(g)}\!\bra{\phi(g)}_R\otimes(a_{S,g}^{\rm phys})^\dag\right)|\psi_{\rm phys}}_{\rm phys}\nn\\
&=&\braket{\phi_{\rm kin}|\Pi_{\rm phys}\left(\ket{\phi(g)}\!\bra{\phi(g)}_R\otimes(a_{S,g}^{\rm phys})^\dag\right)|\psi_{\rm phys}}_{\rm kin}\nn\\
&=&\Big\langle\left(\ket{\phi(g)}\!\bra{\phi(g)}_R\otimes a_{S,g}^{\rm phys}\right)\phi_{\rm phys}\Big|\psi_{\rm phys}\Big\rangle_{\rm kin}\nn\\
&=&\Big\langle F_{a_{S,g}^{\rm phys},R}(g)\phi_{\rm phys}\Big|\psi_{\rm phys}\Big\rangle_{\rm phys}\,.\nn
\ea 
When proceeding from the second to the third line, we made use of the symmetry of $\Pi_{\rm phys}$ with respect to the kinematical inner product $\braket{\cdot|\cdot}_{\rm kin}$, which holds for unimodular groups, see Eq.~\eqref{piphyssym}. In the last line, we have invoked the definition of the physical inner product Eq.~\eqref{PIP} once more. Since $\ket{\psi_{\rm phys}},\ket{\phi_{\rm phys}}\in\ch_{\rm phys}$ were arbitrary, we conclude that 
\ba 
F_{(a_{S,g}^{\rm phys})^\dag,R}(g)\upharpoonright \ch_{\rm phys}&=& \left(F_{a_{S,g}^{\rm phys},R}(g)\upharpoonright \ch_{\rm phys}\right)^*\,.\nn
\ea 

The fact that the $*$-homomorphism is strong (and hence applies to the $G$-twirl expression of the relational observables) in the case that $U_R$ is the left regular representation of $G$ can be easily checked by noting that in this case $\braket{g|g'}=\delta(g,g')$ and $\Pi_S^{\rm phys}(g)=\mathbf{1}_S$. The proof of the strong homomorphism property with respect to multiplication and addition for this case can also be found in \cite{Bartlett:2007zz}.
\end{proof}

That is, the relational observables inherit the algebraic properties of the physical system observables in $\ca_{S,g}^{\rm phys}$, which we want to describe relative to the reference system $R$. By contrast, for arbitrary $B,C\in\ca_{\rm kin}$ it is \emph{not} true that $\cg(B\cdot C)\approx \cg(B)\cdot\cg(C)$ and so for general kinematical operators, the $G$-twirl does \emph{not} define an algebra homomorphism. The reason for this is that, as we shall see below, $\ch_{\rm phys}$ is isomorphic to $\ch_{S,g}^{\rm phys}$ and not to $\ch_{\rm kin}$.

\subsection{The physical inner product as a conditional inner product}\label{sec:phys_cond}

In Eq.~\eqref{PIP}, we have given the definition of the inner product on $\ch_{\rm phys}$ in terms of the one on $\ch_{\rm kin}$. Using our insights from the previous section on relational observables, we can now rewrite the physical inner product equivalently as a \emph{conditional inner product}. This will become useful later when generalizing the Page-Wootters formalism to general groups. The following is a generalization of \cite[Lemma 4.2]{Hoehn:2014fka} and \cite[Corollary 1]{Hoehn:2019owq} to non-Abelian groups.

\begin{corol}\label{cor_PIP}
Let $G$ be any unimodular Lie group. Then we can write the physical inner product in the equivalent forms
\ba
\braket{\psi_{\rm phys}|\phi_{\rm phys}}_{\rm phys}=\braket{\psi_{\rm kin}|\,\phi_{\rm phys}}_{\rm kin} = \braket{\psi_{\rm phys}|\left(\ket{\phi(g)}\!\bra{\phi(g)}_R\otimes \mathbf{1}_S\right)\,|\phi_{\rm phys}}_{\rm kin}\,,\q\forall\,g\in G\,.
\ea
\end{corol}

\begin{proof}
Lemmas~\ref{lem_G} and~\ref{lem_Rd} entail that $\Pi_{\rm phys}=\Pi_{\rm phys}\,\left(\ket{\phi(g)}\!\bra{\phi(g)}_R\otimes \mathbf{1}_S\right)\,\Pi_{\rm phys}$ for any $g\in G$ when $G$ is unimodular. Inserting this into the physical inner product in Eq.~\eqref{PIP} yields the claimed result.
\end{proof}

In the context of relational quantum dynamics and the Page-Wootters formalism (i.e.\ in the context of the translation group $G=(\mathbb{R},+)$), the 
conditional inner product $\braket{\psi_{\rm phys}|\left(\ket{\phi(g)}\!\bra{\phi(g)}_R\otimes \mathbf{1}_S\right)\,|\phi_{\rm phys}}_{\rm kin}$ was introduced in \cite{Smith:2017pwx,Smith:2019imm} and later shown to be equivalent to the physical inner product of constraint quantization \cite{Hoehn:2019owq,Hoehn:2020epv}.\footnote{For periodic clocks, i.e.\ $G= \rm{U}(1)$, the equivalence does not hold when the group action on the coherent states features an infinite isotropy group such as $(\mathbb{Z},+)$ \cite{periodic}.} The conditional inner product also appeared in \cite{Hoehn:2014fka} for lattice field theory models.

The physical inner product written in the form Eq.~\eqref{PIP} is manifestly gauge-invariant. By contrast, we can interpret the conditional inner product as a \emph{gauge-fixed} description of the manifestly gauge-invariant physical inner product. In line with this interpretation, the corollary tells us that (for unimodular groups) this gauge-fixed expression is actually independent of the gauge-fixing condition that the frame $R$ is in orientation $g\in G$. These observations will be crucial below in Sec.~\ref{ssec_condprob}, when establishing the conditional probability interpretation of relational observables.

\subsection{Left-right systems of coherent states}\label{ssec_LR}

We noted above that for LR systems of coherent states, the physical system Hilbert space will be frame-orientation-independent. But which coherent states systems are of the LR type? This question will also be relevant for the discussion of symmetries (as opposed to gauge transformations) in the next subsection.

\begin{definition}[Left-right system of coherent states]\label{def_LR_coherent}
A left-right (or LR) system of coherent states for a group $G$ is a (left) system of coherent states $\{U_R,\ket{\phi(g)_R}\}$ with (possibly trivial) isotropy subgroup $H$ where the action $\ket{\phi(g')} \mapsto \ket{\phi(g' g^{-1})}$ is unitary; namely there exists a unitary representation $V_R$ of $G$ which acts on the coherent states as $V_R(g)\ket{\phi(g')}=\ket{\phi(g' g^{-1})}$. We write $\{U_R,V_R,\ket{\phi(g)_R}\}$ for such a left-right system of coherent states.
\end{definition}

Before we analyze LR systems in more detail, let us begin with a clarification. Consider an LR system of coherent states $\{U_R,V_R,\ket{\phi(g)_R}\}$ for a group $G$ with a non-trivial isotropy subgroup $H$. As demonstrated in Lemma~\ref{lemCompact}, we can then turn to a system of states $\{\ket{\phi(x)}_R\}_{x\in X}$ on the cosets $X\simeq G/H$. The symmetry group acts as $U_R(g)\ket{\phi(x)}_R=e^{i\varphi}\ket{\phi(gx)}_R$ for some phase $\varphi\in\mathbb{R}$. It might be tempting to demand, by definition, that LR systems of coherent states should now also satisfy $V_R(g)\ket{\phi(x)}_R=e^{i\theta}\ket{\phi(xg^{-1})}_R$ for some phase $\theta\in\mathbb{R}$. But for this to make sense, we would need that $xg^{-1}$ is a valid coset for every coset $x$, i.e.\ that for every $g\in G$ there is some $g'\in G$ such that $Hg^{-1}=g'H$. Here, we will \emph{not} make this assumption. We will thus demand validity of
\begin{equation}
U_R(g)V_R(k)\ket{\phi(h)}_R = \ket{\phi(ghk^{-1})} \quad\mbox{for all }g,k\in G
\label{eqDoubleAction}
\end{equation}
only for group elements $h\in G$, and will not demand an analogous equation (even with phases) for cosets.

Let us now classify the LR systems of coherent states for compact Lie groups. Due to Eq.~(\ref{eqDoubleAction}), and since $\{\ket{\phi(g)}_R\}_{g\in G}$ spans the carrier space $\ch_R$, $[U_R(g),V_R(k)] = 0$ for all $g,k \in G$. Hence $\ch_R$ carries a representation $W_R$ of $G \times G$, where $W_R(g,e) = U_R(g)$ and $W_R(e,g) = V_R(g)$. $W_R$ acts on coherent states as $W_R(g,k)\ket{\phi(h)}_R = \ket{\phi(ghk^{-1})}_R$. This allows us to obtain the following classification.

\begin{lemma}
\label{LemLR}
Let $G$ be a compact Lie group, and consider a corresponding LR system of coherent states $\{U_R,V_R,\ket{\phi(g)}_R\}$ on $\ch_R$. Then this space decomposes as
\begin{align}
    \ch_R \simeq \bigoplus_{i \in \ci}  \ch_i \otimes \bar \ch_i  \ ,
\end{align}
such that the representation $W_R(g,k)=U_R(g)V_R(k)$ of $G\times G$ decomposes into irreducible representations
\begin{align}\label{eq-LR-rep}
    W_R(g,k) \simeq \bigoplus_{i \in \ci} \rho_i(g) \otimes \bar \rho_i(k)   \ ,
\end{align}
where $\bar \rho_i$ is the complex conjugate representation of $\rho_i$ and there are no repeated irreducible representations in $\ci$. The unitary representations $U_R$ and $V_R$ decompose as
\begin{align}
    U_R(g)\simeq \bigoplus_{i \in \ci} \rho_i(g) \otimes \I_{\bar \ch_i}
\end{align}
and 
\begin{align}
    V_R(g)\simeq \bigoplus_{i \in \ci} \I_{\ch_i} \otimes  \bar \rho_i(g).
\end{align}
The reference coherent state is of the form:
\begin{align}\label{eq-LR-seed}
    \ket{\phi(e)}_R = \sum_{i \in \ci} \alpha_i \ket{\phi(e)}_i  \ ,
\end{align}
where $\ket{\phi(e)}_i:=\frac{1}{\sqrt{\dim(\ch_i)}} \sum_{j = 1}^{\dim(\ch_i)} \ket{e_j}_i \otimes \ket{e_j}_i$ is just the maximally entangled state across $\ch_i \otimes \bar \ch_i$. 

Conversely, given any representation $W_R$ of $G \times G$ on $\ch$ of the form of Eq.~\eqref{eq-LR-rep}, there always exists a normalized vector $\ket{\phi(e)} \in \ch$ which will be a valid seed state. It is of the form Eq.~\eqref{eq-LR-seed}, where the requirement
\begin{align}
    \int_{g \in G} \ketbra{\phi(g)}{\phi(g)}_R dg = n \mathbf{1} \ ,
\end{align}
constrains the possible values of $\{\alpha_i\}_{i \in \ci}$ (for an exact description of these constraints, see Example~\ref{Ex1Compact}).
\end{lemma}

This lemma is proven in full in Appendix~\ref{app:LR_coherent} and we sketch a proof here. A LR system of coherent states with irreducible $W_R(g,g')$ acts on $\ch_i \otimes \ch_j$ with $\ch_i$ and $\ch_j$ both irreducible.  We now restrict to the representation $W_R(g,g)$ of $G$ and make use of the isomorphisms $\ch_i \otimes \ch_j \cong \ch_i \otimes \bar \ch_j^* \cong \Hom(\bar \ch_j, \ch_i)$ with $\Hom(\bar \ch_j, \ch_i)$  the space of linear maps from $\bar \ch_j$ to $\ch_i$. The isomorphism maps the state $\ket{\phi(e)}_R \in \ch_i \otimes \ch_j$, which is invariant under $\rho_i(g) \otimes \rho_j(g)$, to an element $M_{\phi(e)} \in \Hom(\bar \ch_j, \ch_i)$  which is necessarily invariant under the action of $g$. This shows that the $G$-equivariant subspace $\Hom_G(\bar \ch_j, \ch_i)$ of $\Hom(\bar \ch_j, \ch_i)$ is non trivial, where the $G$-equivariant subspace is the space of all maps $M \in \Hom(\bar \ch_j, \ch_i)$ such that $\rho_i(g) M = M \bar \rho_j(g)$ for all $g \in G$. Due to the irreducibility of the representations involved, Schur's lemma entails that $\ch_j \cong \bar \ch_i$. For an irrep $\ch_i \otimes  \ch_j \cong \ch_i \otimes \bar \ch_i$  the specific form of $\ket{\phi(e)}_i$ follows directly from the fact that $M_{\phi(e)} \in \Hom(\bar \ch_j, \ch_i) \cong \Hom(\ch_i, \ch_i)$ is proportional to the identity matrix and applying the explicit isomorphism  $\Hom( \ch_i, \ch_i) \to \ch_i \otimes \bar \ch_i$ yields a state $\ket{\phi(e)}_i$ of the form given in Equation~\eqref{eq-LR-seed}.  The extension to reducible representations requires Example~\ref{Ex1Compact}.

We emphasise that the above lemma classifies systems of LR states \emph{up to isomorphism}. For instance the singlet state $\frac{1}{\sqrt{2}} (\ket{01} - \ket{10})$ generates a system of LR coherent states under the action of $U \otimes V$ with $U,V \in \SU(2)$, even though it is not explicitly of the form $\rho(g) \otimes \bar \rho(g)$ acting on a state $\frac{1}{\sqrt{2}} \sum_i e_i \otimes e_i$. However since these are self dual representations, namely $U \cong \bar U$ we have that the representation $U \otimes V \cong U \otimes \bar V$. Under this isomorphism the seed state $\frac{1}{\sqrt{2}}(\ket{01} - \ket{10})$  is mapped to $\frac{1}{\sqrt{2}}(\ket{00} + \ket{11})$ in keeping with the above lemma.

It is interesting to note that the possible seed states of an LR coherent state system decompose into the so-called \emph{fusion} or \emph{entangling products} of the irrep carrier spaces $\ch_i$ and $\bar\ch_i$. The fusion product of $\ch_i\otimes\bar\ch_i$ for compact $G$ is defined to be the subspace of it that is invariant under $U_R(g)\otimes V_R(g)$, for all $g\in G$. The fusion product can clearly be constructed using a coherent group averaging projector. We see that $\ket{\phi(e)}$ is invariant under the action of $U_R(g)\otimes V_R(g)$, for all $g\in G$ and so the same conclusion holds for the maximally entangled state $\ket{\phi(e)}_i$ across $\ch_i\otimes\bar\ch_i$. From the isomorphism $\ch_i \otimes \bar\ch_i \cong \Hom(\bar \ch_i^*, \ch_i) \cong \Hom(\ch_i, \ch_i)$ one can apply Schur's lemma to show that the G-equivariant subspace $\Hom_G(\ch_i, \ch_i)$  of  $\Hom(\ch_i, \ch_i)$ is one-dimensional. This implies (using the same isomorphism in the opposite direction) that the $G$-invariant subspace (i.e.\ the fusion product) of $\ch_i\otimes\bar\ch_i$ is the one-dimensional subspace  $\ket{\phi(e)}_i$. This shows that the seed state of an LR coherent state system is (up to the coefficients $\alpha_i$) just the direct sum of the fusion products of the irrep carrier spaces $\ch_i$ and $\bar\ch_i$.  Fusion products appear when fusing neighbouring subregions in gauge theories across their joint interface \cite{Donnelly:2016auv,Donnelly:2011hn,Donnelly:2014fua,Geiller:2019bti}.

\subsection{Symmetries as frame reorientations versus gauge transformations}\label{ssec_symgauge}

In gauge field theories it is standard to distinguish between gauge transformations and symmetries. Gauge transformations change the description of the physical state, but not the state itself, while symmetries change the physical situation and thus the physical state. As such, symmetries have to commute with gauge transformations. Gauge transformations are generated by the constraints of the theory, while symmetries are generated by so-called \emph{charges}, which are specific gauge-invariant observables. Typically, gauge transformations have to vanish asymptotically in field theory, while symmetries may act non-trivially there, changing physical asymptotic data that also influences the gauge-invariant data in the bulk. Extending the distinction between gauge transformations and symmetries to \emph{finite} subregions in gauge theories necessitates the introduction of so-called \emph{edge modes} \cite{Donnelly:2016auv}. From the point of view of the subregion, edge modes are additional degrees of freedom that populate its boundary and encode the fact that a finite boundary breaks \emph{a priori} gauge-invariance (e.g.\ of cross-boundary data such as cross-boundary Wilson loops). They turn out to be crucial for constructing a well-defined phase or Hilbert space for the subregion of interest \cite{Donnelly:2016auv,Donnelly:2011hn,Donnelly:2014fua,Donnelly:2014gva,Geiller:2019bti,Carrozza:2021sbk}. As shown in \cite{Donnelly:2016auv} for classical non-Abelian theories, symmetries and gauge-transformations act on ``opposite'' sides of the edge mode, e.g.\ gauge transformations by left and symmetries by right action. 

Of relevance for our discussion is that edge modes have recently been identified, at the classical level, as dynamical reference frames for the local gauge group $G$ \cite{Carrozza:2021sbk}. They are thus reference frames in the same sense in which the quantum reference frames here are reference systems for the group $G$ of interest; in particular, upon quantization, edge modes are nothing but field theory versions of quantum reference frames. Indeed, one can also construct mechanical analogs of edge modes, e.g.\ for a subgroup of particles subject to translation invariance, mimicking a gauge field in a finite subregion of spacetime \cite{Carrozza:2021sbk}. In either field theory or mechanical toy models, edge modes can be viewed as ``internalized'' external reference frames for the subregion or subgroup of particles of interest; they are dynamical degrees of freedom originating in the complement of the subregion or subgroup (e.g.\ an external particle). While there is no unique edge mode frame, relational observables relative to a choice of edge mode hence encode how the subregion or subgroup of interest relates to its complement. Importantly, symmetries were identified in \cite{Carrozza:2021sbk} as reorientations of the edge mode frame in both the field theoretic and mechanical context; as they only act on the frame (gauge transformations act on all degrees of freedom), they thus change the physical situation by changing the relation between frame and remaining degrees of freedom. It was further shown that edge frame changes at the level of relational observables are particular relation-conditional symmetries. We have seen the analogous situation in special relativity in Sec.~\ref{ssec_SR}, distinguishing between Lorentz ``gauge'' and ``symmetry'' transformations, acting on the observer's tetrad $e^\mu_A$, where the symmetries indeed amount to reorientations of the frame.

The aim of the present subsection is to import this identification of symmetries as frame reorientations and their distinction from gauge transformations into the quantum theory and thus the context of quantum reference frames. In Sec.\ \ref{sec_relobschanges}, we shall then demonstrate how symmetries that depend on the relation between the old and new frame give rise to the transformations from the relational observables relative to the first frame to those relative to the second. 

While gauge transformations act on \emph{all} degrees of freedom\,---\,here both $R$ and $S$\,---\,, symmetries only act on the frame\,---\,in our case $R$. Recalling that symmetries have to act on the ``opposite'' side of the frame from gauge transformations and that we chose the latter to act by left multiplication, since we can already anticipate that symmetries must act by right multiplication\footnote{We emphasize that although we call the group action $x \mapsto x g^{-1}$ an action from the right, it is technically a left group action. This is because $x \mapsto x (gh)^{-1} = (xh^{-1})g^{-1}$ and hence the element $(gh)$ acts on $x$ first by the action of $h$ then $g$, which is the definition of a left group action.} on the frame~\cite{superselection_kitaev_2004}. However, we can also understand more thoroughly why the left multiplication by the group does not give rise to physical reorientations of the frame $R$. Indeed, for non-Abelian $G$, the left transformation $U_R(g) \otimes \I_S$ of the frame does not commute with arbitrary gauge transformations $U_R(g')\otimes U_S(g')$ for any $g\in G$ and will thus also not in general commute with the coherent group averaging `projector' $\Pi_{\rm phys}=\int_Gdg\,U_{RS}(g)$. A left transformation $U_R(g) \otimes \I_S$ will thus map a physical state $\ket{\psi_{\rm phys}}$ outside the physical Hilbert space $\ch_{\rm phys}$, in violation of gauge-invariance, unless $g$ is contained in the center of $G$. Therefore, $U_R(g) \otimes \I_S$ does not in general correspond to a physical transformation (reorientation) of the frame $R$.

\begin{figure}[t]
\centering
\includegraphics[trim=0 100 0 50,clip,width= 
400pt]{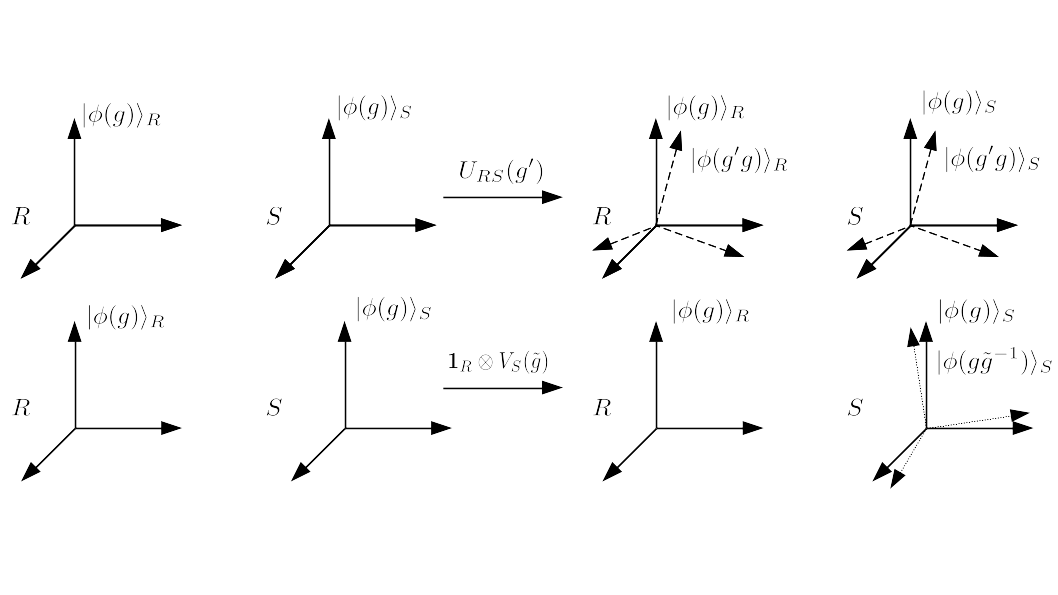}
\caption{The set of orthonormal axes in $\Rl^3$ is acted on regularly by $G \cong \SO(3)$. In the top, we see a reference $R$ and a system $S$ which are acted on by the left regular representation of $G$. This is a gauge transformation, and the relative orientation between $R$ and $S$ is preserved. On the bottom, we see that $R$ is acted on by the right regular action. This commutes with the gauge action $U_{RS}(g')$ and hence is a physical symmetry. The relative orientation between $R$ and $S$ is changed by this action. \label{fig:Left_and_right}}
\end{figure}

Suppose, therefore, that we are given an LR system $\{U_R,V_R,\ket{\phi(g)}_R\}$ of coherent states for the group $G$ and the frame $R$, according to Definition~\ref{def_LR_coherent}. The representation $V_R$, corresponding to the action from the right on the coherent states, does commute with arbitrary gauge transformations
\ba
[V_R(g)\otimes\I_S,U_{RS}(g')]=0\,,\q\q\forall\,g,g'\in G\,,
\ea
and so also with $\Pi_{\rm phys}$. Right transformations of the frame $R$ thus \emph{always} preserve the physical Hilbert space $\ch_{\rm phys}$. This means that the unitary $V_{R}(g)\otimes\I_S$ can be interpreted as a physical reorientation of the frame $R$ by the group element $g$, see Fig.~\ref{fig:Left_and_right} for an illustration in the case of a regular group action. In particular, the generators of the group action from the right $V_R$ will commute with all gauge transformations and thus constitute gauge-invariant observables of the theory; the Lie algebra of the action from the right of $G$ comprises the ``charges'' in $\ca_{\rm phys}$. For example, in models of relational dynamics \cite{pageEvolutionEvolutionDynamics1983,Hoehn:2019owq,Hoehn:2020epv}, the Hamiltonian of the clock would be such a charge, while in particle models subject to translation invariance \cite{Vanrietvelde:2018pgb,Vanrietvelde:2018dit,Carrozza:2021sbk} the frame particle's momentum would constitute the charge. By contrast, the generators of the gauge transformations $U_{RS}(g)$ constitute the constraints of the theory.

That $\cv_R^{\rm phys}(g):=V_R(g)\otimes\mathbf{1}_S\bullet V_R(g)^\dagger\otimes\mathbf{1}_S$ corresponds to physical frame reorientations can also be inferred from looking at its action on the relational observables relative to $R$ given in Eq.~\eqref{qrelobs}:
\ba 
\cv^\phys_R(g) \left(\strut F_{f_S,R}(g_1)\right) &=& (V_R(g)\otimes\mathbf{1}_S)\int_G dg'\, U_{RS}(g')\left(\strut \ket{\phi(g_1)}\!\bra{\phi(g_1)}_R\otimes f_S\right) U_{RS}(g')^\dagger (V_R(g)^\dagger \otimes\mathbf{1}_S) \nn\\
&=& F_{f_S,R}(g_1g^{-1})\,,\label{obsreorient}
\ea 
where we have used that $[U_{RS}(g'),V_R(g)\otimes\mathbf{1}_S]=0$ for all $g,g'\in G$. Note that $\cv_R^\phys$ is a unitary representation of $G$ on $\ca_\phys$. Hence, the frame orientation label is translated by $g^{-1}$ on the right. This is precisely the quantum analog of the classical situation in \cite[Eq.\ (93)]{Carrozza:2021sbk}. In particular, recall from Sec.~\ref{ssec_relobs} that the relational observable family $F_{f_S,R}(g_1)$ defines an orbit in $\ca_{\rm phys}$ of dimension $\dim\, X$ or less, where $X=G/H$ is the coset space. The frame reorientations generated by $V_R$ therefore constitute the transformations \emph{along} that orbit in the algebra, mapping the relational observable from the point $g_1$ on the orbit to the point $g_1g^{-1}$, see Fig.~\ref{fig:orbits}. Since the orbits are in general not of the same dimension (e.g., for an already invariant $f_S$ it is a point), they rather define a stratification of the physical observable algebra $\ca_{\rm phys}$.

\begin{figure}[t]
\centering
\includegraphics[width= 400pt]{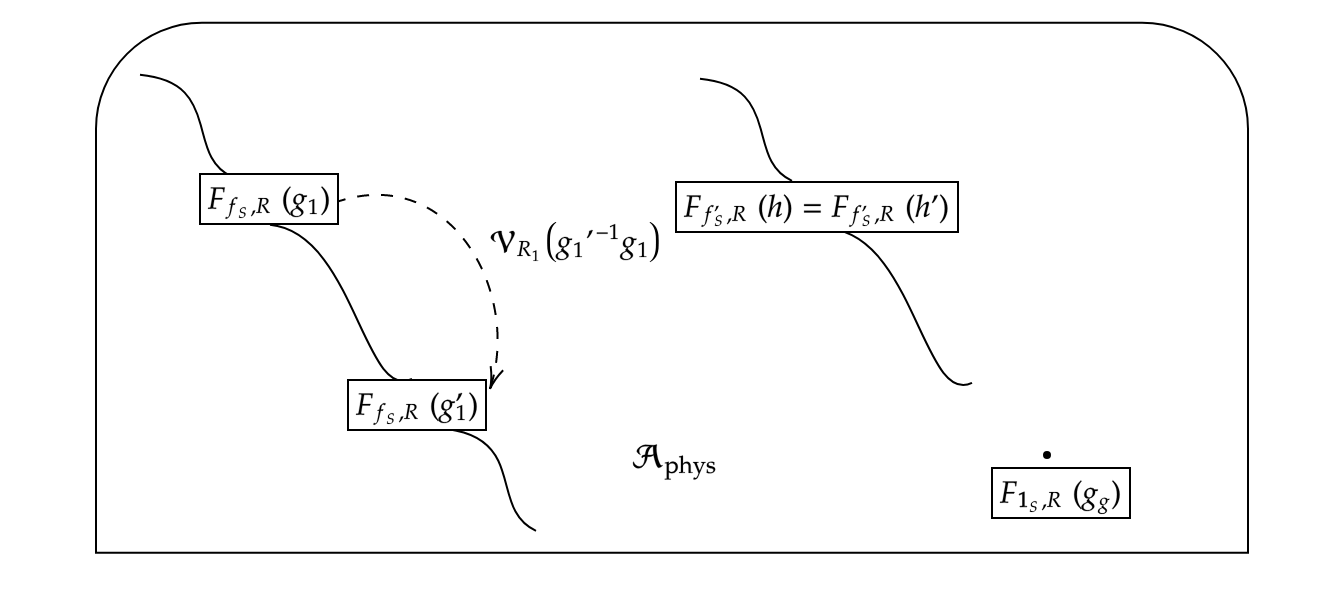}
\caption{This figure illustrates $G$-orbits of different relational observables relative to $R$ under $\cv_R$ in the space $\ca_\phys$. The orbit  $\{F_{\mathbf{1}_S,R}(g)\}_{g \in G}$ is a single point. The orbit $\{F_{f'_S,R}(g)\}_{g \in G}$   is the right line. We see that two observables $F_{f'_S,R}(h)$ and $F_{f'_S,R}(h')$ where $h,h'$ belong to the same right coset are equal. The orbit $\{F_{f_S,R}(g)\}_{g \in G}$ is represented by the left line.  One can transform between  $F_{f_S,R}(g_1)$ and $F_{f_S,R}(g_1')$ on this orbit with the operator $\cv_R^{\phys}({g_1'}^{-1} g_1)$ as represented by the dotted arrow. Different orbits generally being of different dimension, the relational observable families define a stratification of $\ca_{\rm phys}$. \label{fig:orbits}}
\end{figure}

We can thus conclude that $V_R$ defines the representation of the symmetry group (associated with $R$), whilst the global representation $U_{RS}$ defines the action of the gauge group; note, however, that abstractly both the symmetry and gauge group are $G$.
In other words, if the coherent state system which we employ to model the orientations of the quantum reference frame does not admit a unitary  action from the right, then the ensuing representation will not admit any symmetries in the above sense of frame reorientations. Recall that in that case the physical system Hilbert space and observable algebra become frame orientation dependent. In the next section, we shall elucidate why these two properties are intertwined.

\section{Jumping into a  QRF perspective I: extending the Page-Wootters formalism}\label{sec:gen_PW}

As shown in \cite{Hoehn:2019owq,Hoehn:2020epv,Hoehn:2021flk,periodic} for Abelian groups, there exist two unitarily equivalent ways to reduce from the perspective-neutral Hilbert space $\ch_{\rm phys}$ to an internal QRF perspective. In this section, we build up on these works to extend the  Page-Wootters formalism \cite{pageEvolutionEvolutionDynamics1983,pageClockTimeEntropy1994,giovannettiQuantumTime2015,Smith:2017pwx,Smith:2019imm} to general groups, yielding a `relational Schr\"odinger picture' in which states, but not observables of $S$ depend on the frame orientation $g$.\footnote{Apart from the fact that $\ca_{S,g}^{\rm phys}$ may depend on $g$.} In the next section, we will discuss the second method, yielding a `relational Heisenberg picture' for internal frame perspectives.

\subsection{Generalized Page-Wootters reduction of states and observables}

We recall the physical system Hilbert space $\ch_{S,g}^{\rm phys}\subset\ch_S$ from Sec.~\ref{ssec_relobs} and define the \emph{Schr\"odinger reduction map} $\calr_{\mathbf{S}}(g):\ch_{\rm phys}\rightarrow\ch_{S,g}^{\rm phys}$ from perspective-neutral states to the description of $S$ relative to $R$ being in orientation $g\in G$ by the conditioning operation 
\ba 
\calr_{\mathbf{S}}(g)=\bra{\phi(g)}_R\otimes\mathbf{1}_S\,.\label{Sredmap}
\ea 
The bold label $\mathbf{S}$ stands here for `Schr\"odinger picture', rather than the system $S$. 
We write the conditional image states as
\ba
\ket{\psi_S^{\rm phys}(g)} := \calr_{\mathbf{S}}(g)\,\ket{\psi_{\rm phys}} \label{condstate3}
\ea
since, in analogy to the Schr\"odinger equation of the Page-Wootters formalism, they satisfy the covariance property under  (gauge transformation induced) changes of the frame orientation label:
\ba
U_S(g')\,\ket{\psi_S^{\rm phys}(g)} &=& \braket{\phi(g)|_R\,\otimes U_S(g')\,|\psi_{\rm phys}} = \braket{\phi(g)|_R\,U_R^\dag(g')\otimes \mathbf{1}_S\,|\psi_{\rm phys}}=\calr_{\mathbf{S}}(g'g)\ket{\psi_{\rm phys}}\nn\\
&=&\ket{\psi_S^{\rm phys}(g'g)}\,,\label{condcov}
\ea
where we have made use of Eq.~\eqref{stateinv} in the second equality. To see that the conditional states indeed lie in the physical system Hilbert space $\ch_{S,g}^{\rm phys}$, we note that using the resolution of the identity \eqref{resid}, we can rewrite physical states as
\ba
\ket{\psi_{\rm phys}}=\int_G dg'\ket{\phi(g')}_R\otimes\ket{\psi_S^{\rm phys}(g')}\,,\label{psuse}
\ea
so that for a unimodular group $G$ a redefinition of the integration measure yields
\ba 
\ket{\psi_S^{\rm phys}(g)}=\calr_{\mathbf{S}}(g)\ket{\psi_{\rm phys}}=\int_G dg'\braket{\phi(g)|\phi(g'g)}\,U_S(g')\ket{\psi_S^{\rm phys}(g)}\underset{\eqref{Sproj}}{=}\Pi_{S}^{\rm phys}(g)\,\ket{\psi_S^{\rm phys}(g)}\,.\nn
\ea 
Lemma~\ref{lem_Sproj} tells us that, as we change the frame orientation according to Eq.~\eqref{condcov}, the physical system Hilbert space $\ch_{S,g}^{\rm phys}$ in which the reduced state lives, changes inside the kinematical system Hilbert space $\ch_S$, unless condition \eqref{projgind} is satisfied. We will explain this state of affairs from the point of view of symmetries shortly. 

However, first recall that Lemma~\ref{lem_Sproj} also tells us that the isotropy group $H$ of the frame orientation states is an isotropy group of the physical system Hilbert space $\ch_{S,g}^{\rm phys}$, i.e.\ stabilizing it inside $\ch_S$. Indeed, suppose $g'=ghg^{-1}$ with $h\in H$ (so that $g'$ is in the isotropy group of $\ket{\phi(g)}_R$) in Eq.~\eqref{condcov}, then we clearly have 
\ba
U_S(g')\,\ket{\psi_S^{\rm phys}(g)}=e^{i\varphi(h)}\,\ket{\psi_S^{\rm phys}(g)}\,,\label{Sstateisotropy}
\ea 
where $\varphi(h)$ is some phase. Hence, up to phase, physical system states are $H$-invariant, just like the system observables inside the relational observables relative to $R$ are effectively $H$-invariant, see Sec.~\ref{ssec_relobs}. More precisely, we can now also see why $f_S,f_S'$ with $f_S\sim f_S'$ in Eq.~\eqref{HfS} yield indistinguishable physical system observables. Noting that (\ref{Sstateisotropy}) implies $U_S(h)\Pi_S^{\rm phys}(e)=e^{i\varphi(h)}\Pi_S^{\rm phys}(e)$, and recalling from Lemma~\ref{lem_Sproj} that $[U_S(h),\Pi_S^{\rm phys}(e)]=0$, we have 
\ba 
\Pi_S^{\rm phys}(e)\,f_S\,\Pi_S^{\rm phys}(e)&=&\f{1}{{\rm Vol}(H)}\,\int_H d_H h\, U_S(h)\,\Pi_S^{\rm phys}(e)\, f_S\, \Pi_S^{\rm phys}(e)\,U^\dag_S(h)
\nn\\
&=&\f{1}{{\rm Vol}(H)}\,\Pi_S^{\rm phys}(e)\int_H d_H h\, U_S(h)\,f_S\, U^\dag_S(h)\, \Pi_S^{\rm phys}(e)\nn\\
&=&\Pi_S^{\rm phys}(e)\,f'_S\,\Pi_S^{\rm phys}(e)\,,\label{isotropyobs}
\ea 
where in arriving at the last line, we have invoked Eq.~\eqref{HfS}, before reversing the steps. Thus, $f_S$ and $f_S'$ project to the \emph{same} element $f_{S,e}^{\rm phys}\in\ca_{S,e}^{\rm phys}$. In other words, \emph{the frame $R$ can only resolve properties of $S$ that are also invariant under the frame orientation isotropy group}. 

Similarly, Lemma~\ref{lem_Sproj} entails that the center of $G$ is also a subgroup of the isotropy group of $\ch_{S,g}^{\rm phys}$ inside $\ch_S$, however one that now transforms physical system states non-trivially inside it, i.e.\ not just by a phase, and is therefore changing the physical situation.

Owing to the redundancy in the description of $\ch_{\rm phys}$, the Schr\"odinger reduction is invertible (this is not the case when applied to the full kinematical Hilbert space $\ch_{\rm kin}$). Due to the resolution of the identity~\eqref{resid}, this holds regardless of whether the frame isotropy group $H$ is trivial or not:
\begin{lemma} \label{lem_invred}
For unimodular groups $G$, the Schr\"odinger reduction maps are invertible, with the inverse $\calr^{-1}_{\mathbf{S}}(g):\ch_{S,g}^{\rm phys}\rightarrow\ch_{\rm phys}$ given by
\ba 
\calr_{\mathbf{S}}^{-1}(g):=\Pi_{\rm phys}\,\left(\ket{\phi(g)}_R\otimes \mathbf{1}_S\right)=\int_G dg'\ket{\phi(g'g)}_R\otimes U_S(g')\,.
\ea
\end{lemma}

\begin{proof}
The right inverse property is immediate
\ba 
\calr_{\mathbf{S}}(g)\cdot\calr_{\mathbf{S}}^{-1}(g)=\left(\bra{\phi(g)}_R\otimes\mathbf{1}_S\right)\Pi_{\rm phys}\left(\ket{\phi(g)}_R\otimes\mathbf{1}_S\right)=\Pi_S^{\rm phys}(g)\,.\nn
\ea 
The left inverse property follows from
\ba
\calr_{\mathbf{S}}^{-1}(g)\cdot\calr_{\mathbf{S}}(g)=\Pi_{\rm phys}\left(\ket{\phi(g)}\!\bra{\phi(g)}_R\otimes \mathbf{1}_S\right)\approx\mathbf{1}_{RS}\,,\nn
\ea
where the last weak equality is entailed by Lemmas~\ref{lem_G} and~\ref{lem_Rd}.
\end{proof}

It is now also easy to see from the perspective of symmetries why $\ch_{S,g}^{\rm phys}$ `rotates' (resp.\ does not `rotate') inside $\ch_S$ when the frame does not (resp.\ does) admit symmetries. We have $\ch_{S,g}^{\rm phys}=\calr_{\mathbf{S}}(g)\big[\ch_{\rm phys}\big]$ and $\ch_{S,g'g}^{\rm phys}=\calr_{\mathbf{S}}(g'g)\big[\ch_{\rm phys}\big]$.
When $R$ admits symmetries, we can write 
\ba
\calr_\mathbf{S}(g'g)=\calr_\mathbf{S}\left(g(g^{-1}g'g)\right)=\calr_\mathbf{S}(g)\left(V_R(g^{-1}g'g)\otimes\I_S\right)\,.
\ea 
Now recall from Sec.~\ref{ssec_symgauge} that $V_{R}(g^{-1}g'g)\otimes\I_S$ is a \emph{physical} transformation that commutes with all gauge transformations and thus leaves $\ch_{\rm phys}$ invariant. Hence, 
\ba 
\ch_{S,g'g}^{\rm phys}=\calr_{\mathbf{S}}(g)\big[\left(V_R(g^{-1}g'g)\otimes\I_S\right)\ch_{\rm phys}\big]=\calr_{\mathbf{S}}(g)\big[\ch_{\rm phys}\big]=\ch_{S,g}^{\rm phys}\,.\label{hsggprime}
\ea 
By contrast, when $R$ does not admit any symmetries, we can only write
\ba 
\calr_\mathbf{S}(g'g)=\calr_\mathbf{S}(g)\left(U_R^\dag(g')\otimes\I_S\right)\,.\label{nonsym}
\ea 
If $g'$ is in the isotropy group of $\ket{\phi(g)}_R$, then $\calr_\mathbf{S}(g'g)=e^{i\varphi(h)}\calr_\mathbf{S}(g)$ by definition of $\calr_\mathbf{S}(g'g)$. When $g'$ is contained in the center of $G$, then $U_R^\dag(g')\otimes\I_S$ commutes with all gauge transformations and so is a physical transformation that leaves $\ch_{\rm phys}$ invariant. In both cases, reasoning as in Eq.~\eqref{hsggprime}, we conclude that $\ch_{S,g'g}^{\rm phys}=\ch_{S,g}^{\rm phys}$. However, if $g'$ is now any other element from a coset $x\in X$, then, as explained in Sec.~\ref{ssec_symgauge}, $U_R^\dag(g')\otimes\I_S$ is a non-gauge-invariant transformation that specifically does \emph{not} in general preserve $\ch_{\rm phys}$. In this case, we will in general have the following inequality:
\ba 
\ch_{S,g'g}^{\rm phys}=\calr_{\mathbf{S}}(g)\big[\left(U^\dag_R(g')\otimes\I_S\right)\ch_{\rm phys}\big]\neq\calr_{\mathbf{S}}(g)\big[\ch_{\rm phys}\big]=\ch_{S,g}^{\rm phys}\,.
\ea 

The invertibility of the Schr\"odinger reduction map allows us to define an encoding map ${\mathcal{E}_{\mathbf{S}}^g:\ca_{S,g}^{\rm phys}\rightarrow\ca_{\rm phys}}$, embedding physical system operators into Dirac observables on $\ch_{\rm phys}$ via
\ba 
\mathcal{E}_\mathbf{S}^g\left(f_{S,g}^{\rm phys}\right):=\calr_\mathbf{S}^{-1}(g)\,f_{S,g}^{\rm phys}\,\calr_{\mathbf{S}}(g)=\Pi_{\rm phys}\left(\ket{\phi(g)}\!\bra{\phi(g)}_R\otimes f_{S,g}^{\rm phys}\right)\,.\label{Sencode}
\ea 
This coincides (weakly) with the relational Dirac observables. More precisely, we have the following extension to general unimodular groups of corresponding results in \cite{Hoehn:2019owq,Hoehn:2020epv,periodic,Hoehn:2021flk}:
\begin{theorem}\label{thm_obs}
Let $f_S\in\ca_S$ be arbitrary. Reducing the relational observable $F_{f_S,R}(g)$ with the Schr\"odinger reduction map yields the corresponding physical system observable $f_{S,g}^{\rm phys}$ given in Eq.~\eqref{fsys}:
\ba
\calr_\mathbf{S}(g)\,F_{f_{S},R}(g)\,\calr_\mathbf{S}^{-1}(g) =\Pi_S^{\rm phys}(g)\,f_S\,\Pi_S^{\rm phys}(g)= f_{S,g}^{\rm phys}\,.
\ea
Conversely, let $f_{S,g}^{\rm phys}\in\ca_{S,g}^{\rm phys}$ be arbitrary. Then it embeds via the encoding map into $\ca_{\rm p hys}$ as its corresponding relational observable
\ba
\mathcal{E}_\mathbf{S}^g\left(f_{S,g}^{\rm phys}\right)\approx F_{f_{S,g}^{\rm phys},R}(g)\,.\nn
\ea
\end{theorem}

\begin{proof}
The first statement follows from
\ba
\calr_\mathbf{S}(g)\,F_{f_{S},R}(g)\,\calr_\mathbf{S}^{-1}(g) &=& \left(\bra{\phi(g)}_R\otimes \mathbf{1}_S\right)\Pi_{\rm phys}\left(\ket{\phi(g)}\!\bra{\phi(g)}_R\otimes f_{S}\right)\,\Pi_{\rm phys}\,\left(\ket{\phi(g)}_R\otimes \mathbf{1}_S\right)\nn\\
&\underset{\eqref{Sproj}}{=}&\Pi_S^{\rm phys}(g)\,f_S\,\Pi_S^{\rm phys}(g)\,,\nn
\ea
while its converse is an immediate consequence of Eqs.~\eqref{qrelobs} and~\eqref{Sencode}.
\end{proof}
In conjunction with Theorem~\ref{lem_homo} and Lemma~\ref{lem_weakequiv}, this result establishes an equivalence between the observable algebras $\ca_{S,g}^{\rm phys}$ and $\ca_{\rm phys}$. In more detail, Corollary~\ref{CorIsometry} below shows that $\calr_\mathbf{S}(g)$ is an isometry, hence conjugation with it defines a $*$-isomorphism between $\ca_{S,g}^{\rm phys}$ and $\ca_{\rm phys}$.

Importantly, note that the first statement of Theorem~\ref{thm_obs} is the quantum analog of our observation in Sec.~\ref{ssec_SR}, case (A), that ``relational observables'' in special relativity look simple in the internal perspective of the frame relative to which they are formulated; in fact, their description coincides with the non-invariant quantity that is supposed to be described relative to the frame.

\subsection{Conditional probability interpretation of relational observables}\label{ssec_condprob}

The original Page-Wootters formalism for quantum clocks \cite{pageEvolutionEvolutionDynamics1983,pageTimeInaccessibleObservable1989,pageClockTimeEntropy1994,Smith:2017pwx,Smith:2019imm,giovannettiQuantumTime2015} provides a conditional probability interpretation for the relational quantum dynamics encoded in the physical states solving a Hamiltonian constraint. This conditional probability interpretation was subsequently criticised by Kucha\v{r} in \cite{kucharTimeInterpretationsQuantum2011a} as being inconsistent with the constraints of the theory, dynamics and relativistic covariance. Using the equivalence between the formulation of the relational dynamics in terms of relational observables on $\ch_{\rm phys}$ and the Schr\"odinger picture of the Page-Wootters formalism, these criticisms were, however, resolved in \cite{Hoehn:2019owq,Hoehn:2020epv}, confirming the consistency of the  interpretation (see also \cite{Dolby:2004ak,Gambini:2008ke,giovannettiQuantumTime2015} for earlier, but distinct proposals for a resolution). Here, we extend this conditional probability interpretation to perspective-neutral states of general unimodular groups. Specifically, this will equip the relational observables introduced in Eq.~\eqref{qrelobs2} with a consistent conditional probability interpretation \emph{on} physical states, extending \cite{Hoehn:2019owq,Hoehn:2020epv} (see also \cite{Chataignier:2020fys}).

To this end, we first note that Corollary~\ref{cor_PIP} and Eqs.~\eqref{Sredmap} and~\eqref{condstate3} imply a preservation of inner products:
\begin{corol}\label{CorIsometry}
For $G$ unimodular, the physical inner product in Eq.~\eqref{PIP} equals the inner product on the physical system Hilbert space $\ch_{S,g}^{\rm phys}$:
\ba 
\braket{\psi_{\rm phys}|\phi_{\rm phys}}_{\rm phys}=\braket{\psi_S^{\rm phys}(g)|\phi_S^{\rm phys}(g)}\,,\q\q\forall\,g\in G\,.\nn
\ea 
Given Lemma~\ref{lem_invred}, the perspective-neutral Hilbert space $\ch_{\rm phys}$ and the physical system Hilbert space $\ch_{S,g}^{\rm phys}$ are therefore isometrically isomorphic, and every $\calr_\mathbf{S}(g)$ is an isometry.
\end{corol}

Next, we note that also expectation values of relational observables are preserved (extending results in \cite{Hoehn:2019owq,Hoehn:2020epv,periodic,Hoehn:2021flk} to general unimodular groups).
\begin{corol}\label{cor_exp}
Let $f_{S,g}^{\rm phys}\in\ca_{S,g}^{\rm phys}$ be some physical system observable and $F_{f_{S,g}^{\rm phys},R}(g)$ its corresponding relational observable. Then the expectation values are preserved as follows
\ba
\braket{\psi_{\rm phys}|\,F_{f_{S,g}^{\rm phys},R}(g)\,|\phi_{\rm phys}}_{\rm phys} = \braket{\psi_S^{\rm phys}(g)|\,f_{S,g}^{\rm phys}\,|\phi_S^{\rm phys}(g)}\,,\nn
\ea
where $\ket{\phi_S^{\rm phys}(g)} = \calr_\mathbf{S}(g)\,\ket{\phi_{\rm phys}}\in\ch_{S,g}^{\rm phys}$ is the state of $S$ conditioned on the frame $R$ being in orientation $g$ (and similarly for $\ket{\psi_S^{\rm phys}(g)}$).
\end{corol}

\begin{proof}
Using the second statement of Theorem~\ref{thm_obs}, the definition Eq.~\eqref{PIP} of the physical inner product, and the final statement of Corollary~\ref{cor_exp}, we find
\ba
\braket{\psi_{\rm phys}|\,F_{f_{S,g}^{\rm phys},R}(g)\,|\phi_{\rm phys}}_{\rm phys} &=&\braket{\psi_{\rm kin}|\, F_{f_{S,g}^{\rm phys},R}(g)\,|\phi_{\rm phys}}_{\rm kin} = \braket{\psi_{\rm kin}|\,\calr_\mathbf{S}^{-1}(g)\,f_{S,g}^{\rm phys}\,\calr_\mathbf{S}(g)\,|\phi_{\rm phys}}_{\rm kin}\nn\\
&=&\braket{\psi_S^{\rm phys}(g)|\,f_{S,g}^{\rm phys}\,|\phi_S^{\rm phys}(g)}\,.\nn
\ea
\end{proof}

These observations enable us to develop a conditional probability interpretation for the relations between the system $S$ of interest and the frame $R$. In particular, suppose $f_{S,g}^{\rm phys}\in\ca_{S,g}^{\rm phys}$ is self-adjoint so that it permits a spectral decomposition
\ba
f_{S,g}^{\rm phys} = \int_{\sigma(f_{S,g}^{\rm phys})} f\, dE_f,
\ea
where $d E_f$ is the corresponding spectral measure, and $\sigma(f_{S,g}^{\rm phys})$ denotes the spectrum of $f_{S,g}^{\rm phys}$. Given any normalized state $\ket{\psi_{\rm phys}}\in\ch_{\rm phys}$, the expectation value of this operator with respect to the corresponding state $\ket{\psi_S^{\rm phys}(g)}=\mathcal{R}_{\mathbf{S}}(g)\ket{\psi_{\rm phys}}$ is
\ba 
   \braket{f_{S,g}^{\rm phys}(g)}=\int_{\sigma(f_{S,g}^{\rm phys})} f\, dP(f),
\ea 
where $dP(f):=\langle\psi_S^{\rm phys}(g)|d E_f|\psi_S^{\rm phys}(g)\rangle$ defines a probability measure on the spectrum. To see the meaning of this probability, consider finite intervals $I=[f_{\min},f_{\max}]$ with $f_{\min}\leq f_{\max}$ with the corresponding spectral projection $E^I=E^I_S$, and invoke Corollary~\ref{cor_exp} to obtain
\ba
P(f_{S,g}^{\rm phys}\in I \ \mbox{when} \ \text{$R$ is in orientation $g$})&=&\bra{\psi_S^{\rm phys}(g)}E^I_S\ket{\psi_S^{\rm phys}(g)}\nn\\
&=&\bra{\psi_{\rm phys}}F_{E^I_S,R}(g)\ket{\psi_{\rm phys}}_{\rm phys}\label{condprob}.
\ea
In extension of the Page-Wootters formalism \cite{pageEvolutionEvolutionDynamics1983}, we can interpret this as the conditional probability that the observable $f_{S,g}^{\rm phys}$ takes a value $f$ in the interval $I$ when the reference system $R$ is in orientation $g$. Owing to the last line, this conditional probability is a manifestly gauge-invariant quantity, despite the expression in terms of non-invariant operators in the first line. In other words, the expression in the first line can be regarded as a gauge-fixed expression of a gauge-invariant quantity (see \cite{Hoehn:2019owq} for further discussion). Note also that the relational Dirac observables $F_{E^I_S,R}(g)$ are spectral projections too, thanks to the homomorphism established in Theorem~\ref{lem_homo}. Eq.~\eqref{condprob} tells us that this relational observable encodes in a precise sense the question ``what is the probability that $f_{S,g}^{\rm phys}$ takes value $f$, given that the frame is in orientation $g$?''. On physical states, relational observables thus come with a clear conditional probability interpretation that underscores their relational nature. 

Furthermore, multi-event probabilities for two observables $f_{S,g}^{\rm phys},\tilde f_{S,g'}^{\rm phys}$ can also be defined \cite{Hoehn:2019owq}:
\ba 
   P(f^{\rm phys}_{S,g}\in I \mbox{ when }g\,\,|\,\, \tilde f_{S,\tilde g}^{\rm phys}\in\tilde I \mbox{ when }\tilde g)=
   \frac{\bra{\psi_{\rm phys}} F_{E_S^{\tilde I},R}(\tilde g) F_{E_S^I,R}(g) F_{E_S^{\tilde I},R}(\tilde g)\ket{\psi_{\rm phys}}_{\rm phys}
   }{\bra{\psi_{\rm phys}}F_{E_S^{\tilde I},R}(\tilde g)\ket{\psi_{\rm phys}}_{\rm phys}
   }
\ea 
This can be generalized to arbitrarily many events in an obvious manner.

\section{Jumping into a QRF perspective II: symmetry reduction}\label{sec_Heisenberg}

In the previous section, we have seen how to reduce from the perspective-neutral physical Hilbert space $\ch_{\rm phys}$ into a `relational Schr\"odinger picture', in which the states of the system of interest transform unitarily under (gauge transformation induced) frame orientation label changes, while observables remained fixed. Here, we shall introduce an \emph{a priori} different reduction procedure, which, however, will \emph{a posteriori} be unitarily equivalent on physical states and yield a corresponding `relational Heisenberg picture'. The equivalence of the two reduction methods has so far only been shown in the Abelian case \cite{Hoehn:2019owq,Hoehn:2020epv, periodic,Hoehn:2021flk}. The below is its extension to general unimodular groups. This reduction procedure is the quantum analog of a phase space reduction through gauge-fixing \cite{Vanrietvelde:2018pgb,Vanrietvelde:2018dit,Hoehn:2018whn,hoehnHowSwitchRelational2018}. The basic idea is to (kinematically) disentangle the frame $R$ from the system $S$ and subsequently to condition the frame on an arbitrary orientation. On account of the kinematical disentanglement, it is then clear that the resulting system state will no longer transform under (gauge induced) changes of frame orientation label.

To construct the (kinematical) disentangler, we first need to find some (not necessarily normalized) state
\ba 
\ket \theta_R = \int_G dg N(g) \ket{\phi(g)}_R\label{trivialstate}
\ea
such that $N(g)$ are phases which obey the following reproducing equality:
\ba
N(g)=  \braket{\phi(g)|\theta}_R = \int_G dg' N(g') \braket{\phi(g)|\phi(g')}_R \,. \label{eq:reproducing} 
\ea
This will be the state that the frame will be kinematically disentangled into. Invoking that the coherent state overlap $\braket{\phi(g)|\phi(g')}_R$ is a reproducing kernel, \cite[Eq.\ (2.3.8)]{Perelomov} guarantees that any coefficients of an expansion of a generic state $\ket \theta_R = \int_G dg N(g) \ket{\phi(g)}_R$ into the (possibly overcomplete) basis of coherent states $\{\ket{\phi(g)}\}_{g \in G}$ is such that the coefficients $N(g)$ obey this equation. The non-trivial condition we are imposing is that one can choose $N(g)$ to be a phase. For example,
in the case of $G$ compact this is always possible, as one can simply choose $\ket \theta_R = \int_G dg e^{i \alpha(g)} \ket{\phi(g)}_R\in\ch_R$ for arbitrary real function $\alpha(g)$, since $\braket{\theta|\theta}_R={\rm Vol}(G)$ is finite when Eq.~\eqref{eq:reproducing} is satisfied. For non-compact groups, the situation is more subtle as the coherent states are then distributions and we are essentially looking for plane wave states, hence $\ket{\theta}_R$ will be another distribution. For Abelian groups, it is also possible to find such states, in which case the phases $N(g)$ even constitute a unitary one-dimensional representation of $G$ \cite{Hoehn:2019owq,Hoehn:2020epv,Hoehn:2021wet}. It is also clear that Eq.~\eqref{eq:reproducing} can always be satisfied for \emph{any} phases $N(g)$ for the special case of the regular representation of an arbitrary unimodular $G$ since in that case $\braket{g|g'}_R=\delta(g,g')$; in fact, for ideal frames, we can even choose $N(g)=1$, $\forall\,g\in G$. This suggests that it should be possible to find states $\ket\theta_R$ such that Eq.~\eqref{eq:reproducing} can be satisfied also for non-compact non-Abelian unimodular groups beyond the regular representation. However, we have not attempted to characterize the representations for which this holds and shall just assume that some such state $\ket\theta_R$ can be found.

We can then define the kinematical disentangler (`trivialization map'), which is a frame-conditional transformation (e.g., see \cite{Hoehn:2021wet}) by
\ba
\ct_R:=\int_Gdg\,N(g)\ket{\phi(g)}\!\bra{\phi(g)}_R\otimes U_S^\dag(g)\,.
\ea
It draws its name from the following properties:
\begin{lemma}\label{lem_triv}
The disentangler $\ct_R$ has a weak left-inverse, which is given by
    \ba 
\ct_R^\dag=\int_Gdg\,N^*(g)\ket{\phi(g)}\!\bra{\phi(g)}_R\otimes U_S(g)\,.
\label{eqTRdagger}
\ea
That is,
    \ba 
    \ct_R^\dag\cdot\ct_R\approx\mathbf{1}_{RS}\,.
    \ea 
    Furthermore, the kinematical disentangler maps the physical projector $\Pi_{\rm phys}$ into a product operator
    \ba 
    \ct_R\Pi_{\rm phys}\ct_R^\dag =\ketbra{\theta}{\theta}_R \otimes \Pi_S^{\rm phys}(e)\,,
    \ea 
    where $\ket\theta_R$ is the state given in Eq.~\eqref{trivialstate}, and
     kinematically disentangles the frame $R$ from $S$ in arbitrary physical states
    \ba 
    \ct_R\ket{\psi_{\rm phys}}=\ket\theta_R\otimes \ket{\psi_S^{\rm phys}(e)}\,.
    \ea 
    Here $\ket{\psi_S^{\rm phys}(e)}$ is the physical system state when $R$ is ``in the origin''.
\end{lemma}
\begin{proof}
The first statement is verified by direct computation
\ba
\ct_R^\dag\cdot\ct_R&=&\int_G dg \int_G dg'\,N^*(g)N(g')\ket{\phi(g)}\!\braket{\phi(g)|\phi(g')}\!\bra{\phi(g')}_R \otimes U_S(gg'^{-1})\nn\\
&\approx& \int_G dg \int_G dg'\,N^*(g)N(g')\ket{\phi(g)}\!\braket{\phi(g)|\phi(g')}\!\bra{\phi(g')}_R U_R^\dagger(gg'^{-1}) \otimes \mathbf{1}_S \nn\\
&=& \int_G dg \int_G dg'\,N^*(g)N(g')\ket{\phi(g)}\!\braket{\phi(g)|\phi(g')}\!\bra{\phi(g)}_R \otimes \mathbf{1}_S \nn\\
&=& \int_G dg \,N^*(g) \ketbra{\phi(g)}{\phi(g)} \int dg' N(g')\braket{\phi(g)|\phi(g')}_R \otimes \mathbf{1}_S \nn\\
&=& \int_Gdg \,N^*(g)N(g) \ketbra{\phi(g)}{\phi(g)} \otimes \mathbf{1}_S  = \mathbf{1}_{RS}
\ea 
where we have used the fact that $\mathbf{1}_S \otimes U_R(g) \approx  U_S^\dagger(g) \otimes \mathbf{1}_R$ from the first to the second line and Equation~\eqref{eq:reproducing} in going from the penultimate line to the last line. Finally, we used that $N(g)$ is a phase and the resolution of the identity \eqref{resid}.

The second statement follows from
\ba
\ct_R\Pi_{\rm phys}\ct_R^\dag&=&\int_G dg \int_G dg' \int_G d\tilde g\,N(g)N^*(g')\braket{\phi(g)|U_R(\tilde g)|\phi(g')}_R\ket{\phi(g)}\!\bra{\phi(g')}_R\otimes U_S(g^{-1}\tilde g g')\nn\\
&=&\int_G dg \int_G dg'\,N(g)N^*(g')\ket{\phi(g)}\!\bra{\phi(g')}_R\otimes \int_Gd\tilde g\braket{\phi(e)|\phi(g^{-1}\tilde gg')}_R U_S(g^{-1}\tilde gg')\nn\\
&=&\int_G dg \int_G dg'\,N(g)N^*(g')\ket{\phi(g)}\!\bra{\phi(g')}_R\otimes\Pi_S^{\rm phys}(e)\nn\\
&=&\ketbra{\theta}{\theta}_R \otimes\Pi_S^{\rm phys}(e)\,,
\ea
using unimodularity and Eq.~\eqref{Sproj} when passing from the second to the third line.

The third statement follows from
\ba
\ct_R\ket{\psi_{\rm phys}} &=&\int_G dg\,N(g)\ket{\phi(g)}\!\bra{\phi(g)}_R \otimes U_S^\dagger(g) \ket{\psi_{\rm phys}}\nn\\
&=&\int_G dg\,N(g)\ket{\phi(g)}\!\bra{\phi(g)}_R  U_R(g) \otimes \mathbf{1}_S \ket{\psi_{\rm phys}}\nn\\
&=&\int_G dg\,N(g)\ket{\phi(g)}\!\bra{\phi(e)}_R  \otimes \mathbf{1}_S  \ket{\psi_{\rm phys}}\nn\\
&=&\ket\theta_R\otimes \ket{\psi_S^{\rm phys}(e)}\,.
\ea
\end{proof}

We can exploit this lemma, to define the Heisenberg reduction map $\calr_\mathbf{H}:\ch_{\rm phys}\rightarrow\ch_{S,e}^{\rm phys}$, which upon disentangling the frame, conditions it on an \emph{arbitrary} orientation $g\in G$:
\ba \label{Heisenbergred}
\calr_{\mathbf{H}}:=N^*(g)\left(\bra{\phi(g)}_R\otimes\mathbf{1}_S\right)\ct_R 
\ea 
Indeed, invoking Eq.~\eqref{eq:reproducing}, we find that when acting on physical states, this map is independent of $g$:
\ba \label{Heisstate}
\calr_{\mathbf{H}}\ket{\psi_{\rm phys}}=\ket{\psi_S^{\rm phys}(e)}\,.
\ea 
The following result establishes (weak) invertibility of the Heisenberg reduction map:
\begin{lemma}\label{lem_Hredinv}
The inverse $\calr_{\mathbf{H}}^{-1}:\ch_{S,e}^{\rm phys}\rightarrow\ch_{\rm phys}$ of the Heisenberg reduction map reads
\ba \label{RHinv}
    \calr_\mathbf{H}^{-1}=\ct_R^\dag\left(\ket\theta_R\otimes\mathbf{1}_S\right)
    \ea 
    and we have, for all $g\in G$,
    \ba 
    \calr_\mathbf{H}= U_S^\dag(g)\calr_\mathbf{S}(g)\,,\q\q\calr_\mathbf{H}^{-1}=\calr_\mathbf{S}^{-1}(g)\,U_S(g)\,.
    \ea 
\end{lemma}
\begin{proof}
    The form of $\calr_\mathbf{H}^{-1}$ is immediate from Lemma~\ref{lem_triv}. The statement $\calr_\mathbf{H}= U_S^\dag(g)\calr_\mathbf{S}(g)$ for all $g\in G$ follows directly from Eq.~\eqref{condcov} and Eq.~\eqref{Heisstate}. Finally, using Eq.~(\ref{trivialstate}) and Eq.~(\ref{eqTRdagger}),
     \ba     
     \calr_\mathbf{H}^{-1}&=&\int_G dg' \int_G dg\,N^*(g)N(g') \ket{\phi(g)}\!\braket{\phi(g)|\phi(g')}_R \otimes U_S(g)\nn\\
    &=&\int_G dg N^*(g) N(g) \ket{\phi(g)}_R \otimes U_S(g) \nn\\
         &=&\int_G dg  \ket{\phi(g)}_R \otimes U_S(g {g'}^{-1}) U_S(g') \nn \\
          &=& \int_G d\tilde g  \ket{\phi(\tilde g g')}_R \otimes U_S(\tilde g) U_S(g') \nn \\
     &=& \calr_\mathbf{S}^{-1}(g')\,U_S(g'),
    \ea 
    where we have used Eq.~\eqref{eq:reproducing} in going from the first to second line, the assumption that $N(g)$ is a phase in going to line three, as well as the substitution $\tilde g = g {g'}^{-1}$ in line four, and finally Lemma~\ref{lem_invred} to get the last line.
\end{proof}

This lemma thereby establishes the weak unitary equivalence of the two reduction maps. Strong equivalence (on a larger domain of definition than $\ch_{\rm phys}$) only holds for the regular representation.

Let us now clarify why we call the result of $\calr_\mathbf{H}$ a relational Heisenberg picture. Consider the ``Heisenberg observables'' for the system $S$
\ba \label{Heisobs}
f_{S,g'}^{\rm phys}(g):=U_S^\dag(g)\,f_{S,g'}^{\rm phys}\,U_S(g)
\ea 
for arbitrary $f_{S,g'}^{\rm phys}\in\ca_{S,g'}^{\rm phys}$. Lemma~\ref{lem_Sproj} tells us that, unless condition~\eqref{projgind} is fulfilled, $f_{S,g'}^{\rm phys}(g)$ will no longer be contained in $\ca_{S,g'}^{\rm phys}$. Indeed, recalling Eq.~\eqref{fsys}, we have
\ba \label{Heisobs0}
f_{S,g'}^{\rm phys}(g)=U_S^\dag(g)\,\Pi_S^{\rm phys}(g')\,f_S\,\Pi_S^{\rm phys}(g')\,U_S(g)
\ea 
for some $f_S\in\ca_S$. Eq.~\eqref{projgind2} implies that we can rewrite this as
\ba\label{ggprime}
f_{S,g'}^{\rm phys}(g)=\Pi_S^{\rm phys}(g^{-1}g')\,f_S(g)\,\Pi_S^{\rm phys}(g^{-1}g')\,.
\ea 
Since the kinematical Heisenberg observable $f_S(g)$ is contained in $\ca_S$ for all $g\in G$, we have ${f_{S,g'}^{\rm phys}(g)\in\ca_{S,g^{-1}g'}^{\rm phys}}$. Changing the frame orientation $g$ of $R$ in the relational Heisenberg picture therefore `rotates' $\ca_{S,g}^{\rm phys}$ through the kinematical algebra $\ca_S$, unless condition Eq.~\eqref{projgind} is satisfied. Recall that, even if Eq.~\eqref{projgind} is not satisfied for all $g\in G$, if either $g$ is in the isotropy group of frame orientation state $\ket{\phi(g')}_R$, hence of the form $g=g'h^{-1}g'^{-1}$ for some $h\in H$, or if $g'$ resides in the center of $G$, then we still have $\ca_{S,g^{-1}g'}^{\rm phys}=\ca_{S,g'}^{\rm phys}$, cf.\ the discussion around Eqs.~\eqref{isotropyobs} and~\eqref{nonsym}. 

We can, however, consider the following `evolution', setting $g=g'$ on the right hand side in Eq.~\eqref{ggprime}:
\ba \label{Heisobs2}
f_{S,e}^{\rm phys}(g):=\Pi_S^{\rm phys}(e)\,f_S(g)\,\Pi_S^{\rm phys}(e)\,.
\ea
This projection of the kinematical Heisenberg `evolution' thus is guaranteed to remain inside $\ca_{S,e}^{\rm phys}$, which is why we denote it $f_{S,e}^{\rm phys}(g)$ (rather than $f_{S,g}^{\rm phys}(g)$). This is what we need for the following observable embedding map.

Let us consider the Heisenberg picture encoding map $\mathcal{E}_\mathbf{H}:\ca_{S,e}^{\rm phys}\rightarrow\ca_{\rm phys}$, in analogy to its Schr\"odinger picture counter part \eqref{Sencode}, to embed the physical Heisenberg observables in Eq.~\eqref{Heisobs2} into the algebra of relational observables:
\ba 
\mathcal{E}_\mathbf{H}\left(f_{S,e}^{\rm phys}(g)\right):=\calr_\mathbf{H}^{-1}\,f_{S,e}^{\rm phys}(g)\,\calr_\mathbf{H}\,.
\ea 
Since $\calr_{\mathbf{H}}$ and its inverse only map between $\ch_{\rm phys}$ and $\ch_{S,e}^{\rm phys}$, and $\ch_{S,e}^{\rm phys}$ and $\ch_{S,g}^{\rm phys}$ are distinct subspaces of $\ch_S$  when condition~\eqref{projgind} is violated and $g\notin H$ or $g$ is not an element of the center of $G$, we can only embed elements of $\ca_{S,e}^{\rm phys}$ into $\ca_{\rm phys}$ using the Heisenberg map. When condition~\eqref{projgind} is satisfied, we have $\ca_{S,e}^{\rm phys}=\ca_{S,g}^{\rm phys}$, for all $g\in G$ and Eq.~\eqref{Heisobs2} reduces to $f_{S,e}^{\rm phys}(g)=U_S^\dag(g)\,f_{S,e}^{\rm phys}\,U_S(g)$.

In analogy to Theorem~\ref{thm_obs}, we then have:
\begin{theorem}
   Let $f_S\in\ca_S$ be arbitrary. Reducing the relational observable $F_{f_S,R}(g)$ with the Heisenberg reduction map yields the corresponding physical system observable $f_{S,g}^{\rm phys}$ given in Eq.~\eqref{fsys}:
\ba
\calr_\mathbf{H}\,F_{f_{S},R}(g)\,\calr_\mathbf{H}^{-1} =\Pi_S^{\rm phys}(e)\,f_S(g)\,\Pi_S^{\rm phys}(e)= f_{S,e}^{\rm phys}(g)\,.
\ea
Conversely, let $f_{S,e}^{\rm phys}(g)\in\ca_{S,e}^{\rm phys}$ be arbitrary. Then it embeds via the encoding map into $\ca_{\rm p hys}$ as its corresponding relational observable
\ba
\mathcal{E}_\mathbf{H}\left(f_{S,e}^{\rm phys}(g)\right)\approx F_{f_{S,g}^{\rm phys},R}(g)\,.\nn
\ea 
\end{theorem}

This theorem clarifies why we call the result of the reduction by $\calr_{\mathbf{H}}$ the relational Heisenberg picture: the reduced physical states are fixed under (gauge transformation induced) changes of orientation label of the frame, while the reduced relational observables transform in a projected covariant way under them. This result establishes another equivalence between the algebras $\ca_{\rm phys}$ and $\ca_{S,e}^{\rm phys}$, preserving algebraic properties, thanks to Theorem~\ref{lem_homo}.

\begin{proof}
To prove the first statement, we recall from Lemma~\ref{lem_Hredinv} that $\calr_\mathbf{H}^{-1}=\calr_\mathbf{S}^{-1}(g)U_S(g)$ and $\calr_\mathbf{H}= U_S^\dag(g)\calr_\mathbf{S}(g)$ for arbitrary $g\in G$. This permits us to write
\ba 
\calr_\mathbf{H}\,F_{f_{S},R}(g)\,\calr_\mathbf{H}^{-1} &=&U_S^\dag(g)\calr_\mathbf{S}(g)\,F_{f_S,R}(g)\,\calr_\mathbf{S}^{-1}(g)\,U_S(g)\nn\\
&=&U_S^\dag(g)\,\Pi_S^{\rm phys}(g)\, f_S\,\Pi_S^{\rm phys}(g)\,U_S(g)\nn\\
&\underset{\eqref{projgind2}}{=}&\Pi_S^{\rm phys}(e)\,U_S^\dag(g)\,f_S\,U_S(g)\,\Pi_S^{\rm phys}(e)\,,\nn
\ea 
where we have invoked Theorem~\ref{thm_obs} to get from the first to the second line.

For the converse statement, we once more invoke Lemma~\ref{lem_Hredinv} to write
\ba 
\mathcal{E}_\mathbf{H}\left(f_{S,e}^{\rm phys}(g)\right)&=&\calr_\mathbf{S}^{-1}(g)\,U_S(g)\,f_{S,e}^{\rm phys}(g)\,U_S^\dag(g)\,\calr_\mathbf{S}(g)\nn\\
&\underset{\eqref{Heisobs0}}{=}&\calr_\mathbf{S}^{-1}(g)\,\Pi_S^{\rm phys}(g)\,f_S\,\Pi_S^{\rm phys}(g)\,\calr_\mathbf{S}(g)\nn\\
&=&\mathcal{E}_\mathbf{S}^g\left(f_{S,g}^{\rm phys}\right)\nn\\
&\approx&F_{f_{S,g}^{\rm phys},R}(g)\,,\nn
\ea 
where the last line again follows from Theorem~\ref{thm_obs}.
\end{proof}

From Corollary~\ref{cor_exp}, it now also follows immediately that the expectation values of observables are preserved. In particular, since the homomorphism of Theorem~\ref{lem_homo} is also unital by Lemma~\ref{lem_Rd} and Eq.~\eqref{qrelobs}, this also means that the inner product is preserved and that the map $\calr_\mathbf{H}$ is therefore unitary.
\begin{corol}
Let $f_{S,g}^{\rm phys}\in\ca_{S,g}^{\rm phys}$ be some physical system observable and $F_{f_{S,g}^{\rm phys},R}(g)$ its corresponding relational observable. Then the expectation values are preserved as follows
\ba
\braket{\psi_{\rm phys}|\,F_{f_{S,g}^{\rm phys},R}(g)\,|\phi_{\rm phys}}_{\rm phys} = \braket{\psi_S^{\rm phys}(e)|\,f_{S,e}^{\rm phys}(g)\,|\phi_S^{\rm phys}(e)}\,,\nn
\ea
where $\ket{\phi_S^{\rm phys}(e)} = \calr_\mathbf{H}\,\ket{\phi_{\rm phys}}\in\ch_{S,e}^{\rm phys}$ is the relational Heisenberg state associated with the physical state $\ket{\phi_{\rm phys}}$ (and similarly for $\ket{\psi_S^{\rm phys}(e)}$).
\end{corol}

\section{Gauge induced QRF transformations: quantum coordinate changes} \label{sec: QRF transformations as quantum coordinate changes}

In order to consider changes of internal QRF perspectives, we need to introduce a second reference frame. This will also help to illustrate why the physical Hilbert space is called a perspective-neutral Hilbert space. We shall henceforth assume the total kinematical Hilbert space to be of the form\footnote{We can think of this as decomposing our original kinematical system Hilbert space as $\ch_S=\ch_{R_2}\otimes\ch_{S'}$ and then renaming $S'=S$. $S'$ may contain many more possible quantum reference systems.}
\ba 
\ch_{\rm kin}=\ch_{R_1}\otimes\ch_{R_2}\otimes\ch_S\nn 
\ea
where, as before, $\ch_S$ carries an arbitrary unitary representation of the  unimodular Lie group $G$ and every $\ch_{R_i}$, $i=1,2$, admits a coherent state system $\{U_{R_i},\ket{\phi(g)}_{R_i}\}$ on which $G$ acts transitively from the left so that $R_i$ constitutes a complete reference frame if $G$ acts also freely on its orientation states, and otherwise an incomplete reference frame. In particular, our analysis encompasses the situation that one of the $R_i$ is complete and the other is incomplete. As before, we will assume that each coherent state system generates a resolution of the identity. 

Since we now have two types of symmetries, gauge transformations and symmetries as frame reorientations (see Sec.~\ref{ssec_symgauge}), we will also find two distinct types of quantum reference frame transformations. In this section, we shall discuss the gauge transformation induced quantum reference frame transformations. They essentially correspond to a change of gauge as they will change the frame orientation conditioning of the gauge-invariant physical states and such a conditioning amounts to a gauge-fixing of a manifestly gauge-invariant quantity. As such, these quantum reference frame transformations will only change the description of the \emph{same} gauge-invariant states and observables and they will take the form of ``quantum coordinate transformations'', with the gauge-invariant perspective-neutral structure serving as the linking structure of all the gauge-fixed internal perspectives. In Sec.~\ref{sec_relobschanges}, we shall then construct the symmetry-induced quantum reference frame transformations and these will actually change gauge-invariant observables, mapping relational observables relative to $R_1$ to those relative to $R_2$.

Specifically, these sections will also clarify why the perspective-neutral formulation of quantum frame covariance is the quantum analog of classical special covariance, as discussed in Sec.~\ref{ssec_SR}, which likewise features passive and active frame changes.

\subsection{Frame changes in the `relational Schr\"odinger picture'}\label{sec:QRF_change_S}

Upon choosing $R_i$ as the frame, the entire previous discussion applies if one replaces the system $S$ in the previous sections with the composite system $R_jS$, $j\neq i$ here. In particular, the physical system projector now takes the form
\ba 
\Pi_{R_jS}^{\rm phys}(g):=\left(\bra{\phi(g)}_{R_i}\otimes\mathbf{1}_{R_jS}\right)\Pi_{\rm phys}\left(\ket{\phi(g)}_{R_i}\otimes\mathbf{1}_{R_jS}\right),
\ea 
and the physical system Hilbert space becomes a subspace $\ch_{R_jS,g}^{\rm phys}\subset\ch_{R_j}\otimes\ch_S$ of the kinematical Hilbert space for $R_jS$. It follows from \cite{Hoehn:2021wet} that $\ch_{R_jS,g}^{\rm phys}$ will not in general inherit a tensor factorization between $R_j$ and $S$ from the kinematical Hilbert space factors $\ch_{R_j}\otimes\ch_S$, however, this is not a problem for our purposes. The Schr\"odinger reduction map relative to frame $R_i$, $\calr_{\mathbf{S},R_i}(g):\ch_{\rm phys}\rightarrow\ch_{R_jS,g}^{\rm phys}$ maps the perspective-neutral Hilbert space into the physical system Hilbert space of $R_jS$ and reads
\ba 
\calr_{\mathbf{S},R_i}(g):=\bra{\phi(g)}_{R_i}\otimes\mathbf{1}_{R_jS}
\ea 
with the obvious inverse according to Lemma~\ref{lem_invred}. Applied to physical states, this is a gauge-fixing of the gauge-symmetry induced redundant kinematical information inside it. When $R_i$ is a complete reference frame, this amounts to a complete gauge-fixing of the $G$-orbits, while it is a partial gauge-fixing when $R_i$ is incomplete, featuring a non-trivial isotropy group $H_i$.\footnote{We emphasize that $H_i\neq H_j$ is possible, depending on the representations $U_{R_i}$ and $U_{R_j}$ of $G$.} The idea is now to interpret the reduction map $\calr_{\mathbf{S},R_i}$ as the \emph{quantum coordinate map} from the perspective-neutral Hilbert space $\ch_{\rm phys}$, assuming the role of the quantum analog of the frame-independent manifold, into the internal perspective of quantum reference frame $R_i$ on the remaining subsystems $R_j,S$ \cite{Vanrietvelde:2018pgb}. This is the quantum analog of our discussion of classical special covariance in Sec.~\ref{ssec_SR}. As such, we interpret $\ch_{R_jS,g}^{\rm phys}$ and $\ca_{R_jS,g}^{\rm phys}$ as the quantum coordinate description of the \emph{a priori} frame-choice-independent state space $\ch_{\rm phys}$ and observable algebra $\ca_{\rm phys}$ and we substantiate this interpretation shortly. As noted before, $R_i$ will only be able to resolve properties of $R_jS$ that are also invariant under its own isotropy group $H_i$. Specifically, $\calr_{\mathbf{S},R_i}(g_i)$ and $\calr_{\mathbf{S},R_j}(g_j)$ amount to two distinct gauge-fixings of $\ket{\psi_{\rm phys}}$, fixing distinct kinematical frame variables inside it.

It is therefore clear that the change of frame map $V_{R_i\to R_j}(g_i,g_j):\ch_{R_jS,g_i}^{\rm phys}\rightarrow\ch_{R_iS,g_j}^{\rm phys}$ for state descriptions from the perspective of frame $R_i$ in orientation $g_i$ to the new frame $R_j$ in orientation $g_j$ takes the compositional form of ``quantum coordinate changes'', generalizing the QRF transformations in  \cite{Vanrietvelde:2018dit,Vanrietvelde:2018pgb,hoehnHowSwitchRelational2018,Hoehn:2018whn,Hoehn:2019owq,Hoehn:2020epv,Hoehn:2021wet,Hoehn:2021flk,periodic} (see also \cite{castro-ruizTimeReferenceFrames2019,Giacomini:2021gei}) to general groups:
\ba 
V_{R_i\to R_j}(g_i,g_j):=\calr_{\mathbf{S},R_j}(g_j)\cdot\calr_{\mathbf{S},R_i}^{-1}(g_i)\,.
\ea 
By construction, we have $V_{R_i\to R_j}(g_i,g_j)\,\ket{\psi_{R_jS}^{\rm phys}(g_i)} = \ket{\psi_{R_iS}^{\rm phys}(g_j)}$. This map is unitary because the individual quantum coordinate maps are isometries on the relevant Hilbert spaces, regardless of the presence of isotropy groups. In analogy to coordinate changes on a manifold, the change of frame perspective always passes through the perspective-neutral physical Hilbert space since $\ch_{\rm phys}=\calr_{\mathbf{S},R_i}^{-1}(g_i)\left(\ch_{R_jS,g_i}^{\rm phys}\right)$, see Figs.~\ref{fig:H_space_relations} and~\ref{fig:H_space_map} for illustration. In this sense, the perspective-neutral structure links the various internal QRF perspectives on the physics. Since $R_iR_jS$ may contain many more possible choices of quantum reference frames and $\ch_{\rm phys}$ encodes all of them, we call the physical Hilbert space also the quantum reference frame perspective-neutral Hilbert space. This is in analogy to how a spacetime neighbourhood encodes all possible local spacetime frame choices without giving preference to any one of them, in this sense being spacetime-frame-neutral. We emphasize that this picture relies on the gauge-symmetry induced redundancy in the description of the manifestly gauge-invariant physical states $\ket{\psi_{\rm phys}}$ in terms of kinematical structure. This provides a platform for establishing the quantum analog of general covariance.

It is also clear that observable descriptions transform from $R_i$'s to $R_j$'s perspective according to the map $\Lambda_{R_i\to R_j}(g_i,g_j):\ca_{R_jS,g_i}^{\rm phys}\rightarrow\ca_{R_iS,g_j}^{\rm phys}$ given by
\ba
\Lambda_{R_i\to R_j}(g_i,g_j)\left(f_{R_jS,g_i}^{\rm phys}\right):=V_{R_i\to R_j}(g_i,g_j)\,f_{R_jS,g_i}^{\rm phys}\,   V_{R_j\to R_i}(g_j,g_i)\,.\label{obstrans1}
\ea 
This transformation maps the reduced observables relative to the frame perspective of $R_i$ via the perspective-neutral algebra $\ca_{\rm phys}$ into that of $R_j$. The image of this map will in general depend on both frame orientations $g_i,g_j$, even if $f_{R_jS,g_i}^{\rm phys}$ is independent of the orientation $g_i$ of the old frame $R_i$. As explained in \cite{Hoehn:2019owq,Hoehn:2020epv}, this can be interpreted as an indirect self-reference of the new frame $R_j$. This is the quantum analog of the observation in Sec.~\ref{ssec_SR}, case (A), that in special relativity ``relational observables'' look simple in the internal perspective of the frame relative to which the observables are formulated, and more complicated in other internal perspectives (in fact, depending on the relative frame orientation).

One might wonder why we call $\ca_{\rm phys}$ perspective-neutral, given that it contains relational observables such as $F_{f_S,R_i}(g_i)$, which describes the system property $f_S$ relative to frame $R_i$. However, since this encodes the relation between the $R_i$ and $S$ in a manifestly gauge-invariant manner, we can also reduce, i.e.\ gauge-fix it into any one of the internal frame perspectives, providing a description of the relation between $R_i$ and $S$ in some other frame's perspective (see case (A) in Sec.~\ref{ssec_SR} for the analogous situation in special relativity). Furthermore, and more importantly, since $\ca_{\rm phys}$ and the reduced algebras $\ca_{R_jS,g_i}^{\rm phys}$ and $\ca_{R_iS,g_j}^{\rm phys}$ are all isometrically isomorphic and related by the quantum coordinate maps, which thanks to Theorem~\ref{thm_obs} embed the reduced observables as relational observables into $\ca_{\rm phys}$, we can formally write any abstract element in $\ca_{\rm phys}$ as a relational observable describing some property of $R_jS$ relative to $R_i$ and \emph{also} as a relational observable describing some \emph{different} property of $R_iS$ relative to $R_j$. In other words, any abstract element in $\ca_{\rm phys}$ admits various physical interpretations and incarnations as a relational observable. In this sense, reference frames provide the ``context'' within which we can equip an abstract element of $\ca_{\rm phys}$ with a concrete physical interpretation, and in general there exist many such ``contexts''. Thus, when we call $\ca_{\rm phys}$ perspective-neutral we mean it in this abstract sense, i.e.\ prior to choosing a physical context within which to interpret its elements physically (see the analogy with case (B) in special relativity, Sec.~\ref{ssec_SR}). Once more, this is only possible due to the gauge-symmetry induced redundancy in the description of the manifestly gauge-invariant elements of $\ca_{\rm phys}$.

We can write the change of QRF perspective map in a more explicit form:
\begin{theorem}
For $i\neq j$, the change of frame map has the form
\ba
V_{R_i\to R_j}(g_i,g_j)=\int_G dg\,\ket{\phi(g g_i)}_{R_i}\otimes\bra{\phi(g^{-1} g_j)}_{R_j}\otimes U_S(g)\,.\nn
\ea
\end{theorem}
\begin{proof}
The definitions of the reduction maps entail
\ba
V_{R_i\to R_j}(g_i,g_j)=\left(\bra{\phi(g_j)}_{R_j}\otimes \mathbf{1}_{R_iS}\right)\,\Pi_{\rm phys}\left(\ket{\phi(g_i)}_{R_i}\otimes \mathbf{1}_{R_jS}\right)\,.\nn
\ea
Noting that $\Pi_{\rm phys}=\int_G dg\,U_{R_1}(g)\otimes U_{R_2}(g)\otimes U_{S}(g)$, the result follows immediately from the left action of $U_{R_i}$ on $\ket{\phi(g)}_{R_i}$, $i=1,2$.
\end{proof}

When considering the special case of systems $L^2(G)$ carrying left and right regular representations of $G$, restricting the change of frame map to $g_i=g_j=e$ and appending the reference system states, we find
\ba
\ket{\phi(e)}_{R_j}\otimes\ket{\psi_{R_iS}^{\rm phys}(e)} = \mathrm{SWAP}_{i,j}\,\otimes V_{R_i\to R_j}(e,e)\left(\ket{\phi(e)}_{R_i}\otimes\ket{\psi_{R_jS}^{\rm phys}(e)}\right)\,, \label{QRF coordinate change}
\ea
where $\rm{SWAP}_{i,j}$ exchanges the $i,j$ labels.  This is the form of the QRF transformations established in \cite{hamette2020quantum} for general groups except that $V_{R_i\to R_j}(e,e)$ acts with the left regular representation on the systems $\ch_k$, for $k \neq i, j$. In \cite{hamette2020quantum} the systems $\ch_k$ are acted on by the right regular representation when changing reference system. One would recover the exact change of quantum reference frame of~\cite{hamette2020quantum} by requiring that the physical Hilbert space be invariant under the right regular representation. The choice of acting on the left or the right on the coherent states as corresponding to gauge is just conventional (as long as symmetries always act ``on the other side''). Similarly, the expression for $V_{R_i\to R_j}(g_i,g_j)$ coincides with the quantum reference frame transformations established for ideal frames and finite groups in \cite{hamette2020quantum,Krumm:2020fws,Hoehn:2021flk} (when replacing the integral with a sum $\sum_{g\in G}$). Our present construction thus significantly generalizes previous constructions in the literature.

\begin{figure}
    \centering
    \includegraphics[width= 400pt]{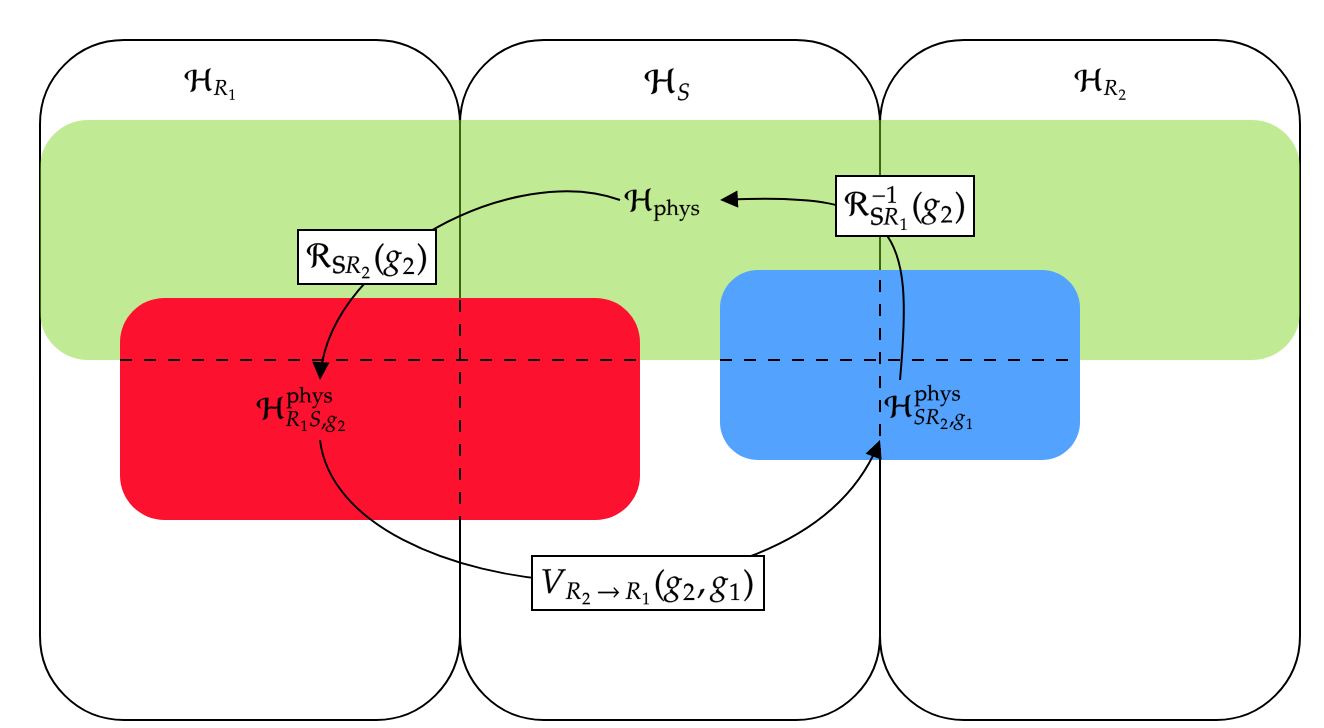}
    \caption{This figure illustrate the relations between different physical and kinematical subsystem Hilbert spaces, as well as their embedding in the larger Hilbert space $\ch_\kin \cong \ch_{R_1} \otimes \ch_S \otimes \ch_{R_2}$ and $\ch_\phys$. In the figure we depict the case $\ch_\phys \subset \ch_\kin$ for simplicity. The physical subsystem space $\ch^\phys_{R_1 S,g_2}$ of subsystems $R_1S$ conditional on $R_2$ in orientation $g_2$ is a subspace of $\ch_{R_1} \otimes \ch_{S}$. Similarly for  $\ch^\phys_{R_2 S,g_1}$ and $\ch_{S} \otimes \ch_{R_2}$. One can map directly between the physical system spaces of $R_1S$ relative to $R_2$ and $R_2S$ relative to $R_1$ via the change of reference system operator $V_{R_1 \to R_2}(g_2,g_1)$. As $R_2$ changes orientation $g_2$, the subspace $\ch_{R_1S,g_2}^{\rm phys}$ ``rotates'' through the kinematical $\ch_{R_1}\otimes\ch_S$ if symmetries are absent (similarly with $R_1$ and $R_2$ interchanged).}
    \label{fig:H_space_relations}
\end{figure}

\begin{figure}
    \centering
    \includegraphics[width= 500pt]{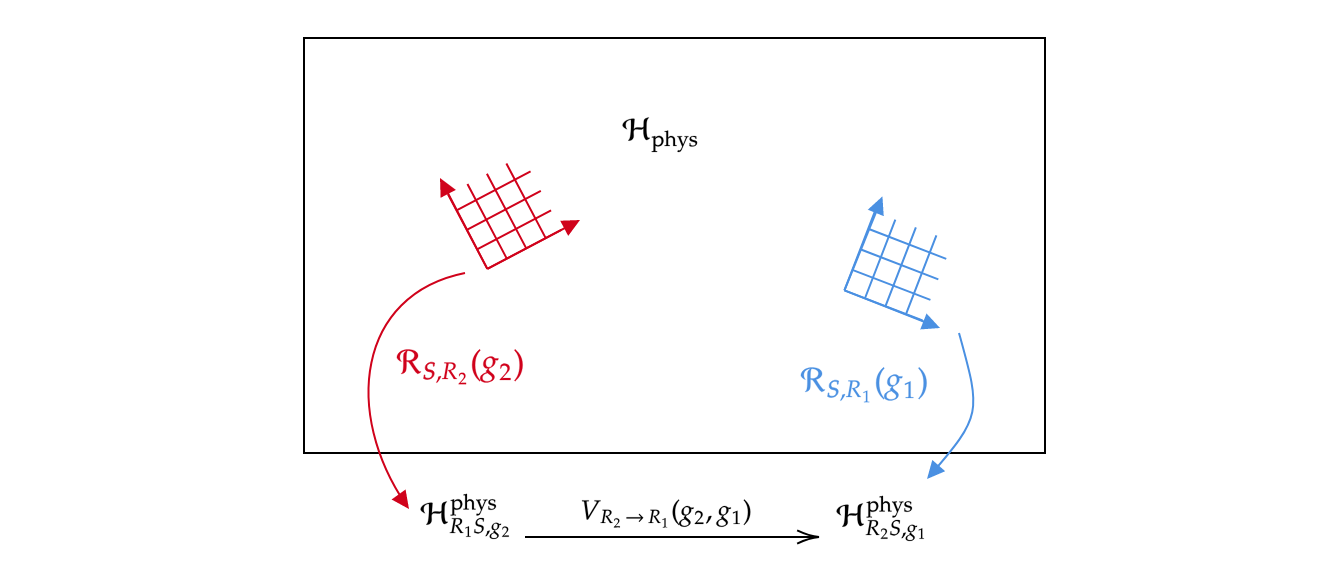}
    \caption{In analogy with special relativity (recall Sec.~\ref{ssec_SR}), one can see the physical Hilbert space $\ch_\phys$ as the global manifold. The reduced Hilbert space $\ch_{R_1S,g_2}$ is the ``quantum coordinate'' description of the physical degrees of freedom from the perspective of $R_2$ in configuration $g_2$, which is analogous to a coordinate patch on the global manifold. The reduced Hilbert space $\ch_{R_1S,g_2}^{\rm phys}$ is obtained from $\ch_\phys$ via the map $\calr_{\mathbf{S},R_2}(g_2)$ which corresponds to a coordinate map in our analogy. Hence, the map $V_{R_2\to R_1}(g_2,g_1)$ is analogous to a change of coordinates.}
    \label{fig:H_space_map}
\end{figure}

\subsection{Frame changes in the `relational Heisenberg picture'}\label{sec:QRF_change_H}

We can repeat the previous calculation also in the relational `Heisenberg picture' of Sec.~\ref{sec_Heisenberg}, i.e.\ invoking the Heisenberg reduction maps $\calr_{\mathbf{H},R_i}:\ch_{\rm phys}\rightarrow\ch_{R_jS,e}^{\rm phys}$ now, instead of Eq.~\eqref{Heisenbergred}, given by
\ba 
\calr_{\mathbf{H},R_i}:=N_i^*(g_i)\left(\bra{\phi(g_i)}_{R_i}\otimes\I_{R_jS}\right)\ct_{R_i}\,,
\ea
where the kinematical disentangler of $R_i$ reads
\ba 
\ct_{R_i}=\int_Gdg\,N_i(g)\ket{\phi(g)}\!\bra{\phi(g)}_{R_i}\otimes U^\dag_{R_jS}(g)\,.
\ea
We then find the quantum coordinate changes $\ch_{R_jS,e}^{\rm phys}\rightarrow\ch_{R_iS,e}^{\rm phys}$ in the form 
\ba 
\tilde V_{R_i\to R_j}&=&\calr_{\mathbf{H},R_j}\cdot\calr_{\mathbf{H},R_i}^{-1}\nn\\
&\approx& U_{R_iS}^\dag(g_j)\calr_{\mathbf{S},R_j}(g_j)\cdot\calr_{\mathbf{S},R_i}^{-1}(g_i)U_{R_jS}(g_i)\\
&=&V_{R_i\to R_j}(e,e)\,,
\ea
where in the second line we have made use of Lemma~\ref{lem_Hredinv}. On account of the unitary equivalence of the two pictures, we thus obtain the same transformations as before, restricted to frame orientations corresponding to the origin. 
It is also clear that one can mix the pictures, e.g.\ mapping in an obvious manner from the relational `Heisenberg picture' relative to $R_i$ into the relational `Schr\"odinger picture' of frame $R_j$.

\section{Symmetry induced QRF transformations: changes of relational observables}\label{sec_relobschanges}

\subsection{Relation-conditional frame reorientations}

The transformation in Eq.~\eqref{obstrans1} shows how to change from the reduced \emph{description} of a \emph{fixed} relational observable from one frame to another. Indeed, Theorem~\ref{thm_obs} implies that $f_{R_jS,g_i}^{\rm phys}$ on the right hand side is the ``quantum coordinate'' description of the relational observable $F_{f_{R_jS},R_i}(g_i)$ from the internal perspective of frame $R_i$. By contrast, the left hand side of Eq.~\eqref{obstrans1} is the ``quantum coordinate'' description of that \emph{same} relational observable from the internal perspective of frame $R_j$. Since the quantum coordinate maps involve conditionings on frame orientations, the different internal descriptions ultimately correspond to different gauge-fixing conditions. Changing the conditioning $\bra{\phi(g_i)}_{R_i}\to\bra{\phi(g_i')}_{R_i}$, but not the frame, amounts to \emph{configuration-independent} gauge transformations. By contrast, changing also the frame in this conditioning, i.e.\ $\bra{\phi(g_i)}_{R_i}\to\bra{\phi(g_j)}_{R_j}$, can be interpreted as \emph{relation-conditional gauge transformations} that depend on the relation between the old and new frame \cite{Krumm:2020fws}\footnote{In \cite{Krumm:2020fws}, the transformations were called relation-conditional symmetries, adhering to quantum information terminology. However, in our context, they would be the gauge transformations.} (see also \cite{Vanrietvelde:2018pgb,Vanrietvelde:2018dit,hoehnHowSwitchRelational2018,Hoehn:2018whn}). This is the quantum analog of the passive frame changes in the form of relation-conditional Lorentz ``gauge'' transformations in special relativity in Sec.~\ref{ssec_SR}.

There is, however, another kind of transformations that too corresponds to frame changes and, instead of changing the ``quantum coordinate'' description of a fixed relational observable, actually changes the relational observable.  Namely, what are the transformations from the relational observables relative to frame $R_1$ to those relative to frame $R_2$? For instance, how does one transform from $F_{f_S,R_1}(g_1)$ to $F_{f_S,R_2}(g_2)$?  For fixed $f_S$, these two families of relational observables correspond to two distinct orbits of dimension smaller or equal to $\dim\,X$  in the physical algebra $\ca_{\rm phys}$. In Sec.~\ref{ssec_symgauge}, we saw that the (configuration independent) symmetries, i.e.\ frame reorientations, defined the transformations \emph{along} these orbits. Here, we are thus looking for the transformations that take one such orbit into the other for arbitrary $f_S$, see Fig.~\ref{fig:orbits2} below. We shall now see that these frame change transformations are given by \emph{relation-conditional frame reorientations}. In particular, this will be the quantum analog of the active frame changes in the form of relation-conditional Lorentz ``symmetry'' transformations in special relativity in Sec.~\ref{ssec_SR}. Indeed, in the classical case this has also been established in both a mechanical and gauge field context in \cite{Carrozza:2021sbk}, and we shall here extend this to the quantum theory. That is, the quantum reference frame transformations in the form of quantum coordinate changes correspond to gauge transformations, while in the form of relational observable changes correspond to symmetries.

The basic idea is to transform the old reference frame $R_1$ by its relation with the new frame $R_2$. This relation is encoded in a relational observable. To build up intuition, let us consider a simple example. Suppose $R_1,R_2$ and $S$ are three particles in one dimension subject to translation invariance. Recalling Example~\ref{ex_5}, we can then consider the relational observables ${F_{q_S,R_1}(0)=q_S-q_1}$ and ${F_{q_2,R_1}(0)=q_2-q_1}$, corresponding to the position of $S,R_2$ when $R_1$ is in the origin; they are just the relative distances to particle $R_1$. In order to transform these relational observables into the corresponding ones relative to the new frame particle $R_2$, i.e.\ ${F_{q_S,R_2}(0)=q_S-q_2}$ and ${F_{q_2,R_2}(0)=q_2-q_2=0}$, we need to translate the old frame position $q_1$ by the relative distance $F_{q_2,R_1}(0)$ between the two frame particles. Note that this relative distance is a translation group-valued operator; since it only acts on the frame $R_1$ this is a relation-conditional frame reorientation. To formalize this, note that the projection onto states with $q_2-q_1=q$ can be written
\ba
F_{\ket{q}\!\bra{q}_2,R_1}(0)=\int_\mathbb{R} dx\,\ket{x}\!\bra{x}_1\otimes\ket{x+q}\!\bra{x+q}_2\otimes\I_S\,,
\ea 
where $\ket{x}_i$ are the eigenstates of $q_i$. We can thus write the translation of $R_1$ conditional on $q_2-q_1$ as a super-operator, which is nothing but a $G$-twirl for the symmetry rather than gauge group in Eq.~\eqref{obsreorient}:
\ba 
\cv_{R_1 \to R_2}(\cdot)&:=&\int_\mathbb{R}dq \int_{\mathbb{R}} dx\,\left(\ket{x}\!\bra{x}_1\otimes\ket{x+q}\!\bra{x+q}_2\otimes\I_S\right)\left(V_{R_1}(q)\otimes\I_{R_2S}\right)\left(\cdot\right)\left(V_{R_1}^\dag(q)\otimes\I_{R_2S}\right)\nn\\
&=&\int_\mathbb{R}dq \int_{\mathbb{R}} dx\,\left(\ket{x}\!\bra{x}_1\otimes\ket{x+q}\!\bra{x+q}_2\otimes\I_S\right)\left(e^{iqp_1}\otimes\I_{R_2S}\right)\left(\cdot\right)\left(e^{-iqp_1}\otimes\I_{R_2S}\right)\,.\label{obsreorient2}
\ea 
This operator is unital $\cv_{R_1\to R_2}(\I)=\I$. As one can check, it yields $\cv_{R_1\to R_2}(q_1\otimes\I_{R_2S})=q_2\otimes\I_{R_1S}$, $\cv_{R_1\to R_2}(\I_{R_1}\otimes f_{R_2S})=\I_{R_1}\otimes f_{R_2S}$ for any $R_2S$ operator $f_{R_2S}$ and in particular
\ba 
\cv_{R_1\to R_2}(F_{q_S,R_1}(0))=F_{q_S,R_2}(0)=q_S-q_2\,,\q\q \cv_{R_1\to R_2}(F_{q_2,R_1}(0))=F_{q_2,R_2}(0)=0\,,\label{relcondsym1}
\ea 
as desired, where we invoke again Example~\ref{ex_5} for the explicit form of the relational observables.

This simple example already highlights that the relation-conditional symmetry transformation $\cv_{R_1\to R_2}$ can \emph{not} be unitary; e.g.\ it maps the non-trivial operator $F_{q_2,R_1}(0)=q_2-q_1$ to zero. This is, however, not a problem as we are not transforming states, only what observable we evaluate in a given physical state. It should thus not be viewed as a physical transformation, but as a change of question that is being asked (similar to the transformation from the energy-momentum tensor $T'_{A'B'}$ associated with observer $O'$ to $T_{AB}$ associated with observer $O$ in special relativity in Sec.~\ref{ssec_SR}). In general, this can also be anticipated from the gauge-invariant subsystem relativity established in \cite{Hoehn:2021wet}: the subalgebra $\ca_{S|R_1}^{\rm phys}\subset\ca_{\rm phys}$ of relational observables $F_{f_S,R_1}(g_1)$ describing the system $S$ relative to frame $R_1$ does not coincide with the subalgebra $\ca_{S|R_2}^{\rm phys}$ of relational observables $F_{f_S,R_2}(g_2)$ describing $S$ relative to $R_2$ and this may alter commutation relations for a fixed $f_S$.

Let us now generalize this construction to a general unimodular group $G$. For simplicity, and as in the example of the translation group above, we shall restrict to the regular representation for both frames in this section. To highlight this, we drop the $\phi$ label from the frame orientation states and simply write $\ket{g}_{R_i}$, minding that we now have $\braket{g|g'}_{R_i}=\delta(g,g')$. We shall explain below why we impose this restriction here.

The configuration of the frame $R_2$ will be determined by observables of the form
$
Q^a_2=\int_Gdg\,Q^a_2(g)\,\ket{g}\!\bra{g}_2
$, where $Q^a_2(g)$, $a=1,\ldots,\dim\,G$, are some functions on the group. For example, $Q^a_2$ could be the $a^{\rm th}$ coordinate on $G$. The relation between the frames is then encoded in the relational observables
\ba 
F_{Q^a_2,R_1}(g_1)&=&\int_Gdg\,dg'Q^a_2(g')\,\ket{gg_1}\!\bra{gg_1}_1\otimes\ket{gg'}\!\bra{gg'}_2\otimes\I_S\,,
\ea 
where the projector onto eigenvalue $Q_2^a(g')$ is
\ba 
F_{\ket{g'}\!\bra{g'}_2,R_1}(e)&=&\int_Gdg\,\ket{g}\!\bra{g}_1\otimes\ket{gg'}\!\bra{gg'}_2\otimes\I_S\,.
\ea 
The generalization of the relation-conditional frame reorientation in Eq.~\eqref{obsreorient2}, taking also non-trivial frame orientations $g_1$ of $R_1$ and $g_2$ of $R_2$ into account, is (cf.\ Eq.~\eqref{obsreorient})
\ba 
\cv_{R_1 \to R_2}^{g_1 \to g_2}(\cdot):=\int_G dg\,dg'\,\left(\ket{g}\!\bra{g}_1\otimes\ket{gg'}\!\bra{gg'}_2\otimes\I_S\right)\left(V_{R_1}(g'g_2^{-1}g_1)\otimes\I_{R_2S}\right)\left(\cdot\right)\left(V_{R_1}^\dag(g'g_2^{-1}g_1)\otimes\I_{R_2S}\right)\,.\label{relcondreorient2}
\ea 
In particular, recalling Eq.~\eqref{obsreorient}, this yields for an arbitrary $R_2S$ observable $f_{R_2S}$
\ba 
\cv_{R_1 \to R_2}^{g_1 \to g_2}\left(F_{f_{R_2S},R_1}(g_1)\right) = \int_G dg\,dg'\,\left(\ket{g}\!\bra{g}_1\otimes\ket{gg'}\!\bra{gg'}_2\otimes\I_S\right)\,F_{f_{R_2S},R_1}(g_2g'^{-1})\,,\label{relcondreorient3}
\ea 
i.e.\ a relation-conditional relabeling of the parameters $g_1$ in the relational observables $F_{f_{R_2S},R_1}(g_1)$, specifying the frame orientation.
It is straightforward to verify that this super-operator is unital, $\cv_{R_1}(\I)=\I$ and produces the desired generalization of Eq.~\eqref{relcondsym1}
\ba 
\cv_{R_1 \to R_2}^{g_1 \to g_2}\left(F_{f_S,R_1}(g_1)\right)&=&F_{f_S,R_2}(g_2)\,,\nn\\\cv_{R_1 \to R_2}^{g_1 \to g_2}\left(F_{Q^a_2,R_1}(g_1)\right)&=&F_{Q_2^a,R_2}(g_2)=\int_Gdg\,U_{R_2}(g)\,Q_2^a\ket{g_2}\!\bra{g_2}_2\,U_{R_2}^\dag(g)\otimes\I_{R_1S}=Q^a_2(g_2)\,\I\,,\label{relcondreorient4}
\ea 
for arbitrary $f_S\in\ca_S$. The expression on the right hand side of the second line is a tautological observable, describing the frame $R_2$ relative to itself, and is thus proportional to the identity. It encodes $Q^a_2$ when $R_2$ is in orientation $g_2$. For example, when we set $g_2=e$ and choose $Q^a_2$ to be coordinate functions on $G$, we obtain the coordinate of the identity element in the last expression, which for the translation group is zero, cf.\ Eq.~\eqref{relcondsym1}. For this reason, the transformation fails to be unitary.

The relation-conditional frame reorientation in Eq.~\eqref{relcondreorient2} is a $G$-twirl for the symmetry group. Just like the $G$-twirl for the gauge group in Eq.~\eqref{qrelobs}, which is a frame-orientation-dressed gauge transformation, generates gauge-invariant observables, the relation-conditional frame reorientation produces symmetry-invariant observables, i.e.\ observables invariant under $R_1$ reorientations. Indeed, the output of $\cv_{R_1 \to R_2}^{g_1 \to g_2}$, as defined above, no longer depends on the old frame $R_1$, so that for an arbitrary $R_2S$ observable $f_{R_2S}$ one finds
\ba
\Big[\left(V_{R_1}^\dag(g)\otimes\I_{R_2S}\right),\cv_{R_1 \to R_2}^{g_1 \to g_2}\left(F_{f_{R_2S},R_1}(g_1)\right)\Big]=0\,,\q\forall\,g\in G\,.
\ea 

Let us consider again the tautological observables, in this case $F_{Q_1^a,R_1}(g_1)=Q_1^a(g_1)\,\I$, describing $R_1$ relative to itself. Since $\cv^{g_1\to g_2}_{R_1 \to R_2}$, as defined in Eq.~\eqref{relcondreorient2}, is unital, it maps this tautological observable into itself. We can, however, slightly modify the definition of the relation-conditional frame reorientation of relational observables: if, instead of Eq.~\eqref{relcondreorient2}, we define it as the relation-conditional orientation label transformation in Eq.~\eqref{relcondreorient3} \emph{for all} relational observables relative to $R_1$, i.e.\ including the tautological ones, then we obtain a non-trivial transformation of the latter, while maintaining the transformation of the remaining relational observables $F_{f_{R_2S},R_1}(g_1)$ as before. In particular, we then find that the tautological observable
\ba 
\cv_{R_1 \to R_2}^{g_1 \to g_2}\left(F_{Q_1^a,R_1}(g_1)\right) &=& \int_G dg\,dg'\,\left(\ket{g}\!\bra{g}_1\otimes\ket{gg'}\!\bra{gg'}_2\otimes\I_S\right)\,Q_1^a(g_2g'^{-1})\nn\\
&=&\int_Gdg\,dg'\,Q_1^a(g')\ket{gg'}\!\bra{gg'}_1\otimes\ket{gg_2}\!\bra{gg_2}_2\otimes\I_S\nn\\
&=&F_{Q_1^a,R_2}(g_2)
\ea 
transforms (upon a variable transformation) into the corresponding non-trivial relational observable encoding the relation between the two frames, however now relative to $R_2$. Defined this way, the relation-conditional  frame reorientations map all the relational observables correctly from $R_1$ to $R_2$. In this sense, it is invertible, despite being non-unitary.

This construction translates the frame change induced observable transformations of \cite{Carrozza:2021sbk} into the quantum theory. In conclusion, configuration-independent symmetries associated with a frame generate the orbits in $\ca_{\rm phys}$ corresponding to the relational observables relative to that frame, while symmetries conditional on the relation between frames $R_1$ and $R_2$ change from the observable orbits in $\ca_{\rm phys}$ associated with $R_1$ to the orbits associated with $R_2$, see Fig.~\ref{fig:orbits2} for illustration. 

\begin{figure}[t]
\centering
\includegraphics[width= 400pt]{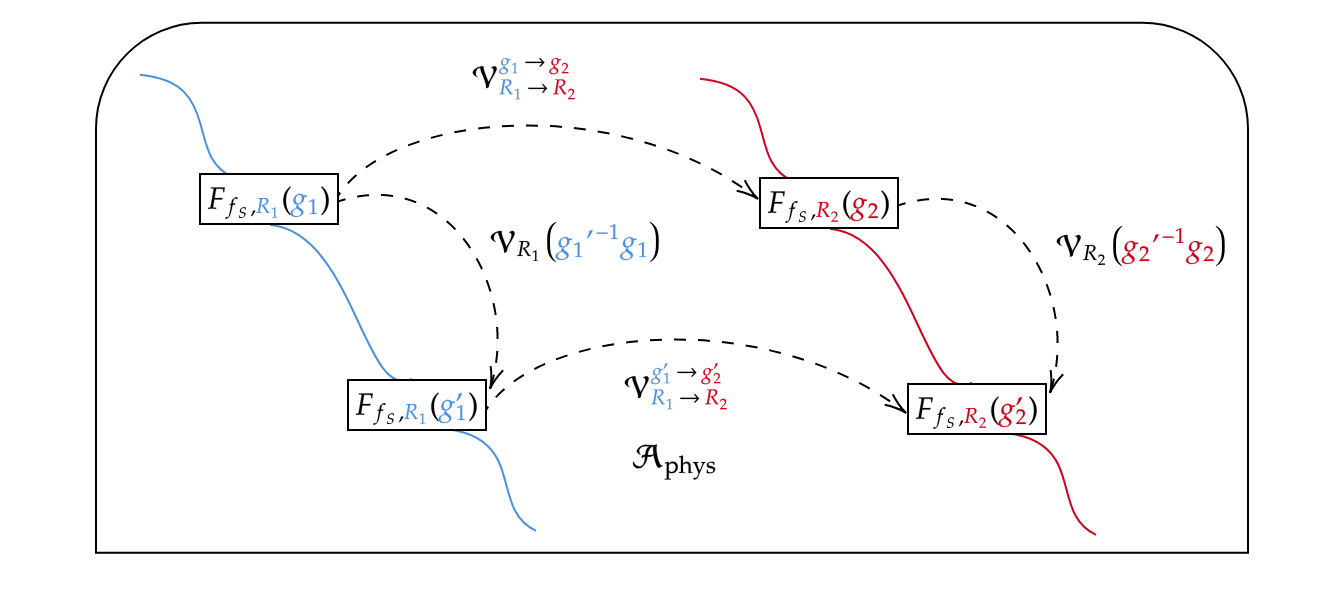}
\caption{This figure illustrate various orbits of observable in $\ca_\phys$. In blue on the left is the orbit of relational observables $\{F_{f_S,R_1}(g)\}_{g \in G}$  describing the system $S$ relative to the frame $R_1$. As can be seen, any two elements $F_{f_S,R_1}(g_1)$ and $F_{f_S,R_1}(g_1')$ in the same orbit are related by the (configuration-independent) symmetry transformation $\cv^{\rm phys}_{R_1}({g_1'} ^{-1} g_1)$ in Eq.~\eqref{obsreorient}. Similarly, on the right in blue, one has the orbit of relational observables $\{F_{f_S,R_2}(g)\}_{g \in G}$  describing the system $S$ relative to the frame $R_2$. One can map between points  $F_{f_S,R_1}(g_1)$ and $F_{f_S,R_2}(g_2)$ on different orbits with the operator $\cv_{R_1 \to R_2}^{g_1 \to g_2}$, which is a relation-conditional symmetry. \label{fig:orbits2}}
\end{figure}

Let us elucidate why we restricted to the regular representation in this section, in contrast to the remainder of this article. As we have seen, for the construction of relational observables and their algebraic properties \emph{on} the physical Hilbert space (i.e.\ their properties defined by weak equalities) it did \emph{not} matter whether the frame configuration states $\ket{\phi(g)}_R$ are orthogonal or not. Neither did it matter for the construction of the reduced theories along either the Page-Wootters or symmetry reduction procedures and, as we have seen, the quantum reference frame changes in the sense of the ``quantum coordinate transformations'' are unitary regardless. The reason is that all of these structures were defined using the gauge group and by acting on physical states, which feature a gauge symmetry induced redundancy in their description in terms of the ambient kinematical structures on $\ca_{\rm kin}$ and $\ch_{\rm kin}$. In particular, the conditioning of the invariant physical states on the non-invariant frame orientations $\ket{\phi(g)}_R$ (which also resides in the relational observables) only projects out \emph{redundant} information. This is also the reason why the reduction maps of both the relational Schr\"odinger and Heisenberg pictures are invertible when restricted to the physical Hilbert space. For this it does not matter whether the $\ket{\phi(g)}_R$ are perfectly distinguishable or not, only that they give rise to a resolution of the identity. 

Now for symmetry-induced observable transformations the situation is different. Since symmetries are \emph{physical transformations} that change the relation between the frame and the remaining degrees of freedom, a generic physical state $\ket{\psi_{\rm phys}}$, while by construction gauge-invariant, will \emph{not} be symmetry-invariant. As such, it does \emph{not} feature a symmetry related redundancy. For the relation-conditional frame reorientations it therefore does matter whether the frame orientation states are orthogonal or not, even when acting on physical states; after all, we are now conditioning on eigenstates of invariant relational observables in Eq.~\eqref{relcondreorient3} and their information certainly is not redundant. While the expression Eq.~\eqref{relcondreorient3} can be easily extended for non-regular representations by simply replacing all instances of $\ket{g}_{R_i}$ by $\ket{\phi(g)}_{R_i}$, it does not provide the correct transformation of relational observables, not even weakly. The construction procedure of $\cv_{R_1 \to R_2}$ is thus different for non-regular representations and we leave it for future work to establish it.

\subsection{Symmetries explain the quantum relativity of subsystems}\label{sec:subsys_rel}

The construction also highlights the quantum frame relativity of gauge-invariant subsystems, as revealed in \cite{Hoehn:2021wet}. The relation-conditional frame reorientation $\cv_{R_1 \to R_2}^{g_1 \to g_2}$ maps the subalgebra $\ca_{S|R_1}^{\rm phys}\subset\ca_{\rm phys}$ of relational observables, describing the system $S$ relative to the first frame $R_1$, into $\ca_{S|R_2}^{\rm phys}$, describing $S$ relative to the second frame $R_2$, i.e.
\ba 
\cv_{R_1 \to R_2}^{g_1 \to g_2}\left(\ca_{S|R_1}^{\rm phys}\right) = \ca_{S|R_2}^{\rm phys}\,.
\ea 
Note that $\cv_{R_1\to R_2}^{g_1\to g_2}$ extends for \emph{fixed} $g_1,g_2$ to the full algebras $\ca_{S|R_1}^{\rm phys}$ and $\ca_{S|R_2}^{\rm phys}$, despite Eq.~\eqref{relcondreorient4}. This is because, using Eq.~\eqref{qrelobs2}, we can rewrite any relational observable as $F_{f_S,R_1}(g)=F_{U_S(g_1g^{-1})f_SU_S^\dag(g_1g^{-1}),R_1}(g_1)$. 

For Abelian groups it was shown in \cite{Hoehn:2021wet} (also beyond the regular representation) that $\ca_{S|R_1}^{\rm phys}$ and $\ca_{S|R_2}^{\rm phys}$ are \emph{non-coinciding} subalgebras of the algebra $\ca_{\rm phys}$ of all gauge-invariant observables. Using the notion of symmetries as frame reorientations, we can now understand this observation in a simpler manner, while also extending it to general unimodular groups. Indeed, as noted before, configuration-independent reorientations of the frame $R_1$ define the transformations \emph{along} the orbits in $\ca_{S|R_1}^{\rm phys}$ defined by each family $F_{f_S,R_1}(g_1)$. That is, configuration-independent symmetries associated with $R_1$, while preserving the subalgebra $\ca_{S|R_1}^{\rm phys}$, act non-trivially on it. In particular, a generic observable in $\ca_{S|R_1}^{\rm phys}$ will fail to commute with the charges generating $R_1$ reorientations, which are elements of $\ca_{\rm phys}$. By contrast, as noted above, the frame change map $\cv_{R_1 \to R_2}^{g_1 \to g_2}$ produces $R_1$-frame-reorientation-invariant observables. That is, \emph{all} observables in $\ca_{S|R_2}^{\rm phys}$ commute with the charges of the $R_1$ reorientations. Hence, $\ca_{S|R_1}^{\rm phys}$ and $\ca_{S|R_2}^{\rm phys}$ are \emph{distinct} subalgebras of $\ca_{\rm phys}$ and can overlap at most in a subset that is invariant under reorientations of \emph{both} $R_1$ and $R_2$.

For the regular representation, one can easily check that $[F_{f_{R_2},R_1}(g_1),F_{f_S,R_1}(g_1')]=0$ for all $f_{R_2}\in\ca_{R_2}$ and $f_S\in\ca_S$. Hence, since the relational observables leave $\ch_{\rm phys}$ invariant, also their restrictions to the physical Hilbert space will commute
and so $\ca_{R_2|R_1}^{\rm phys}$ and $\ca_{S|R_1}^{\rm phys}$ constitute commuting subalgebras of $\ca_{\rm phys}$. Since they also generate $\ca_{\rm phys}$, one can infer by adapting arguments from \cite{Hoehn:2021wet}, that they induce a gauge-invariant tensor factorization in $\ca_{\rm phys}$ and thereby also $\ch_{\rm phys}$. The analogous state of affairs holds if one swaps the frames $R_1$ and $R_2$. Since $\ca_{S|R_1}^{\rm phys}$ and $\ca_{S|R_2}^{\rm phys}$ are distinct subalgebras, this also means that the relational observables relative to the two frames induce two \emph{distinct} tensor factorizations of $\ca_{\rm phys}$ and $\ch_{\rm phys}$. In this sense, the notion of gauge-invariant subsystems depends on the choice of quantum reference frame. This explains why quantum correlations and superpositions are contingent on the quantum frame \cite{giacominiQuantumMechanicsCovariance2019,Vanrietvelde:2018pgb,castro-ruizTimeReferenceFrames2019,Hoehn:2019owq,hamette2020quantum,Hoehn:2021wet}. 

As shown in \cite{Hoehn:2021wet} for Abelian groups, for non-regular representations the relational observables of, say, $R_2$ and $S$ relative to $R_1$ generally do \emph{not} generate commuting subalgebras of $\ca_{\rm phys}$ and so will \emph{not} induce gauge-invariant tensor factorizations on the perspective-neutral Hilbert space and algebra. The kinematical subsystem partitioning thus does not in general survive in a gauge-invariant manner on $\ch_{\rm phys},\ca_{\rm phys}$. We will see an example of this in Sec.~\ref{sec: U(1) example} below. This warrants a rethinking of the notion of subsystems and correlations in those more general cases.

In conclusion, the frames $R_1$ and $R_2$ refer to \emph{different} gauge-invariant degrees of freedom when they describe the kinematical system $S$. At the gauge-invariant level, $S$ is defined in a relational manner and there are multiple distinct ways in which to define a subsystem relationally. This is quantum subsystem relativity.\footnote{In fact, the same applies to the classical theory when defining subsystems relationally on phase space in terms of Poisson-commuting subalgebras of relational observables.}

\section{Examples}\label{sec_examples}

Here, we explicitly illustrate the procedure of generating a coherent state system and finding the physical Hilbert space $\ch_\phys$, and demonstrate the frame orientation dependence of the physical system Hilbert space $\ch_{S,g}^{\rm phys}$. We start by considering three simple quantum systems transforming under the symmetry group $\U(1)$. Since this is an Abelian group and thus condition \eqref{projgind} is satisfied, we find that the physical system subspace does not depend on the frame orientation. However, it constitutes an example for the quantum relativity of subsystems. We then introduce systems of coherent states for the spin $j=1$ representation of $\su2$ and give a first example of three spin $j=1$ kinematical systems. After identifying the physical state space, we find that not just the conditional states depend on the orientation of the frame, but also the physical system Hilbert space itself. However, since $\ch_\phys$ is only one-dimensional in this example, we present a final example of four spin $j=1$ kinematical systems, which displays the change of a three-dimensional physical system Hilbert space inside a 27-dimensional kinematical system Hilbert space as we change the frame orientation. More detailed explanation and proofs can be found in Appendix \ref{sec: app proofs of examples}.

\subsection{Example: U(1) and the quantum relativity of subsystems} \label{sec: U(1) example}
Let us start with a simple example of the compact Abelian group $\U(1)$. We introduce a constraint operator $\hat{C}$, requiring that physical states $\ket{\psi_\phys}$ are annihilated by it.
Consider a system consisting of two identical qubits $\A$ and $\B$, and a qutrit $\C$, subject to
\begin{equation}
      \hat{C}=\frac{\hat{\sigma}_\A^{(z)}}{2} + \frac{\hat{\sigma}_\B^{(z)}}{2} + \hat{\Sigma}_\C = \frac{1}{2}\left( \hat{\sigma}_\A^{(z)} + \hat{\sigma}_\B^{(z)} + \hat{Z}_\C \right) \ ,
\end{equation}
where $\hat{\sigma}_j^{(z)}=\ket{1}\!\bra{1}_j-\ket{-1}\!\bra{-1}_j$ is the Pauli $z$ matrix corresponding to system $j$ and $\hat{Z}_\C=2\hat{\Sigma}_\C=2\cdot\ket{2}\!\bra{2}_\C+0\cdot\ket{0}\!\bra{0}_\C-2\ket{-2}\!\bra{-2}_\C$. Here, we use the notation where eigenstates of an operator are labelled by their eigenvalue, i.e.\ $\hat{C}_i\ket{c_i}=c_i\ket{c_i}$, in contrast to using the computational basis. As a consequence, the kinematical state space of the qubits $\ch_\A$ and $\ch_\B$, respectively, is spanned by the orthonormal basis states $\{ \ket{1}, \ket{-1} \}$, while $\{ \ket{2}, \ket{0}, \ket{-2} \}$ forms an orthonormal basis for the three-level system $\ch_\C$. Hence, given two two-level and one three-level system, $\ch_{\rm kin}=\ch_\A \otimes \ch_\B \otimes \ch_\C$ and $\dim \ch_{\rm kin}=12$. The systems carry the unitary representations of $\U(1)$ generated by the three observables above.

In order to construct the system of coherent states and to express the physical states in terms of the latter, let us start from the seed states $\ket{e}_j \equiv \ket{+}_j = \frac{1}{\sqrt{2}}(\ket{1}_j+\ket{-1}_j)$ for $j=\A,\B$ and $\ket{e}_\C \equiv\ket{+}_\C=\frac{1}{\sqrt{3}}(\ket{2}_\C+\ket{0}_\C+\ket{-2}_\C)$. Then, we let the unitary representations of $\U(1)$ carried by the different systems act on these seed states. They define the coherent states:
\begin{equation}
    \ket{\phi(\theta)}_j=U_j(\theta)\ket{+}_j = e^{i\hat{\sigma}_j^{(z)}\theta}\ket{+}_j = \frac{1}{\sqrt{2}}\left(e^{i\theta}\ket{1}_j+e^{-i\theta}\ket{-1}_j \right) \label{u(1)cohqubitstate}
\end{equation}
for $j=\A,\B$. Similarly, 
\begin{equation}
    \ket{\phi(\theta)}_\C=U_\C(\theta)\ket{+}_\C=e^{i\hat{Z}\theta}\ket{+}_\C = \frac{1}{\sqrt{3}}\left(e^{i2\theta}\ket{2}_\C+\ket{0}_\C+e^{-i2\theta}\ket{-2}_\C \right) \ .
\end{equation}
One can easily check that $\ket{\phi(\theta=0)}_k=\ket{e}_k$ for $k=\A,\B,\C$. These states are chosen such that they yield resolutions of the identity on the respective subsystems, as in Eq.~\eqref{resid}. Note that the representations $U_k$ are irreducible and that the set of coherent states $\{ U_k(\theta),\ket{+}_k\}$ span $\mathbb{C}^2$ and $\mathbb{C}^3$ respectively.

The physical states $\ket{\psi_{\rm phys}} \in \ch_{\rm phys}$ need to satisfy the constraint condition $\hat{C}\ket{\psi_{\rm phys}} =0$. Identifying the subspace of $\ch_{\rm kin}$ with eigenvalue $0$, we find that $\ch_{\rm phys}$ is spanned by 
\begin{equation}
    \{ \ket{1,1,-2}_{\A\B\C}, \ket{1,-1,0}_{\A\B\C}, \ket{-1,1,0}_{\A\B\C}, \ket{-1,-1,2}_{\A\B\C}\} \ .
\end{equation}
Note that an arbitrary physical state $\ket{\psi_{\rm phys}}$ takes the form 
\begin{equation}
    \ket{\psi_{\rm phys}} =\sum_{c_\A+c_\B+c_\C=0} \psi(c_\A,c_\B,c_\C) \ket{c_\A,c_\B,c_\C}_{\A\B\C} \ . \label{u(1)physstate}
\end{equation}

Now, let us apply the Schr\"odinger reduction map onto the physical states. We will analyze both the case in which we condition on the two-level system $\A$ and the one in which we condition on the qutrit $\C$.

Conditioning on system $\A$, the two-level system figures as the reference system $R$. Then, $\calr_{\mathbf{S},\A}(\theta):\ch_{\rm phys}\rightarrow\ch_{\B\C,\theta}^{\rm phys}$ projects physical states onto the \emph{physical system Hilbert space} $\ch_{S,\theta}^{\rm phys}$ with $S=\B\C$, in which the reduced states on $\B$ and $\C$ live, given that the reference frame $\A$ has orientation $g=e^{i\theta}$.

Take an arbitrary physical state of the form \eqref{u(1)physstate} (or explicitly: $\ket{\psi_{\rm phys}} = \alpha \ket{1,1,-2}+\beta\ket{1,-1,0}+\gamma \ket{-1,1,0}+\delta \ket{-1,-1,2}$ with $|\alpha|^2+|\beta|^2+|\gamma|^2+|\delta|^2=1$) and condition on system $A$ being in orientation $\theta$. Thus, we apply the reduction map $\calr_{\mathbf{S},\A}(\theta)=\sqrt{2} \bra{\phi(\theta)}_\A \otimes \mathbf{1}_{\B\C}$. 

Making use of Eq.~\eqref{u(1)cohqubitstate}, we find 
\begin{align}
    \calr_{\mathbf{S},\A}(\theta) \ket{\psi_{\rm phys}} &= e^{-i\theta}\left( \alpha \ket{1,-2}_{\B\C} + \beta \ket{-1,0}_{\B\C} \right) + e^{i\theta}\left( \gamma \ket{1,0}_{\B\C} + \delta \ket{-1,2}_{\B\C} \right) \label{U(1)condstate} \\ &= e^{-i\theta} \sum_{c_\B+c_\C=-1} \psi(c_\B,c_\C) \ket{c_\B,c_\C}_{\B\C} + e^{i\theta} \sum_{c'_\B+c'_\C=1} \psi(c'_\B,c'_\C) \ket{c'_\B,c'_\C}_{\B\C} \ .
\end{align}
We observe that the conditional state of $\B$ and $\C$ depends on the orientation $\theta$ of system $\A$. However, the physical system subspace $\ch_{\B\C,\theta}^{\rm phys}$ does not depend on $\theta$. In fact, using Eq.~\eqref{Sproj}, the projector onto this subspace is given by
\begin{equation}
    \Pi_\A^{\rm phys}(\theta) = \int d\theta' \  _\A\braket{\phi(\theta) | \phi(\theta' + \theta) }_\A U_{\B\C}(\theta') =\int d\theta' \cos(\theta') U_{\B\C}(\theta') \ .\label{U1physproj}
\end{equation}
We can see that the physical system subspace $\ch_{\B\C,\theta}^{\rm phys}$ does not depend on $\theta$ by checking condition \eqref{projgind}:
\begin{align}
    [U_{\B\C}(\theta), \Pi_\A^{\rm phys}(\theta)] = 0 \ \forall\, e^{i\theta} \in \U(1) \ ,
\end{align}
because $U_{\B\C}(\theta)$ and $U_{\B\C}(\theta')$ commute. In fact, this is clear already from the fact that $\U(1)$ is Abelian.

As we have seen above, the reduced physical Hilbert space $\ch_{\B\C}^{\rm phys}$ is spanned by states of the form \eqref{U(1)condstate} and is thus $4$-dimensional. The kinematical space of $\B$ and $\C$, a qubit and a qutrit, is $6$-dimensional. As discussed already in Ref.~\cite{Hoehn:2021wet}, it does not hold in general that the reduced physical Hilbert space can be factorized into the same subsystems that make up the tensor factorization of the kinematical Hilbert space. The present case is an example of this:   as can be seen from the expression on the right side of Eq.~\eqref{U1physproj}, the physical system projector $\Pi_{\A}^{\rm phys}$ does \emph{not} factorize across the \emph{kinematical} subsystems $\B,\C$. It is thus \emph{not} possible to factorize $\ch^{\rm phys}_{\B\C}=\Pi_\A^{\rm phys}\left(\ch_\B\otimes\ch_\C\right)$ into a product space of the form $\ch_{\B}^{\rm phys} \otimes \ch_{\C}^{\rm phys}$, where the factors correspond to $\B$ and $\C$ ``as seen'' from $\A$'s perspective.

Conditioning on system $\C$ instead, the map $\calr_{\mathbf{S},\C}(\theta):\ch_{\rm phys}\rightarrow\ch_{\A\B}^{\rm phys}$ projects physical states onto the physical system Hilbert space $\ch_{\A\B}^{\rm phys}$, given that the reference system $\C$ has orientation $\theta$. Applying $\calr_{\mathbf{S},\C}(\theta)=\sqrt{3} \bra{\phi(\theta)}_\C \otimes \mathbf{1}_{\A\B}$ to $\ket{\psi_{\phys}}$ from above, we find 
\begin{align}
    \calr_{\mathbf{S},\C}(\theta) \ket{\psi_{\rm phys}} &= e^{-2i\theta} \alpha \ket{1,1}_{\A\B} + \beta \ket{1,-1}_{\A\B}  +  \gamma \ket{-1,1}_{\A\B}+ e^{2i\theta} \delta \ket{-1,-1}_{\A\B} \ .
\end{align}
Thus, the reduced physical Hilbert space $\ch_{\A\B}^{\rm phys}$ is $4$-dimensional, just like the kinematical Hilbert space of the two qubits $\ch_{\A}\otimes \ch_{\B}$. It follows that $\ch_{\A\B}^{\rm phys}$ can be factorized into the same systems that constitute the tensor factors of the kinematical Hilbert space and $\ch_{\A}^{\rm phys} \simeq \ch_{\B}^{\rm phys}\simeq \Cl^2$ with basis states $\{\ket{1},\ket{-1}\}$. As such, this example displays the frame-dependent factorizability mentioned in Ref.~\cite{Hoehn:2021wet}.

\subsection{Systems of $j=1$ $\SU(2)$ coherent states}\label{sec:j_1_coherent_states}

In this section, we introduce a system of $\SU(2)$ coherent states which do not admit of a unitary action by multiplication from the right and thus do not admit of symmetries as frame reorientations. These systems of coherent states will allow us to construct explicit examples exhibiting frame orientation dependence of the physical system Hilbert space.

Let us consider the spin $j=1$ representation of $\su2$ generated by the following matrices: 

\begin{align}
    \rho_{j=1}(\sigma_X) = 
    \begin{pmatrix}
        0 & \sqrt{2} & 0 \\
        \sqrt{2} & 0 & \sqrt{2} \\
        0 & \sqrt{2} & 0 
    \end{pmatrix} , \ 
      \rho_{j=1}(\sigma_Y) = i
    \begin{pmatrix}
        0 & -\sqrt{2} & 0 \\
        \sqrt{2} & 0 & -\sqrt{2} \\
        0 & \sqrt{2} & 0 
    \end{pmatrix}, \ 
     \rho_{j=1}(\sigma_Z) = 
     \begin{pmatrix}
         2 & 0 & 0 \\
         0 & 0 & 0 \\
         0 & 0 & -2
     \end{pmatrix} .
\end{align}

We use as a basis for  $\ch_R=\Cl^3$ the eigenstates of $\rho_{j=1}(\sigma_Z)$ denoted $\ket{-2}$, $\ket{0}$ and $\ket{2}$.

As shown in Appendix~\ref{app:j_1_coherent_states}, this representation can be used to define an $\SU(2)$ system of coherent states in $\Cl^3$ as follows:
\begin{align}
    \ket{\phi(e)} &= \ket{-2} +  \ket 0 + \ket 2 \\
    \ket{\phi(g = e^{i (c_X \sigma_X + c_Y \sigma_Y + c_Z \sigma_Z)})} &= e^{i (c_X \rho_{j=1}(\sigma_X) + c_Y \rho_{j=1}(\sigma_Y) + c_Z \rho_{j=1}(\sigma_Z))} (\ket{-2} +  \ket 0 + \ket 2) \ .
\end{align}
We observe that the choice of $\ket{\phi(e)}$ is not arbitrary. Certain choices have non-trivial isotropy groups; an example of such a state is $\ket 2$, which is stabilized by the $\U(1)$ subgroup $e^{i  \rho_{j=1}(\sigma_Z)t}$. By contrast, $\ket{\phi(e)}$ above features a trivial isotropy group.

It follows from Lemma~\ref{LemLR} that this system of coherent states does not admit of a unitary action by right multiplication on the coherent states, as $\dim\ch_R=3$ and $\ch_R$ already furnishes an irrep of $\SU(2)$.

\subsubsection{An example of frame orientation dependence of physical subsystems: three kinematical spin $j=1$  systems} \label{sec: SU(2) example with 3 particles}

Let us consider three quantum systems $\Cl^3$ each carrying a $j=1$ irreducible representation of $\SU(2)$. Using standard decomposition rules we have $1 \otimes 1 \otimes 1 = 0 \oplus 1^{\oplus 3} \oplus 2^{\oplus 2} \oplus 3$, showing that the physical subspace is one-dimensional.

As shown in Appendix~\ref{app:three_spin_1}, the one-dimensional physical Hilbert space is spanned by:
\begin{align}
    \ket{\psi_\phys} = \frac{1}{\sqrt{6}} ( \ket{0,-2,2} - \ket{2,-2,0}) + (\ket{2,0,-2} - \ket{0,2,-2} ) + (\ket{-2,2,0}-  \ket{-2,0,2} ) \ .
\end{align}

We can condition on the state $\ket{\phi(e)}_\A = \frac{1}{\sqrt{3}} (\ket{2}_\A + \ket{0}_\A + \ket{-2}_\A)$ to obtain the reduced state:
\begin{align}
    \ket{\psi^{\rm phys}_{\B\C}(e)} = \frac{1}{\sqrt{18}} (\ket{-2,2} - \ket{-2,0} + \ket{0,-2} - \ket{2,-2}  + \ket{2,0}-  \ket{0,2})  \ .
\end{align}

We therefore have

\begin{align}
    \ch^{\phys}_{\B\C, e} := \spann( \ket{\psi^{\rm phys}_{\B\C}(e)}) \simeq \Cl^1 \ .
\end{align}

We can obtain $\ch^{\phys}_{\B\C, g}$ by the following map: $\bra{\phi(g)}_\A \otimes \I_{\B\C}: \ch_\phys \to \ch^{\phys}_{\B\C, g}$. This is equivalent to the image of $\ch^{\phys}_{\B\C, e}$ under $\rho_{j=1}^\B(g) \otimes \rho_{j=1}^\C(g)$. 

Since $\ch^{\phys}_{\B\C, e}$ is a subspace of $\ch_B \otimes \ch_\C$ we can determine the image of $\ch^{\phys}_{\B\C, e}$ under $\rho_{j=1}^\B(g) \otimes \rho_{j=1}^\C(g)$ by decomposing $\rho_{j=1}^\B(g) \otimes \rho_{j=1}^\C(g)$ into irreducible representations of $\SU(2)$. The support of $\ch^{\phys}_{\B\C, e}$ on these irreducible representations will tell us how it transforms under the action of $\rho_{j=1}^\B(g) \otimes \rho_{j=1}^\C(g)$. We observe that if and only if $\ch^{\phys}_{\B\C, e}$ has full support on a trivial representation of $\SU(2)$ then $\ch^{\phys}_{\B\C, g}$ is independent of $g$. If $\ch^{\phys}_{\B\C, e}$ has support on a non-trivial irrep, then it will not be preserved under the action of $\rho_{j=1}^\B(g) \otimes \rho_{j=1}^\C(g)$, and hence will not be equal to $\ch^{\phys}_{\B\C, g}$ (since here the reduced spaces are one-dimensional).

Using standard representation theory, we find that $\rho_{j=1}^\B(g) \otimes \rho_{j=1}^\C(g)$ decomposes into $0 \oplus 1 \oplus 2$. By explicitly finding the bases for these irreducible representations in terms of the initial tensor product bases we can find the support of $\ket{\psi^{\rm phys}_{\B\C}(e)}$. This is done in Appendix~\ref{app:three_spin_1} where the state $\ket{\psi^{\rm phys}_{\B\C}(e)}$ is shown to have support on the spin $j=1$ irreducible representation. 
From this, it follows that 
$ \ch_{\B\C,G} = \spann(\{\ch_{\B\C,g}|g \in \SU(2)\}) \cong \Cl^3\ ,$
while $\ch^{\phys}_{\B\C, e} \simeq \Cl^1 $. Conceptually, this means that not only the conditional state $\ket{\psi_{\B\C}^{\rm phys}(g)}$ depends on the frame orientation of the reference system $\A$, but the physical system Hilbert space itself. As mentioned previously, the physical system Hilbert space in which the reduced state lives, changes inside the kinematical system Hilbert space $\ch_\B\otimes\ch_\C$. This is not very surprising in the case of this example as the physical system Hilbert space is only one-dimensional. This is why a second, more involved example is given in the next section.

Moreover, note that we can re-write $\ket{\psi^{\rm phys}_{\B\C}(e)}$ in terms of the irreducible decomposition basis as:

\begin{align}
    \ket{\psi^{\rm phys}_{\B\C}(e)} = \ket{Z=0, j=1} - \ket{Z=2, j=1} - \ket{Z=-2, j=1} \ . 
\end{align}

Under the $\SU(2)$ group action this state generates a system of $\SU(2)$ coherent states in the $j=1$ irrep of $\ch_B \otimes \ch_C$  as shown in Appendix~\ref{app:three_spin_1}.

\subsubsection{An example of frame orientation dependence of physical subsystems: four kinematical spin $j=1$  systems} \label{sec: SU(2) example with 4 particles}

We want to give another example for frame orientation dependence of physical subsystems, for which the dependence of the physical system Hilbert space on the frame orientation is even more noticeable.
Let us now consider four quantum systems $\Cl^3$, each carrying a $j=1$ irreducible representation of $\SU(2)$. Again, we can use the standard $\SU(2)$ tensor product decomposition rules to get $1 \otimes 1 \otimes 1 \otimes 1 = 0^{\oplus 3} \oplus 1^{\oplus 6} \oplus 2^{\oplus 6} \oplus 3^{\oplus 3} \oplus 4$. Thus, the physical Hilbert space is three-dimensional and all physical states satisfy the following constraint:
\begin{align}
    \rho_{j=1}(K) \otimes \I^{\otimes 3} + \I \otimes \rho_{j=1}(K) \otimes \I^{\otimes 2} + \I^{\otimes 2} \otimes \rho_{j=1}(K) \otimes \I + \I^{\otimes 3} \otimes \rho_{j=1}(K)  \ket{\psi_\phys} = 0 , \ \forall K \in \su2 \ .
\end{align}
In order to find a basis for $\ch_{\rm phys}$, we need to identify the subspace which is invariant under all three elements of $\su2$. We find that $\ch_{\rm phys}$ is spanned by three orthogonal basis vectors (see Appendix \ref{app: Four spin j=1 systems} for details), such that 
\begin{align}
    \ket{\psi_\phys}=\alpha \ket{v_1}+\beta \ket{v_2}+\gamma \ket{v_3}\ ,\ \alpha, \beta, \gamma \in \Cl,\ |\alpha|^2+|\beta|^2+|\gamma|^2=1 \ .
\end{align}

Each system has coherent states $\{\ket{\phi(g)}|g \in \SU(2)\}$ with $\ket{\phi(e)} = \ket{-2} + \ket{0} + \ket{1}$. A conditional state of the system $\B\C\D$, conditioned on the reference system $\A$ being in state $\ket{\phi(e)}$, is given by
\begin{align}
    \ket{\psi^{\rm phys}_{\B\C\D}(e)} = \alpha \ket{v_1'} + \beta \ket{v_2'} + \gamma \ket{v_3'} \ .
\end{align}
A basis for $\ch^\phys_{\B\C\D,e}$ is given by $\ket{v_i'} = \bra{\phi(e)}_\A \ket{v_i}_{\A\B\C\D}$ for $i = 1,2,3$:

The physical system Hilbert space $\ch^\phys_{\B\C\D,e}$ is three-dimensional. We now show that $\ch^\phys_{\B\C\D,g}$ is $g$-dependent. In other words, we show that $\ch^\phys_{\B\C\D,e}$ is not closed under the action of $U_\B(g) \otimes U_\C(g) \otimes U_\D(g)$.

To see this, note that the action of $U_\B(g) \otimes U_\C(g) \otimes U_\D(g)$ decomposes into irreducible representations as $1 \otimes 1 \otimes 1 = 0 \oplus 1^{\oplus 3} \oplus 2^{\oplus 2} \oplus 3$. The general strategy is to find a basis for each of the irreps and to check which of the invariant subspaces the conditional state $\ket{\psi^{\rm phys}_{\B\C\D}(e)}$ has support in. The reason for this is that each of these irreps gives rise to an invariant subspace. Applying the twirl will then rotate the state inside this invariant subspaces, eventually spanning the entire irrep subspaces the state has support in. By checking which of the subspaces the conditional state has support in, we can find the dimension of $\ch^{\rm phys}_{\B\C\D,G}=\spann(\{\ch^{\rm phys}_{\B\C\D,g}|g \in \SU(2)\})$. We prove frame orientation dependence of the physical system Hilbert space by showing that it is different from $\dim \ch^{\rm phys}_{\B\C\D,g}=3$.

In Appendix~\ref{app: Four spin j=1 systems} we show that the conditional state $\ket{\psi^{\rm phys}_{\B\C\D}(e)}$ has support in the highest weight subspaces of the representations $ 1^{\oplus 3}$,  $2^{\oplus 2}$ and $3$ of the total representation $1 \otimes 1 \otimes 1 = 0 \oplus 1^{\oplus 3} \oplus 2^{\oplus 2} \oplus 3$. This implies that  $\dim \ch^{\rm phys}_{\B\C\D,G}=26$ and hence the physical system Hilbert space is frame orientation dependent. Moreover we observe that  it is only one dimension short of the full kinematical system Hilbert space  $\ch_\B \otimes \ch_\C \otimes \ch_\D$ (which is $3^3=27$ dimensional).

Let us see what this implies conceptually. For each choice of the frame orientation of the reference system $\A$, the resulting physical system Hilbert space $\ch^{\rm phys}_{\B\C\D,g}$ is three-dimensional. However, the physical state \emph{space} of the system actually depends on the frame orientation. So, for a different choice of frame orientation of the first particle, different conditional states of the system are possible (or \emph{not} possible). All in all, almost any of the $27$ basis states of the kinematical system Hilbert space is a possible basis state for the physical system state (in fact every state except for the one spanning the $j=0$ irrep). However, for each choice of reference frame orientation at a time, the state space is only three-dimensional.

This is different from the case of $\U(1)$: there, we naturally observe the dependence of the specific conditional state on the orientation of the reference frame. However, the \emph{possibilities} of (physical) states that the system can be found in, captured by the physical system Hilbert space, are always the same and \emph{independent} of the frame orientation.

\section{Perspective-neutral vs.\ purely perspective-dependent approaches}\label{sec:perspective_neutral_vs_dependent}

The changes of quantum reference frames introduced in \cite{giacominiQuantumMechanicsCovariance2019,giacominiRelativisticQuantumReference2019,hamette2020quantum,Ballesteros:2020lgl,Mikusch:2021kro} did not make use of a perspective-neutral structure (in particular, did not invoke a gauge symmetry principle), but rather mapped directly between what are here reduced descriptions, i.e.\ frame perspectives. Let us therefore refer to these approaches as purely perspective-dependent approaches. The original changes of reference frame of \cite{giacominiQuantumMechanicsCovariance2019} were derived in \cite{Vanrietvelde:2018pgb} within the perspective-neutral approach in the context of the regular representation of the translation group. Similarly, as discussed in Sec.~\ref{sec: QRF transformations as quantum coordinate changes}, the present work recovers the quantum reference frame transformations for general groups established in \cite{hamette2020quantum} within the perspective-neutral approach in the form of quantum coordinate changes; this encompasses the transformations of \cite{giacominiQuantumMechanicsCovariance2019,Vanrietvelde:2018pgb} as a special case.

Let us elucidate, why the purely perspective-dependent approaches \cite{giacominiQuantumMechanicsCovariance2019,giacominiRelativisticQuantumReference2019,hamette2020quantum,Ballesteros:2020lgl,Mikusch:2021kro} are nonetheless \emph{not} in general equivalent to the perspective-neutral approach \cite{Vanrietvelde:2018pgb,Vanrietvelde:2018dit,hoehnHowSwitchRelational2018,Hoehn:2018whn,Hoehn:2019owq,Hoehn:2020epv,Hoehn:2021wet,Hoehn:2021flk,castro-ruizTimeReferenceFrames2019,Giacomini:2021gei} significantly generalized here:\\~

\noindent\textbf{(1)} First of all, the quantum reference frame transformations of the purely perspective-dependent approach for general groups were shown in \cite{hamette2020quantum} to be unitary if and only if the unitary representation carried by the systems are the left and right regular representations acting on states in the Hilbert space $L^2(G)$. In other words, the change between two quantum frames in \cite{hamette2020quantum} is unitary only if both are \emph{ideal} frames, i.e.\ $G$ acts regularly on itself and the frame orientations $\{\ket{\psi((g)}_{R_i}\}$ are perfectly distinguishable.\footnote{Note that \emph{ideal} and \emph{non-ideal} frames were referred to in \cite{hamette2020quantum} as \emph{perfect} and \emph{imperfect} quantum reference frames.} This also applies to the case of spin quantum reference frames proposed in \cite{Mikusch:2021kro}, which effectively constitute an ideal frame (owing to $j\to\infty$). By contrast, in the present work, the quantum coordinate changes are unitary regardless of whether the frames are ideal, complete or even incomplete. Moreover,~\cite{hamette2020quantum} showed that adopting a purely perspective-dependent approach leads to multiple (inequivalent) ways of defining changes of quantum reference frame for non-trivial isotropy groups $H$, whereas the perspective-neutral approach used here allows one to \emph{derive} the unique change of reference frame. 

These results are, however, not in contradiction because the frame change maps are not applied to the same sets of states in the two types of approaches. In the purely perspective-dependent approach of \cite{hamette2020quantum}, the frame transformations are applied to \emph{kinematical}, i.e.\ non-gauge-invariant states of the form $\ket{\phi(e)}_{R_i}\otimes\ket{\psi_{R_jS}}$,\footnote{Note that this state will also not be contained in $\ch_{\rm kin}$ for non-compact $G$, in which case $\ket{\phi(e)}_R$ is a distribution. However, the point is that these states are not coherently group averaged and thus not contained in $\ch_{\rm phys}$.} where $\ket{\psi_{R_jS}}$ is an \emph{arbitrary} state in the \emph{kinematical} system (and second frame) Hilbert space $\ch_{R_j}\otimes\ch_S$; such states are also called \emph{alignable} \cite{Krumm:2020fws,Hoehn:2021flk}. By contrast, in the perspective-neutral approach, the quantum coordinate transformations are always applied to states $\ket{\psi_{R_jS}^{\rm phys}(g)}$ in the \emph{physical} system Hilbert space $\ch_{R_jS,g}^{\rm phys}$, which is generally a strict subspace of $\ch_{R_j}\otimes\ch_S$, as explained in this manuscript (and in \cite{Hoehn:2019owq,Hoehn:2020epv,Hoehn:2021wet} for Abelian groups);\footnote{In fact, for periodic clocks ($\U(1)$ temporal reference frames), it is even possible that the physical system Hilbert space $\ch_{R_jS}^{\rm phys}$ is \emph{not} a subspace of the kinematical one $\ch_{R_j}\otimes\ch_S$, in which case one has to invoke a rigged Hilbert space construction and physical system states are distributions \cite{periodic}. This is completely analogous to the discussion of RAQ in Sec.~\ref{sec_ph} where $\ch_{\rm phys}$ is not a subspace of $\ch_{\rm kin}$ for non-compact groups, except now pertaining to the system Hilbert spaces only.} it is the subspace consistent with the physical Hilbert space and hence with gauge-invariance. As observed in \cite{Hoehn:2021wet} and the example of Sec.~\ref{sec: U(1) example}, this subspace need not even factorize into the same subsystems constituting the kinematical complement of the frame. It is only in the special case of ideal reference frames that these spaces coincide, i.e.\ that $\ch_{R_j}\otimes\ch_S=\ch_{R_jS,g}^{\rm phys}$, for all $g\in G$. Indeed, as seen in Sec.~\ref{sec: QRF transformations as quantum coordinate changes}, in this special case both approaches yield the same unitary quantum reference frame transformations. Otherwise, the internal perspective Hilbert space one would \emph{assume} in the purely perspective-dependent approach is  different to the space \emph{obtained} in the perspective-neutral approach from the physical Hilbert space.

Our exposition also explains why transformations are not unitary outside the physical Hilbert spaces (extending the observations in \cite{Hoehn:2019owq,Hoehn:2020epv,Hoehn:2021wet} for Abelian groups): as shown in Sec.~\ref{sec_3}, most of the identities crucial to the perspective-neutral approach only hold \emph{weakly}, i.e.\ only on the perspective-neutral Hilbert space $\ch_{\rm phys}$ (which is isometrically isomorphic to the physical system Hilbert spaces $\ch_{R_jS,g}^{\rm phys}$). This is due to the oft-mentioned gauge symmetry induced redundancy in the description of physical states in terms of kinematical data; since one is only removing redundant kinematical  information when conditioning (gauge-fixing) the manifestly gauge-invariant perspective-neutral states $\ket{\psi_{\rm phys}}$ on the non-invariant frame orientation, it does not matter whether these frame orientation states are perfectly distinguishable or not, only that they produce a resolution of the identity~\eqref{resid} to ensure invertibility of the quantum coordinate maps $\calr_{\mathbf{S},R_i}(g_i)$. This is in general not possible with kinematical states which feature no such gauge symmetry induced redundancy. As we have also seen, regular representations are special cases where the relevant operator identities hold also strongly, i.e.\ even on the kinematical Hilbert space $\ch_{\rm kin}$; this is because of the perfect distinguishability of ideal frame orientation states. However, also in that case, the quantum coordinate maps are only invertible on the physical Hilbert space, but importantly their image is all of $\ch_{R_j}\otimes\ch_S$.

Therefore, the perspective-neutral and perspective-dependent approaches are not equivalent for non-ideal quantum reference frames. It is only the perspective-neutral approach which yields unitary and reversible changes of QRF for such systems.
\\~

\noindent\textbf{(2)} We note also that the counter-intuitive property of frame-orientation-dependence of the physical system Hilbert space shows that \emph{one cannot in general assume that there is a fixed physical system Hilbert space from the perspective of a reference system}, as one would in the perspective-dependent approaches. One must therefore not begin from a given perspective of a reference frame $R$ when defining the relative Hilbert spaces of the physical system $\ch_{S|R}$, since this Hilbert space is in general not fixed for different orientations of $R$. Rather, one must first proceed to the perspective-neutral Hilbert space $\ch_{\rm phys}$ (which is fixed) and make use of the tools described in this paper, distinguishing kinematical and physical subsystems, to find how the frame-dependent physical system Hilbert space $\ch^{\rm phys}_{S,g}$ lies in $\ch_S$.\\~

\noindent\textbf{(3)} The fact that the perspective-dependent approaches invoke kinematical system states, rather than only physical ones, also has conceptual implications. Recall that kinematical states/observables can be interpreted as the states/observables of a system that are distinguishable/accessible relative to an external (possibly fictitious) reference frame \cite{Krumm:2020fws,Hoehn:2021flk}. Physical states/observables, on the other hand, are those that are distinguishable/accessible for an internal frame without access to external structures; all external information has been group-averaged out coherently. The perspective-dependent approaches thus tacitly invoke some background frame information and are, in this sense, arguably not fully relational beyond the regular representations. By contrast, the pillar of the perspective-neutral approach is manifest gauge-invariance; it is thus fully relational by construction. The above mentioned non-unitarity of quantum frame transformations in the perspective-dependent approaches is only detectable with at least some access to external structures. Unitarity therefore requires one to go fully relational.

\section{Conclusion and outlook}\label{sec:conclusion}

The main achievements of this work can be split into two broad categories: 
\begin{itemize}
    \item[(i)] We have extended a number of results in the quantum reference frame literature, which were known for special cases (such as specific types of reference frame), to the general case of unimodular Lie groups (incl.\ finite ones). In particular, we have expanded the perspective-neutral approach to quantum frame covariance \cite{Hoehn:2017gst,Vanrietvelde:2018pgb,Vanrietvelde:2018dit,hoehnHowSwitchRelational2018,Hoehn:2018whn,Hoehn:2019owq,Hoehn:2020epv,Hoehn:2021wet,Krumm:2020fws,Hoehn:2021flk,periodic,castro-ruizTimeReferenceFrames2019,Giacomini:2021gei,yang2020switching} to encompass such symmetry groups.
    \item[(ii)] In doing so, we have uncovered previously unknown, and even counter-intuitive, properties of quantum reference frames. 
\end{itemize}

The only previous work on quantum frame covariance for general groups \cite{hamette2020quantum} used a purely perspective-dependent framework, not invoking gauge-invariance. It established reversible changes of quantum reference frame for ideal frames only and found ambiguous transformations for non-trivial isotropy groups, the latter of which were restricted to be normal subgroups. This is consistent with the unitary frame changes for effectively ideal spin frames proposed  in \cite{Mikusch:2021kro} and for an Abelian subgroup of the Lorentz group exhibited in \cite{giacominiRelativisticQuantumReference2019,streiter2021relativistic}, both formulated within a purely perspective-dependent framework. The perspective-neutral approach, on the other hand, was so far restricted to Abelian groups (incl.\ non-regular representations) 
\cite{Vanrietvelde:2018pgb,hoehnHowSwitchRelational2018,Hoehn:2018whn,Hoehn:2019owq,Hoehn:2020epv,Hoehn:2021wet,Krumm:2020fws,Hoehn:2021flk,periodic,castro-ruizTimeReferenceFrames2019,Giacomini:2021gei,yang2020switching}
or the Euclidean group \cite{Vanrietvelde:2018dit}.
As such, the  setting of this work, namely general unimodular groups acting on complete or incomplete frames (where the only restriction on the isotropy subgroup is that it is compact), is a substantial generalization of existing results and approaches. We list some of the results which we have  generalized here to general unimodular Lie groups:
\begin{itemize}
    \item We model quantum frame orientations as fairly general systems of coherent states, subject to the group acting transitively and the orientation states comprising a resolution of the identity. 
    This includes non-trivial  isotropy groups and gives rise to a covariant POVM. 
    \item We provide a construction of  relational observables and equip these with a transparent conditional probability interpretation, extending \cite{Hoehn:2019owq,Hoehn:2020epv,periodic,Chataignier:2019kof,Chataignier:2020fys}.
    \item We generalize the Page-Wootters formalism and quantum symmetry reduction procedure, establish their unitary equivalence and show that they yield a relational ``Schr\"odinger'' and ``Heisenberg picture'', respectively. The latter constitute the internal quantum frame perspectives and since they are also unitarily equivalent to the perspective-neutral formulation, we extend the previous equivalence of these formulations for quantum clocks (``trinity of relational quantum dynamics'') \cite{Hoehn:2019owq,Hoehn:2020epv,periodic} (see also \cite{Hoehn:2021flk}) to general groups. 
    \item We derive the gauge induced changes of quantum reference frame in the form of ``quantum coordinate transformations'' for complete and incomplete reference frames and show that they are always unitary.
    \item We extend the quantum relativity of subsystems \cite{Hoehn:2021wet} to ideal reference frames for general unimodular groups, using the newly introduced concept of physical symmetries.
\end{itemize}
Let us now turn our attention to the novel physical insights generated by this work. 
\begin{itemize}
\item We illustrated how the concept of special covariance in special relativity can be understood in terms of perspective-neutral structure, finding completely analogous structures and relations as in the quantum case. This underscores that the perspective-neutral approach to quantum frame covariance, as formulated in this article for quantum mechanics, constitutes a genuine quantum analog of special covariance. In particular, our construction encompasses the Lorentz group and the quantum reference frames could therefore be the quantum version of the tetrad vectors $e^\mu_A$, which are group-valued and furnish an ideal frame. That is, the regular representation of the Lorentz group for the subsystems would establish the quantum version of the discussion in Sec.~\ref{ssec_SR}.
    \item We imported the distinction between gauge transformations and physical symmetries from gauge theories \cite{Donnelly:2016auv,Geiller:2019bti,Carrozza:2021sbk} and identified the latter as frame reorientations, following the observations in the classical case in \cite{Carrozza:2021sbk}. These frame reorientations generate orbits of relational observables that stratify the algebra of physical observables.
    \item For compact groups, we fully classified the quantum frames admitting such symmetries.
    \item Whether a reference frame admits of a unitary action of the symmetry group (in addition to the gauge group) has important physical consequences. When there is no such action, we showed the existence of a novel physical phenomenon: the physical system Hilbert space is frame orientation dependent, i.e.\ ``rotates'' through the kinematical one as the frame changes orientation. Hence, in this case, not only is the notion of subsystem dependent on the quantum frame, but also on its orientation.
     \item If a reference frame carries a non-trivial isotropy group, it can only resolve properties of other subsystems that are invariant under this isotropy group too. This can be seen as the reference frame enforcing its isotropy properties onto the other systems even if they did not originally have these properties.
\item We established a novel, active type of quantum reference frame changes induced by the symmetries that map from the relational observables relative to one frame to those relative to another. As such, they map between the different stratifications of the physical observable algebra, linking orbits corresponding to different frames.   
    \item The results generated in this work show that the purely perspective-dependent and perspective-neutral approaches are in fact inequivalent, more specifically they are inequivalent  for non-ideal reference frames.

\end{itemize}

Our construction relied on a tensor product representation of the unimodular group so that there are no coupling terms between the frames and subsystems inside the generators. This restriction must be overcome to incorporate truly general situations. For example, in gravitational systems, the generators are the Hamiltonian and diffeomorphism constraints which, owing to the universality of gravity, typically involve couplings between the various degrees of freedom, such as matter and geometric variables. As this becomes quickly very challenging, there exist only preliminary explorations for quantum frame covariance in such scenarios, see e.g.\ \cite{Hohn:2011us} for a semiclassical and \cite{castro-ruizTimeReferenceFrames2019} for a perturbative treatment. Different internal times have also been considered in quantum cosmological models with interactions in \cite{Gielen:2020abd,Gielen:2021igw}, establishing an internal time dependence of singularity resolution; however, temporal frame changes have not been explored. Interactions and different clocks have also been discussed without frame changes in the Page-Wootters formalism in \cite{Smith:2017pwx,Smith:2019imm}.

A natural next step is to  extend the perspective-neutral approach to quantum gauge field theories, given that it is formulated in the language of gauge systems and now encompasses general gauge groups as well. The seemingly most promising avenue is in the context of edge modes \cite{Donnelly:2016auv,Donnelly:2011hn,Donnelly:2014fua,Donnelly:2014gva,Geiller:2019bti,Carrozza:2021sbk,Freidel:2020xyx,Wieland:2021vef,Wieland:2017cmf,Riello:2021lfl}, as they have recently been identified as dynamical reference frames in the sense of this article \cite{Carrozza:2021sbk} and only appear on co-dimension-one surfaces, rather than the full bulk. A proper extension to gravity, on the other hand, will require grappling with the infinite-dimensional diffeomorphism group.

\medskip

\emph{Note:} During the completion of this paper, we became aware of related work by E.\ Castro-Ruiz and O.\ Oreshkov \cite{esteban}, which appears on the arXiv simultaneously.

\section*{Acknowledgements}
We thank Alexander R.\ H.\ Smith for collaboration in the initial stages of this project. This work was supported in part by funding from Okinawa Institute of Science and Technology Graduate University. This research was also supported  by Perimeter
Institute for Theoretical Physics. 
Research at Perimeter Institute is supported in part by the Government of Canada through the Department of Innovation, Science and Economic Development and by the Province of Ontario through the Ministry of Colleges and Universities. PAH is grateful for support from the Foundational Questions Institute under grant number FQXi-RFP-1801A. This work was furthermore supported by the Austrian Science Fund (FWF) through BeyondC (F7103-N48), the John Templeton Foundation (ID\# 61466) as part of The Quantum Information Structure of Spacetime (QISS) Project (qiss.fr) and the European Commission via Testing the Large-Scale Limit of Quantum Mechanics (TEQ) (No.\ 766900) project.

\bibliography{genQRFs}

\providecommand{\href}[2]{#2}\begingroup\raggedright\begin{thebibliography}{100}

\bibitem{Aharonov:1967zza}
Y.~Aharonov and L.~Susskind, ``{Charge Superselection Rule},''
  \href{http://dx.doi.org/10.1103/PhysRev.155.1428}{{\em Phys. Rev.} {\bfseries
  155} (1967) 1428--1431}.

\bibitem{Aharonov:1967zz}
Y.~Aharonov and L.~Susskind, ``{Observability of the Sign Change of Spinors
  under 2pi Rotations},''
  \href{http://dx.doi.org/10.1103/PhysRev.158.1237}{{\em Phys. Rev.} {\bfseries
  158} (1967) 1237--1238}.

\bibitem{Aharonov:1984zz}
Y.~Aharonov and T.~Kaufherr, ``{Quantum frames of reference},''
  \href{http://dx.doi.org/10.1103/PhysRevD.30.368}{{\em Phys. Rev. D}
  {\bfseries 30} (1984) 368--385}.

\bibitem{Araki:1960zz}
H.~Araki and M.~M. Yanase, ``{Measurement of Quantum Mechanical Operators},''
  \href{http://dx.doi.org/10.1103/PhysRev.120.622}{{\em Phys. Rev.} {\bfseries
  120} (1960) 622--626}.

\bibitem{Yanase161}
M.~M. Yanase, ``Optimal measuring apparatus,''
  \href{http://dx.doi.org/10.1103/PhysRev.123.666}{{\em Phys. Rev.} {\bfseries
  123} (Jul, 1961) 666--668}.

\bibitem{angeloPhysicsQuantumReference2011a}
R.~M. Angelo, N.~Brunner, S.~Popescu, A.~J. Short, and P.~Skrzypczyk, ``Physics
  within a quantum reference frame,''
  \href{http://dx.doi.org/10.1088/1751-8113/44/14/145304}{{\em J. Phys. A:
  Math. Theor.} {\bfseries 44} no.~14, (Mar., 2011) 145304}.

\bibitem{smithQuantumReferenceFrames2016}
A.~R.~H. Smith, M.~Piani, and R.~B. Mann, ``Quantum reference frames associated
  with noncompact groups: {{The}} case of translations and boosts and the role
  of mass,'' \href{http://dx.doi.org/10.1103/PhysRevA.94.012333}{{\em Phys.
  Rev. A} {\bfseries 94} no.~1, (July, 2016) 012333}.

\bibitem{miyadera2016approximating}
T.~Miyadera, L.~Loveridge, and P.~Busch, ``Approximating relational observables
  by absolute quantities: a quantum accuracy-size trade-off,''.
  \url{https://doi.org/10.1088/1751-8113/49/18/185301}.

\bibitem{loveridge2017relativity}
L.~Loveridge, P.~Busch, and T.~Miyadera, ``Relativity of quantum states and
  observables,''. \url{https://doi.org/10.1209/0295-5075/117/40004}.

\bibitem{loveridgeSymmetryReferenceFrames2018a}
L.~Loveridge, T.~Miyadera, and P.~Busch, ``Symmetry, {{Reference Frames}}, and
  {{Relational Quantities}} in {{Quantum Mechanics}},''
  \href{http://dx.doi.org/10.1007/s10701-018-0138-3}{{\em Found. Phys.}
  {\bfseries 48} no.~2, (Feb., 2018) 135--198}.

\bibitem{loveridge2011measurement}
L.~Loveridge and P.~Busch, ``‘{M}easurement of quantum mechanical
  operators’ revisited,''
  \href{http://dx.doi.org/doi.org/10.1140/epjd/e2011-10714-3}{{\em The European
  Physical Journal D} {\bfseries 62} no.~2, (2011) 297--307}.
  \url{https://link.springer.com/article/10.1140%2Fepjd%2Fe2011-10714-3}.

\bibitem{busch2011position}
P.~Busch and L.~Loveridge, ``Position measurements obeying momentum
  conservation,'' \href{http://dx.doi.org/10.1103/PhysRevLett.106.110406}{{\em
  Phys. Rev. Lett.} {\bfseries 106} (Mar, 2011) 110406}.
  \url{https://link.aps.org/doi/10.1103/PhysRevLett.106.110406}.

\bibitem{Loveridge:2019phw}
L.~Loveridge and T.~Miyadera, ``{Relative Quantum Time},''
\href{http://dx.doi.org/10.1007/s10701-019-00268-w}{{\em Found. Phys.}
  {\bfseries 49} no.~6, (2019) 549--560}.

\bibitem{loveridge2020relational}
L.~Loveridge, ``A relational perspective on the {W}igner-{A}raki-{Y}anase
  theorem,''. \url{https://doi.org/10.1088/1742-6596/1638/1/012009}.

\bibitem{Hoehn:2014vua}
P.~A. H\"ohn and M.~P. M\"uller, ``{An operational approach to spacetime
  symmetries: Lorentz transformations from quantum communication},''
  \href{http://dx.doi.org/10.1088/1367-2630/18/6/063026}{{\em New J. Phys.}
  {\bfseries 18} no.~6, (2016) 063026},
\href{http://arxiv.org/abs/1412.8462}{{\ttfamily arXiv:1412.8462 [quant-ph]}}.

\bibitem{DeWitt:1967yk}
B.~S. DeWitt, ``{Quantum Theory of Gravity. 1. The Canonical Theory},''
  \href{http://dx.doi.org/10.1103/PhysRev.160.1113}{{\em Phys. Rev.} {\bfseries
  160} (1967) 1113--1148}.

\bibitem{rovelliQuantumGravity2004}
C.~Rovelli, {\em Quantum {{Gravity}}}.
\newblock {Cambridge University Press}, {Cambridge}, 2004.

\bibitem{thiemannModernCanonicalQuantum2008}
T.~Thiemann, {\em Modern Canonical Quantum General Relativity}.
\newblock {Cambridge University Press}, 2008.

\bibitem{Tambornino:2011vg}
J.~Tambornino, ``{Relational Observables in Gravity: a Review},''
  \href{http://dx.doi.org/10.3842/SIGMA.2012.017}{{\em SIGMA} {\bfseries 8}
  (2012) 017}.

\bibitem{Rovelli:1989jn}
C.~Rovelli, ``{Time in quantum gravity: Physics beyond the Schr\"odinger
  regime},''
\href{http://dx.doi.org/10.1103/PhysRevD.43.442}{{\em Phys. Rev. D} {\bfseries
  43} (1991) 442--456}.

\bibitem{Rovelli:1990jm}
C.~Rovelli, ``{Quantum mechanics without time: a model},''
\href{http://dx.doi.org/10.1103/PhysRevD.42.2638}{{\em Phys. Rev. D} {\bfseries
  42} (1990) 2638--2646}.

\bibitem{Rovelli:1990ph}
C.~Rovelli, ``{What is observable in classical and quantum gravity?},''
\href{http://dx.doi.org/10.1088/0264-9381/8/2/011}{{\em Class. Quantum Grav.}
  {\bfseries 8} (1991) 297--316}.

\bibitem{Rovelli:1990pi}
C.~Rovelli, ``{Quantum reference systems},''
\href{http://dx.doi.org/10.1088/0264-9381/8/2/012}{{\em Class. Quantum Grav.}
  {\bfseries 8} (1991) 317--332}.

\bibitem{dittrichPartialCompleteObservables2007}
B.~Dittrich, ``Partial and complete observables for {{Hamiltonian}} constrained
  systems,'' \href{http://dx.doi.org/10.1007/s10714-007-0495-2}{{\em Gen.
  Relativ. Gravit.} {\bfseries 39} no.~11, (Nov., 2007) 1891--1927}.

\bibitem{Dittrich:2005kc}
B.~Dittrich, ``{Partial and complete observables for canonical General
  Relativity},'' \href{http://dx.doi.org/10.1088/0264-9381/23/22/006}{{\em
  Class. Quantum Grav.} {\bfseries 23} (2006) 6155--6184}.

\bibitem{Czech:2018kvg}
B.~Czech, L.~Lamprou, and L.~Susskind, ``{Entanglement Holonomies},''
  \href{http://arxiv.org/abs/1807.04276}{{\ttfamily arXiv:1807.04276
  [hep-th]}}.

\bibitem{Hardy:2018kbp}
L.~Hardy, ``{The Construction Interpretation: Conceptual Roads to Quantum
  Gravity},'' \href{http://arxiv.org/abs/1807.10980}{{\ttfamily
  arXiv:1807.10980 [quant-ph]}}.

\bibitem{Hardy:2019cef}
L.~Hardy, ``{Implementation of the Quantum Equivalence Principle},'' in {\em
  {Progress and Visions in Quantum Theory in View of Gravity}: {Bridging
  foundations of physics and mathematics}}, F.~Finster, D.~Giulini, J.~Kleiner,
  and J.~Tolksdorf, eds., pp.~189--220.
\newblock {Birkh\"{a}user Basel}, 2020.

\bibitem{Donnelly:2015hta}
W.~Donnelly and S.~B. Giddings, ``{Diffeomorphism-invariant observables and
  their nonlocal algebra},''
  \href{http://dx.doi.org/10.1103/PhysRevD.93.024030}{{\em Phys. Rev. D}
  {\bfseries 93} no.~2, (2016) 024030},
  \href{http://arxiv.org/abs/1507.07921}{{\ttfamily arXiv:1507.07921
  [hep-th]}}. [Erratum: Phys.Rev.D 94, 029903 (2016)].

\bibitem{Rovelli:2013fga}
C.~Rovelli, ``{Why Gauge?},''
  \href{http://dx.doi.org/10.1007/s10701-013-9768-7}{{\em Found. Phys.}
  {\bfseries 44} no.~1, (2014) 91--104},
  \href{http://arxiv.org/abs/1308.5599}{{\ttfamily arXiv:1308.5599 [hep-th]}}.

\bibitem{Rovelli:2020mpk}
C.~Rovelli, ``{Gauge Is More Than Mathematical Redundancy},''
  \href{http://dx.doi.org/10.1007/978-3-030-51197-5_4}{{\em Fundam. Theor.
  Phys.} {\bfseries 199} (2020) 107--110},
  \href{http://arxiv.org/abs/2009.10362}{{\ttfamily arXiv:2009.10362
  [hep-th]}}.

\bibitem{Carrozza:2021sbk}
S.~Carrozza and P.~A. H\"ohn, ``{Edge modes as reference frames and boundary
  actions from post-selection},''
  \href{http://arxiv.org/abs/2109.06184}{{\ttfamily arXiv:2109.06184
  [hep-th]}}.

\bibitem{Bartlett:2007zz}
S.~D. Bartlett, T.~Rudolph, and R.~W. Spekkens, ``{Reference frames,
  superselection rules, and quantum information},''
\href{http://dx.doi.org/10.1103/RevModPhys.79.555}{{\em Rev. Mod. Phys.}
  {\bfseries 79} (2007) 555--609}.

\bibitem{bartlett2009quantum}
S.~D. Bartlett, T.~Rudolph, R.~W. Spekkens, and P.~S. Turner, ``Quantum
  communication using a bounded-size quantum reference frame,'' {\em New J.
  Phys.} {\bfseries 11} no.~6, (2009) 063013.

\bibitem{gour2008resource}
G.~Gour and R.~W. Spekkens, ``The resource theory of quantum reference frames:
  manipulations and monotones,'' {\em New Journal of Physics} {\bfseries 10}
  no.~3, (2008) 033023.

\bibitem{gour2009measuring}
G.~Gour, I.~Marvian, and R.~W. Spekkens, ``Measuring the quality of a quantum
  reference frame: The relative entropy of frameness,'' {\em Physical Review A}
  {\bfseries 80} no.~1, (2009) 012307.

\bibitem{Palmer:2013zza}
M.~C. Palmer, F.~Girelli, and S.~D. Bartlett, ``{Changing quantum reference
  frames},''
\href{http://dx.doi.org/10.1103/PhysRevA.89.052121}{{\em Phys. Rev. A}
  {\bfseries 89} no.~5, (2014) 052121}.

\bibitem{smithCommunicatingSharedReference2019}
A.~R.~H. Smith, ``Communicating without shared reference frames,''
  \href{http://dx.doi.org/10.1103/PhysRevA.99.052315}{{\em Phys. Rev. A}
  {\bfseries 99} no.~5, (May, 2019) 052315}.

\bibitem{aaberg2014catalytic}
J.~{\AA}berg, ``Catalytic coherence,'' {\em Physical review letters} {\bfseries
  113} no.~15, (2014) 150402.

\bibitem{lostaglio2015description}
M.~Lostaglio, D.~Jennings, and T.~Rudolph, ``Description of quantum coherence
  in thermodynamic processes requires constraints beyond free energy,'' {\em
  Nature communications} {\bfseries 6} no.~1, (2015) 1--9.

\bibitem{lostaglio2015quantum}
M.~Lostaglio, K.~Korzekwa, D.~Jennings, and T.~Rudolph, ``Quantum coherence,
  time-translation symmetry, and thermodynamics,'' {\em Physical review X}
  {\bfseries 5} no.~2, (2015) 021001.

\bibitem{lostaglio2019coherence}
M.~Lostaglio and M.~P. M{\"u}ller, ``Coherence and asymmetry cannot be
  broadcast,'' {\em Physical review letters} {\bfseries 123} no.~2, (2019)
  020403.

\bibitem{marvian2019no}
I.~Marvian and R.~W. Spekkens, ``No-broadcasting theorem for quantum asymmetry
  and coherence and a trade-off relation for approximate broadcasting,'' {\em
  Physical review letters} {\bfseries 123} no.~2, (2019) 020404.

\bibitem{cwiklinski2015limitations}
P.~{\'C}wikli{\'n}ski, M.~Studzi{\'n}ski, M.~Horodecki, and J.~Oppenheim,
  ``Limitations on the evolution of quantum coherences: towards fully quantum
  second laws of thermodynamics,'' {\em Physical review letters} {\bfseries
  115} no.~21, (2015) 210403.

\bibitem{woods2019autonomous}
M.~P. Woods, R.~Silva, and J.~Oppenheim, ``Autonomous quantum machines and
  finite-sized clocks,'' {\em Annales Henri Poincar{\'e}} {\bfseries 20} no.~1,
  (2019) 125--218.

\bibitem{woods2019resource}
M.~P. Woods and M.~Horodecki, ``The resource theoretic paradigm of quantum
  thermodynamics with control,'' {\em arXiv preprint arXiv:1912.05562} (2019) .

\bibitem{Krumm:2020fws}
M.~Krumm, P.~A. H\"ohn, and M.~P. M\"uller, ``{Quantum reference frame
  transformations as symmetries and the paradox of the third particle},''
  \href{http://dx.doi.org/10.22331/q-2021-08-27-530}{{\em Quantum} {\bfseries
  5} (2021) 530}, \href{http://arxiv.org/abs/2011.01951}{{\ttfamily
  arXiv:2011.01951 [quant-ph]}}.

\bibitem{giacominiQuantumMechanicsCovariance2019}
F.~Giacomini, E.~{Castro-Ruiz}, and {\v C}.~Brukner, ``Quantum mechanics and
  the covariance of physical laws in quantum reference frames,''
  \href{http://dx.doi.org/10.1038/s41467-018-08155-0}{{\em Nat. Commun.}
  {\bfseries 10} no.~1, (Jan., 2019) 494}.

\bibitem{Bojowald:2010qw}
M.~Bojowald, P.~A. H\"ohn, and A.~Tsobanjan, ``{Effective approach to the
  problem of time: general features and examples},''
\href{http://dx.doi.org/10.1103/PhysRevD.83.125023}{{\em Phys. Rev. D}
  {\bfseries 83} (2011) 125023}.

\bibitem{Bojowald:2010xp}
M.~Bojowald, P.~A. H\"ohn, and A.~Tsobanjan, ``{An Effective approach to the
  problem of time},''
\href{http://dx.doi.org/10.1088/0264-9381/28/3/035006}{{\em Class. Quantum
  Grav.} {\bfseries 28} (2011) 035006}.

\bibitem{Hohn:2011us}
P.~A. H\"ohn, E.~Kubalova, and A.~Tsobanjan, ``{Effective relational dynamics
  of a nonintegrable cosmological model},''
\href{http://dx.doi.org/10.1103/PhysRevD.86.065014}{{\em Phys. Rev. D}
  {\bfseries 86} (2012) 065014}.

\bibitem{Hoehn:2017gst}
P.~A. H\"ohn, ``{Reflections on the information paradigm in quantum and
  gravitational physics},''
\href{http://dx.doi.org/10.1088/1742-6596/880/1/012014}{{\em J.\ Phys.\ Conf.\
  Ser.} {\bfseries 880} no.~1, (2017) 012014}.

\bibitem{Vanrietvelde:2018pgb}
A.~Vanrietvelde, P.~A. H\"ohn, F.~Giacomini, and E.~Castro-Ruiz, ``{A change of
  perspective: switching quantum reference frames via a perspective-neutral
  framework},'' \href{http://dx.doi.org/10.22331/q-2020-01-27-225}{{\em
  Quantum} {\bfseries 4} (9, 2020) 225}.

\bibitem{Vanrietvelde:2018dit}
A.~Vanrietvelde, P.~A. H\"ohn, and F.~Giacomini, ``{Switching quantum reference
  frames in the N-body problem and the absence of global relational
  perspectives},'' \href{http://arxiv.org/abs/1809.05093}{{\ttfamily
  arXiv:1809.05093 [quant-ph]}}.

\bibitem{hoehnHowSwitchRelational2018}
P.~A. H\"ohn and A.~Vanrietvelde, ``{How to switch between relational quantum
  clocks},'' \href{http://dx.doi.org/10.1088/1367-2630/abd1ac}{{\em New J.
  Phys.} {\bfseries 22} no.~12, (2020) 123048},
  \href{http://arxiv.org/abs/1810.04153}{{\ttfamily arXiv:1810.04153 [gr-qc]}}.

\bibitem{Hoehn:2018whn}
P.~A. H\"ohn, ``{Switching internal times and a new perspective on the 'wave
  function of the universe'},''
\href{http://dx.doi.org/10.3390/universe5050116}{{\em Universe} {\bfseries 5}
  (2019) 116}.

\bibitem{Hoehn:2019owq}
P.~A. H\"ohn, A.~R.~H. Smith, and M.~P.~E. Lock, ``{Trinity of relational
  quantum dynamics},''
  \href{http://dx.doi.org/10.1103/PhysRevD.104.066001}{{\em Phys. Rev. D}
  {\bfseries 104} no.~6, (2021) 066001},
  \href{http://arxiv.org/abs/1912.00033}{{\ttfamily arXiv:1912.00033
  [quant-ph]}}.

\bibitem{Hoehn:2020epv}
P.~A. H\"ohn, A.~R. Smith, and M.~P. Lock, ``{Equivalence of Approaches to
  Relational Quantum Dynamics in Relativistic Settings},''
  \href{http://dx.doi.org/10.3389/fphy.2021.587083}{{\em Front. in Phys.}
  {\bfseries 9} (2021) 181}, \href{http://arxiv.org/abs/2007.00580}{{\ttfamily
  arXiv:2007.00580 [gr-qc]}}.

\bibitem{Hoehn:2021wet}
P.~A. H\"ohn, M.~P.~E. Lock, S.~A. Ahmad, A.~R.~H. Smith, and T.~D. Galley,
  ``{Quantum Relativity of Subsystems},''
  \href{http://arxiv.org/abs/2103.01232}{{\ttfamily arXiv:2103.01232
  [quant-ph]}}.

\bibitem{Hoehn:2021flk}
P.~A. H\"ohn, M.~Krumm, and M.~P. M\"uller, ``{Internal quantum reference
  frames for finite Abelian groups},''
  \href{http://arxiv.org/abs/2107.07545}{{\ttfamily arXiv:2107.07545
  [quant-ph]}}.

\bibitem{periodic}
L.~Chataignier, P.~A. H\"ohn, and M.~P.~E. Lock, ``{Relational dynamics with
  periodic clocks},'' {\em to appear} .

\bibitem{castro-ruizTimeReferenceFrames2019}
E.~Castro-Ruiz, F.~Giacomini, A.~Belenchia, and v.~C. Brukner, ``{Quantum
  clocks and the temporal localisability of events in the presence of
  gravitating quantum systems},''
  \href{http://dx.doi.org/10.1038/s41467-020-16013-1}{{\em Nat. Commun.}
  {\bfseries 11} (2020) 2672}.

\bibitem{Giacomini:2021gei}
F.~Giacomini, ``{Spacetime Quantum Reference Frames and superpositions of
  proper times},'' \href{http://dx.doi.org/10.22331/q-2021-07-22-508}{{\em
  Quantum} {\bfseries 5} (2021) 508},
  \href{http://arxiv.org/abs/2101.11628}{{\ttfamily arXiv:2101.11628
  [quant-ph]}}.

\bibitem{yang2020switching}
J.~M. Yang, ``Switching quantum reference frames for quantum measurement,''
  {\em Quantum} {\bfseries 4} (2020) 283.

\bibitem{giacominiRelativisticQuantumReference2019}
F.~Giacomini, E.~{Castro-Ruiz}, and {\v C}.~Brukner, ``Relativistic {{Quantum
  Reference Frames}}: {{The Operational Meaning}} of {{Spin}},''
  \href{http://dx.doi.org/10.1103/PhysRevLett.123.090404}{{\em Phys. Rev.
  Lett.} {\bfseries 123} no.~9, (Aug., 2019) 090404}.

\bibitem{streiter2021relativistic}
L.~F. Streiter, F.~Giacomini, and {\v{C}}.~Brukner, ``Relativistic {B}ell test
  within quantum reference frames,'' {\em Physical Review Letters} {\bfseries
  126} no.~23, (2021) 230403.

\bibitem{hamette2020quantum}
A.-C. de~la Hamette and T.~D. Galley, ``Quantum reference frames for general
  symmetry groups,'' \href{http://dx.doi.org/10.22331/q-2020-11-30-367}{{\em
  {Quantum}} {\bfseries 4} (2020) 367}.

\bibitem{Ballesteros:2020lgl}
A.~Ballesteros, F.~Giacomini, and G.~Gubitosi, ``{The group structure of
  dynamical transformations between quantum reference frames},''
  \href{http://dx.doi.org/10.22331/q-2021-06-08-470}{{\em Quantum} {\bfseries
  5} (2021) 470}, \href{http://arxiv.org/abs/2012.15769}{{\ttfamily
  arXiv:2012.15769 [quant-ph]}}.

\bibitem{Giacomini:2020ahk}
F.~Giacomini and v.~Brukner, ``{Einstein's Equivalence principle for
  superpositions of gravitational fields},''
  \href{http://arxiv.org/abs/2012.13754}{{\ttfamily arXiv:2012.13754
  [quant-ph]}}.

\bibitem{Giacomini:2021aof}
F.~Giacomini and v.~Brukner, ``{Quantum superposition of spacetimes obeys
  Einstein's Equivalence Principle},''
  \href{http://arxiv.org/abs/2109.01405}{{\ttfamily arXiv:2109.01405
  [quant-ph]}}.

\bibitem{Barbado:2020snx}
L.~C. Barbado, E.~Castro-Ruiz, L.~Apadula, and v.~Brukner, ``{Unruh effect for
  detectors in superposition of accelerations},''
  \href{http://dx.doi.org/10.1103/PhysRevD.102.045002}{{\em Phys. Rev. D}
  {\bfseries 102} no.~4, (2020) 045002},
  \href{http://arxiv.org/abs/2003.12603}{{\ttfamily arXiv:2003.12603
  [quant-ph]}}.

\bibitem{Mikusch:2021kro}
M.~Mikusch, L.~C. Barbado, and v.~Brukner, ``{Transformation of Spin in Quantum
  Reference Frames},'' \href{http://arxiv.org/abs/2103.05022}{{\ttfamily
  arXiv:2103.05022 [quant-ph]}}.

\bibitem{Savi:2020qdl}
M.~F. Savi and R.~M. Angelo, ``{Quantum resource covariance},''
  \href{http://dx.doi.org/10.1103/PhysRevA.103.022220}{{\em Phys. Rev. A}
  {\bfseries 103} no.~2, (2021) 022220},
  \href{http://arxiv.org/abs/2005.09612}{{\ttfamily arXiv:2005.09612
  [quant-ph]}}.

\bibitem{Chataignier:2019kof}
L.~Chataignier, ``Construction of quantum {D}irac observables and the emergence
  of {WKB} time,'' \href{http://dx.doi.org/10.1103/PhysRevD.101.086001}{{\em
  Phys. Rev. D} {\bfseries 101} (Apr, 2020) 086001}.
  \url{https://link.aps.org/doi/10.1103/PhysRevD.101.086001}.

\bibitem{Chataignier:2020fys}
L.~Chataignier, ``{Relational observables, reference frames, and conditional
  probabilities},'' \href{http://dx.doi.org/10.1103/PhysRevD.103.026013}{{\em
  Phys. Rev. D} {\bfseries 103} no.~2, (2021) 026013},
  \href{http://arxiv.org/abs/2006.05526}{{\ttfamily arXiv:2006.05526 [gr-qc]}}.

\bibitem{Gielen:2020abd}
S.~Gielen and L.~Men\'endez-Pidal, ``{Singularity resolution depends on the
  clock},'' \href{http://dx.doi.org/10.1088/1361-6382/abb14f}{{\em Class.
  Quant. Grav.} {\bfseries 37} no.~20, (2020) 205018},
  \href{http://arxiv.org/abs/2005.05357}{{\ttfamily arXiv:2005.05357 [gr-qc]}}.

\bibitem{Gielen:2021igw}
S.~Gielen and L.~Men\'endez-Pidal, ``{Unitarity, clock dependence and quantum
  recollapse in quantum cosmology},''
  \href{http://arxiv.org/abs/2109.02660}{{\ttfamily arXiv:2109.02660 [gr-qc]}}.

\bibitem{Isham:1984sb}
C.~J. Isham and K.~V. Kuchar, ``{Representations of Space-time Diffeomorphisms.
  1. Canonical Parametrized Field Theories},''
  \href{http://dx.doi.org/10.1016/0003-4916(85)90018-1}{{\em Annals Phys.}
  {\bfseries 164} (1985) 288}.

\bibitem{Isham:1984rz}
C.~J. Isham and K.~V. Kuchar, ``{Representations of Space-time Diffeomorphisms.
  2. Canonical Geometrodynamics},''
  \href{http://dx.doi.org/10.1016/0003-4916(85)90019-3}{{\em Annals Phys.}
  {\bfseries 164} (1985) 316}.

\bibitem{Giulini:1998kf}
D.~Giulini and D.~Marolf, ``{A Uniqueness theorem for constraint
  quantization},'' \href{http://dx.doi.org/10.1088/0264-9381/16/7/322}{{\em
  Class. Quantum Grav.} {\bfseries 16} (1999) 2489--2505}.

\bibitem{Giulini:1998rk}
D.~Giulini and D.~Marolf, ``{On the generality of refined algebraic
  quantization},'' \href{http://dx.doi.org/10.1088/0264-9381/16/7/321}{{\em
  Class. Quantum Grav.} {\bfseries 16} (1999) 2479--2488}.

\bibitem{diracLecturesQuantumMechanics1964}
P.~A.~M. Dirac, {\em Lectures on {{Quantum Mechanics}}}.
\newblock {Belfer Graduate School of Sciencem Yeshiva University, New York},
  1964.

\bibitem{Henneaux:1992ig}
M.~Henneaux and C.~Teitelboim, {\em {Quantization of Gauge Systems}}.
\newblock Princeton University Press, Princeton,
1992.
\newblock

\bibitem{holevoProbabilisticStatisticalAspects1982}
A.~S. Holevo, {\em Probabilistic and {{Statistical Aspects}} of {{Quantum
  Theory}}}, vol.~1 of {\em Statistics and {{Probability}}}.
\newblock {North-Holland}, {Amsterdam}, 1982.

\bibitem{buschOperationalQuantumPhysics}
P.~Busch, M.~Grabowski, and P.~J. Lahti,
  \href{http://dx.doi.org/10.1007/978-3-540-49239-9}{{\em Operational {{Quantum
  Physics}}}}, vol.~31 of {\em Lecture {{Notes}} in {{Physics Monographs}}}.
\newblock Springer-Verlag, Berlin, Heidelberg, 1995.

\bibitem{Smith:2017pwx}
A.~R.~H. Smith and M.~Ahmadi, ``{Quantizing time: Interacting clocks and
  systems},''
\href{http://dx.doi.org/10.22331/q-2019-07-08-160}{{\em Quantum} {\bfseries 3}
  (2019) 160}.

\bibitem{Smith:2019imm}
A.~R.~H. Smith and M.~Ahmadi, ``{Quantum clocks observe classical and quantum
  time dilation},'' \href{http://dx.doi.org/10.1038/s41467-020-18264-4}{{\em
  Nature Commun.} {\bfseries 11} no.~1, (2020) 5360},
  \href{http://arxiv.org/abs/1904.12390}{{\ttfamily arXiv:1904.12390
  [quant-ph]}}.

\bibitem{Dittrich:2006ee}
B.~Dittrich and J.~Tambornino, ``{A perturbative approach to {D}irac
  observables and their space-time algebra},''
  \href{http://dx.doi.org/10.1088/0264-9381/24/4/001}{{\em Class. Quantum
  Grav.} {\bfseries 24} (2007) 757--784}.

\bibitem{Dittrich:2007jx}
B.~Dittrich and J.~Tambornino, ``{Gauge invariant perturbations around symmetry
  reduced sectors of general relativity: Applications to cosmology},''
  \href{http://dx.doi.org/10.1088/0264-9381/24/18/001}{{\em Class. Quantum
  Grav.} {\bfseries 24} (2007) 4543--4586}.

\bibitem{Donnelly:2016auv}
W.~Donnelly and L.~Freidel, ``{Local subsystems in gauge theory and gravity},''
  \href{http://dx.doi.org/10.1007/JHEP09(2016)102}{{\em JHEP} {\bfseries 09}
  (2016) 102}, \href{http://arxiv.org/abs/1601.04744}{{\ttfamily
  arXiv:1601.04744 [hep-th]}}.

\bibitem{Geiller:2019bti}
M.~Geiller and P.~Jai-akson, ``{Extended actions, dynamics of edge modes, and
  entanglement entropy},''
  \href{http://dx.doi.org/10.1007/JHEP09(2020)134}{{\em JHEP} {\bfseries 09}
  (2020) 134}, \href{http://arxiv.org/abs/1912.06025}{{\ttfamily
  arXiv:1912.06025 [hep-th]}}.

\bibitem{Donnelly:2011hn}
W.~Donnelly, ``{Decomposition of entanglement entropy in lattice gauge
  theory},'' \href{http://dx.doi.org/10.1103/PhysRevD.85.085004}{{\em Phys.
  Rev. D} {\bfseries 85} (2012) 085004},
  \href{http://arxiv.org/abs/1109.0036}{{\ttfamily arXiv:1109.0036 [hep-th]}}.

\bibitem{Donnelly:2014fua}
W.~Donnelly and A.~C. Wall, ``{Entanglement entropy of electromagnetic edge
  modes},'' \href{http://dx.doi.org/10.1103/PhysRevLett.114.111603}{{\em Phys.
  Rev. Lett.} {\bfseries 114} no.~11, (2015) 111603},
  \href{http://arxiv.org/abs/1412.1895}{{\ttfamily arXiv:1412.1895 [hep-th]}}.

\bibitem{Donnelly:2014gva}
W.~Donnelly, ``{Entanglement entropy and nonabelian gauge symmetry},''
  \href{http://dx.doi.org/10.1088/0264-9381/31/21/214003}{{\em Class. Quant.
  Grav.} {\bfseries 31} no.~21, (2014) 214003},
  \href{http://arxiv.org/abs/1406.7304}{{\ttfamily arXiv:1406.7304 [hep-th]}}.

\bibitem{Freidel:2020xyx}
L.~Freidel, M.~Geiller, and D.~Pranzetti, ``{Edge modes of gravity. Part I.
  Corner potentials and charges},''
  \href{http://dx.doi.org/10.1007/JHEP11(2020)026}{{\em JHEP} {\bfseries 11}
  (2020) 026}, \href{http://arxiv.org/abs/2006.12527}{{\ttfamily
  arXiv:2006.12527 [hep-th]}}.

\bibitem{Wieland:2021vef}
W.~Wieland, ``{Gravitational SL(2, \ensuremath{\mathbb{R}}) algebra on the
  light cone},'' \href{http://dx.doi.org/10.1007/JHEP07(2021)057}{{\em JHEP}
  {\bfseries 07} (2021) 057}, \href{http://arxiv.org/abs/2104.05803}{{\ttfamily
  arXiv:2104.05803 [gr-qc]}}.

\bibitem{Wieland:2017cmf}
W.~Wieland, ``{Fock representation of gravitational boundary modes and the
  discreteness of the area spectrum},''
  \href{http://dx.doi.org/10.1007/s00023-017-0598-6}{{\em Annales Henri
  Poincare} {\bfseries 18} no.~11, (2017) 3695--3717},
  \href{http://arxiv.org/abs/1706.00479}{{\ttfamily arXiv:1706.00479 [gr-qc]}}.

\bibitem{Riello:2021lfl}
A.~Riello, ``{Edge modes without edge modes},''
  \href{http://arxiv.org/abs/2104.10182}{{\ttfamily arXiv:2104.10182
  [hep-th]}}.

\bibitem{pageEvolutionEvolutionDynamics1983}
D.~N. Page and W.~K. Wootters, ``Evolution without evolution: {{Dynamics}}
  described by stationary observables,''
  \href{http://dx.doi.org/10.1103/PhysRevD.27.2885}{{\em Phys. Rev. D}
  {\bfseries 27} (1983) 2885}.

\bibitem{pageClockTimeEntropy1994}
D.~N. Page, ``Clock {{Time}} and {{Entropy}},'' in {\em Physical {{Origins}} of
  {{Time Asymmetry}}}, J.~J. Halliwell, J.~{P{\'e}rez-Mercader}, and W.~H.
  Zurek, eds.
\newblock {Cambridge University Press}, {Cambridge}, 1994.

\bibitem{giovannettiQuantumTime2015}
V.~Giovannetti, S.~Lloyd, and L.~Maccone, ``{Quantum Time},''
  \href{http://dx.doi.org/10.1103/PhysRevD.92.045033}{{\em Phys. Rev. D}
  {\bfseries 92} no.~4, (2015) 045033},
  \href{http://arxiv.org/abs/1504.04215}{{\ttfamily arXiv:1504.04215
  [quant-ph]}}.

\bibitem{Gomes:2019xhu}
H.~Gomes, ``{Gauging the Boundary in Field-space},''
  \href{http://dx.doi.org/10.1016/j.shpsb.2019.04.002}{{\em Stud. Hist. Phil.
  Sci. B} {\bfseries 67} (2019) 89--110},
  \href{http://arxiv.org/abs/1902.09258}{{\ttfamily arXiv:1902.09258
  [physics.hist-ph]}}.

\bibitem{Gomesgauge}
H.~Gomes, ``Gauge-invariance and the empirical significance of symmetries,''
  2020.
\newblock \url{http://philsci-archive.pitt.edu/16981/}.

\bibitem{Freitag}
E.~Freitag, {\em Unitary representations of the Poincar\'e group}.
\newblock Self-Publishing, 2016.

\bibitem{Glasner}
S.~Glasner, {\em Proximal Flows}.
\newblock Springer, 2006.

\bibitem{Schroeck}
F.~E.~J. Schroeck, {\em Quantum Mechanics on Phase Space}.
\newblock Kluwer Academic Publishers, 1996.

\bibitem{Cassinelli}
G.~Cassinelli, E.~De~Vito, P.~J. Lahti, and A.~Levrero, {\em The Theory of
  Symmetry Actions in Quantum Mechanics}.
\newblock Springer, 2004.

\bibitem{Schottenloher}
M.~Schottenloher, {\em A Mathematical Introduction to Conformal Field Theory}.
\newblock Springer Verlag, Berlin Heidelberg, 2008.

\bibitem{Perelomov}
A.~Perelomov, {\em Generalized Coherent States and Their Applications}.
\newblock Springer Verlag, Berlin Heidelberg, 1986.

\bibitem{Brunetti:2009eq}
R.~Brunetti, K.~Fredenhagen, and M.~Hoge, ``{Time in quantum physics: From an
  external parameter to an intrinsic observable},''
  \href{http://dx.doi.org/10.1007/s10701-009-9400-z}{{\em Found. Phys.}
  {\bfseries 40} (2010) 1368--1378},
  \href{http://arxiv.org/abs/0909.1899}{{\ttfamily arXiv:0909.1899 [math-ph]}}.

\bibitem{Nachbin}
L.~Nachbin, {\em The Haar Integral}.
\newblock D. van Nostrand Company, Inc., Rochester, New York, USA, 1956.

\bibitem{gel2016generalized}
I.~M. Gel'fand and N.~Y. Vilenkin, {\em Generalized functions, Volume 4:
  Applications of Harmonic Analysis}, vol.~380.
\newblock American Mathematical Soc., 2016.

\bibitem{schwartz1950theorie}
L.~Schwartz, {\em Th{\'e}orie des distributions: Tome 1. Par L. Schwartz}.
\newblock Hermann (Chartres), 1950.

\bibitem{Potocek2015}
V.~Poto{\v{c}}ek and S.~M. Barnett, ``On the exponential form of the
  displacement operator for different systems,''
  \href{http://dx.doi.org/10.1088/0031-8949/90/6/065208}{{\em Phys. Script.}
  {\bfseries 90} no.~6, (2015) 065208}.

\bibitem{Dittrich:2016hvj}
B.~Dittrich, P.~A. H\"ohn, T.~A. Koslowski, and M.~I. Nelson, ``{Can chaos be
  observed in quantum gravity?},''
  \href{http://dx.doi.org/10.1016/j.physletb.2017.02.038}{{\em Phys. Lett. B}
  {\bfseries 769} (2017) 554--560},
  \href{http://arxiv.org/abs/1602.03237}{{\ttfamily arXiv:1602.03237 [gr-qc]}}.

\bibitem{Dittrich:2015vfa}
B.~Dittrich, P.~A. H\"ohn, T.~A. Koslowski, and M.~I. Nelson, ``{Chaos, Dirac
  observables and constraint quantization},''
  \href{http://arxiv.org/abs/1508.01947}{{\ttfamily arXiv:1508.01947 [gr-qc]}}.

\bibitem{Hoehn:2014fka}
P.~A. H\"ohn, ``{Quantization of systems with temporally varying discretization
  I: Evolving Hilbert spaces},''
  \href{http://dx.doi.org/10.1063/1.4890558}{{\em J. Math. Phys.} {\bfseries
  55} (2014) 083508}, \href{http://arxiv.org/abs/1401.6062}{{\ttfamily
  arXiv:1401.6062 [gr-qc]}}.

\bibitem{superselection_kitaev_2004}
A.~Kitaev, D.~Mayers, and J.~Preskill, ``Superselection rules and quantum
  protocols,'' \href{http://dx.doi.org/10.1103/physreva.69.052326}{{\em
  Physical Review A} {\bfseries 69} no.~5, (May, 2004) }.
  \url{http://dx.doi.org/10.1103/PhysRevA.69.052326}.

\bibitem{pageTimeInaccessibleObservable1989}
D.~N. Page, ``Time as an {{Inaccessible Observable}},'' {\em NSF-ITP-89-18}
  (1989) .

\bibitem{kucharTimeInterpretationsQuantum2011a}
K.~V. Kucha{\v r}, ``Time and interpretations of quantum gravity,''
  \href{http://dx.doi.org/10.1142/S0218271811019347}{{\em Int. J. Mod. Phys. D}
  {\bfseries 20} no.~supp01, (July, 2011) 3--86}.

\bibitem{Dolby:2004ak}
C.~E. Dolby, ``{The Conditional probability interpretation of the Hamiltonian
  constraint},'' \href{http://arxiv.org/abs/gr-qc/0406034}{{\ttfamily
  arXiv:gr-qc/0406034}}.

\bibitem{Gambini:2008ke}
R.~Gambini, R.~A. Porto, J.~Pullin, and S.~Torterolo, ``{Conditional
  probabilities with Dirac observables and the problem of time in quantum
  gravity},''
\href{http://dx.doi.org/10.1103/PhysRevD.79.041501}{{\em Phys. Rev. D}
  {\bfseries 79} (2009) 041501}.

\bibitem{esteban}
E.~Castro-Ruiz and O.~Oreshkov, ``Relative subsystems and quantum reference
  frame transformations,'' {\em simultaneous appearance on the arXiv} (2021) .

\end{thebibliography}\endgroup
\bibliographystyle{utphys}

\appendix

\section{Proof of Example~\ref{Ex1Compact}}
\label{SecProofEx1}
As in the statement of Example~\ref{Ex1Compact}, label the inequivalent irreps of $G$ that appear in the decomposition of $U_R(g)$ on $R$ by $q$. Then, every normalized state $\ket{\phi(e)}$ can be decomposed as
\[
   \ket{\phi(e)}=\bigoplus_q e^{i\theta_q}\beta_q\ket{\phi(e)_q},
\]
where $\beta_q\geq 0$, $\theta_q\in\mathbb{R}$, and the $\ket{\phi(e)_q}$ are normalized. Furthermore, via Schmidt decomposition, there are non-negative coefficients $\lambda_i^{(q)}$ with $\sum_i \lambda_i^{(q)}=1$, and orthonormal systems $\{\ket{i}_{\mathcal{M}^{(q)}}\}\subset\mathcal{M}^{(q)}$ and $\{\ket{i}_{\mathcal{N}^{(q)}}\}\subset\mathcal{N}^{(q)}$ such that
\[
   \ket{\phi(e)_q}=\sum_{i=1}^{d_q} \sqrt{\lambda_i^{(q)}}\ket{i}_{\mathcal{M}^{(q)}}\otimes\ket{i}_{\mathcal{N}^{(q)}},
\]
where $d_q:=\min\{\dim\,\mathcal{M}^{(q)},\dim\,\mathcal{N}^{(q)}\}$. For the sake of this proof, let us normalize the Haar measure such that $\int_G dg=1$. Then, using Schur's lemma, we obtain
\begin{eqnarray*}
\int_G dg U_R(g)\ket{\phi(e)}\bra{\phi(e)} U_R(g)^\dagger&=&\bigoplus_q \beta_q^2 \int_G dg U_R^{(q)}(g)\otimes\mathbf{1}^{(q)}\ket{\phi(e)_q}\bra{\phi(e)_q} U_R^{(q)}(g)^\dagger\otimes\mathbf{1}^{(q)}\\
&=& \bigoplus_q \beta_q^2 \sum_{i,j=1}^{d_q} \sqrt{\lambda_i^{(q)}\lambda_j^{(q)}} \int_G dg U_R^{(q)}(g)\ket{i}\bra{j}_{\mathcal{M}^{(q)}} U_R^{(q)}(g)^\dagger\otimes \ket{i}\bra{j}_{\mathcal{N}^{(q)}}\\
&=& \bigoplus_q \beta_q^2 \frac{\mathbf{1}_{\mathcal{M}^{(q)}}}{\dim \mathcal{M}^{(q)}}\otimes \sum_{i=1}^{d_q}\lambda_i^{(q)} \ket{i}\bra{i}_{\mathcal{N}^{(q)}}.
\end{eqnarray*}
For the right-hand side to become a positive multiple of the identity, we need that $\beta_q>0$ and that all $\lambda_i^{(q)}$ for a fixed $q$ are equal, hence $(\lambda_i^{(q)})^{-1}=d_q$. Furthermore, the sum on the right-hand side must be proportional to the identity on $\mathcal{N}^{(q)}$, thus $d_q=\dim\,\mathcal{N}^{(q)}$. This proves that the inequality~\eqref{eqNecessaryId} is a necessary condition. And then, for the expression above to equal $n\cdot\mathbf{1}$ for some $n>0$, we need $\beta_q^2/(\dim\,\mathcal{M}^{(q)}\dim\,\mathcal{N}^{(q)})$ to be constant in $q$. Substituting into the definition of $\ket{\phi(e)}$, and absorbing the phases $e^{i\theta_q}$ into the definition of $\ket{i}_{\mathcal{M}}$ proves that~\eqref{eqNecessaryId2} is a necessary condition, too.

Conversely, direct calculations easily shows that the given conditions are sufficient for obtaining a resolution of the identity.

\section{Proof of Lemma~\ref{LemLR}}\label{app:LR_coherent}

In this section we prove Lemma~\ref{LemLR} via a sequence of lemmas.

Let us observe that systems of coherent states $\{\ket{\phi(g)},U_R,V_R\}$ are such that $\ch$ contains at least one vector which is invariant under $U_R(g)V_R(g)$ for all $g \in G$, namely $\ket{\phi(e)}$.

We first find which irreps admit of a vector which is invariant under $U_R(g)V_R(g)$ for all $g \in G$ and then show that this vector, which is unique up to a constant, does indeed generate a system of LR coherent states when acted on by the irreps.

Finally we extend these lemmas to the case where $V_R$ and $W_R$ are reducible, providing a full classification of all systems of LR coherent states for compact $G$.

\begin{lemma}
	The irreducible representations $\rho_1 \otimes \rho_2$ of $G \times G$ which admit of a vector $v \in \ch_1 \otimes \ch_2$ which is invariant under the action $\rho_1(g) \otimes \rho_2(g)$ of $G \subset G \times G$ are those where $\ch_2 \cong \bar \ch_1$ for all irreducible representations $\ch_1$ of $G$.
\end{lemma}
\begin{proof}
An isomorphism of representations is given by a vector space isomorphism $\phi$ which is $G$ equivariant.
	
	Observe that $ \ch_1 \otimes  \ch_2 \cong  \ch_1 \otimes \bar  \ch_2^* \cong \Hom(\bar \ch_2, \ch_1)$ as representations. Here $\ch^*$ is the dual representation, while $\bar \ch$ is the complex conjugate representation. The dual representation $\rho^*(g)$ is related to the representation $\rho(g)$ as $\rho^*(g) f = f \rho(g^{-1}$ for all $f \in \ch^*$. The complex conjugate representation is just obtained by taking the complex conjugate $\bar \rho(g)$ of the matrix $\rho(g)$, where we remember that the space $\bar \ch$ has the same elements as $\ch$, but the scalar product is defined as $\lambda \cdot v = \bar \lambda v$. The space $\Hom(\bar \ch_2, \ch_1)$ is the space of linear maps from $\bar \ch_2$ to $\ch_1$. Let us label the isomorphism $\varphi: \ch_1 \otimes  \ch_2 \mapsto  \Hom(\bar \ch_2, \ch_1)$.

	The representation $\Hom(\bar \ch_2,\ch_1)$ carries a representation $\rho_3(g)$ which acts as $(\rho_3(g) M) v_2 = \rho_1(g) (M \bar \rho_2(g^{-1}) v_2)$ with $v_2 \in \bar \ch_2$ and $M \in \Hom(\bar \ch_2,\ch_1)$. 
	Since $\varphi$ is equivariant we have that $\phi((\rho_1(g) \otimes \rho_2(g)) \cdot) = \rho_3(g)\phi(\cdot)$.

	We now show that there exists an element in $\Hom(\bar \ch_2,\ch_1)$ which is invariant under the action of $\rho_3(g)$. This is therefore an element of $\Hom_G(\bar \ch_2,\ch_1)$.

	Observe that the system of LR coherent states admits of a state $\ket{\phi(e)} \in \ch_1 \otimes \ch_2$ which is invariant under the action of $\rho_1(g) \otimes \rho_2(g)$. The image of this under $\varphi$ is given by $M_{\ket{\phi(e)}}$. Since $\varphi$ is equivariant  this implies that $\rho_3(g) M_{\ket{\phi(e)}} = M_{\ket{\phi(e)}}$ for all $g \in G$.
	
	Since we have assumed $\rho_1$ and $\rho_2$ (and therefore $\bar \rho_2$) to be irreducible we can use Schur's lemma to conclude that $\rho_1 \cong \bar \rho_2$. Schur's lemma states that for complex irreducible representation $\ch$ and $\ch'$ then the space $\dim \ \Hom_G(\ch,\ch') = 1$ when $\ch \cong \ch'$ and $\dim \ \dim \ \Hom_G(\ch,\ch') = 0$ otherwise.
\end{proof}

We have established that the only irreducible representations which admit of an invariant vector under $G$ are those of the form $\ch \otimes \bar \ch$. We now find an explicit form of this vector.

\begin{lemma}
	Given a Hilbert space $\ch \otimes \bar \ch$ carrying a representation $\rho(g) \otimes \bar \rho(h)$ of $G \times G$ then the vector $v(e) = \sum_i e_i \otimes e_i$, where $\{e_i\}_i$ an orthonormal basis for $V$ and $V^*$ is invariant under $\rho(g) \otimes \bar \rho(g)$
\end{lemma}

\begin{proof}
Using the isomorphism $\phi: \ch_1 \otimes  \ch_2 \mapsto  \Hom(\bar \ch_2, \ch_1)$ we concluded that for irreducible representations the image $\phi(\ket{\phi(e)}) = M_{\ket{\phi(e)}}$ was proportional to the identity matrix. Hence it is of the form
\begin{align}
     M_{\ket{\phi(e)}} = c \sum_i  e_i \otimes e_i^*
\end{align}
with $\{e_i\}$ an orthonormal basis for $\ch_1$ and $\{e_i^*\}$ the dual basis for $\bar \ch_2^* \cong \ch_2$. 

The explicit form of the equivariant isomorphism $\varphi$ is $\varphi: e_i \otimes e_j \mapsto e_i \otimes e_j^*$. One can check equivariance explicitly:
\begin{align}
    \phi((\rho(g) \otimes \bar \rho(g)) e_i \otimes e_j) &= \phi(\rho(g)e_i \otimes \bar \rho(g)e_j)\nn \\
    &= \rho(g) e_i \otimes (\bar \rho(g) e_j)^* \nn\\
    &= \rho(g) e_i \otimes (e_j)^* \bar \rho^T(g)\nn\\
    &=  \rho(g) e_i \otimes e_j^*  \rho(g^{-1})\nn\\
    &= \rho(g) e_i \otimes \rho^*(g)e_j^*  
\end{align}
Applying the inverse of $\phi$ to this explicit decomposition of  $M_{\ket{\phi(e)}}$ gives:
\begin{align}
    \ket{\phi(e)} = c \sum_i  e_i \otimes e_i
\end{align}
with $\{e_i\}$ an orthonormal basis for $\ch_1$ and $\ch_2 \cong \bar \ch_1$.
\end{proof}

\begin{lemma}
    Given $\ch \otimes \bar \ch$ carrying an representation $\rho(g) \otimes \bar \rho(h)$ of $G \times G$ then the invariant vector $\ket{\phi(e)}$ generates a system of LR coherent states.  
\end{lemma}

\begin{proof}
Let us generate the set of states $\ket{\phi(g)} = U(g) \otimes \mathbf{1} \ket{\phi(g)}$. We now show that $U(g') \otimes \bar U(\tilde g)$ acts on $\ket{\phi(g)}$ as the actions from the left and the right:
\begin{align}
    (\rho(g') \otimes \bar \rho(\tilde g))\ket{\phi(g)} &=   (\rho(g') \otimes\bar \rho(\tilde g)) \rho(g) \otimes \mathbf{1} \ket{\phi(e)} \nn \\
    &= \rho(g'g) \otimes \bar \rho(\tilde g)  \ket{\phi(e)} \nn \\
    &= \rho(g'g)\rho^\dagger(\tilde g) \otimes \mathbf{1} \ket{\phi(e)}  \nn \\
    & = \rho(g'g\tilde g^{-1})\ket{\phi(e)}  \nn \\
    & = \ket{\phi(g'g\tilde g^{-1})}
\end{align}
where we made use of the fact that $\mathbf{1} \otimes \bar \rho(g) \ket{\phi(e)} = \rho^\dagger(g) \otimes \mathbf{1}\ket{\phi(e)}$.
\end{proof}

We now extend these proofs to reducible representations.  

Given a system of coherent states with a reducible representation $W_R(g,g') = U_R(g) V_R(g')$, we note that such a representation decomposes into irreducible representations as
\begin{align}
    W_R(g,g') = \bigoplus_i W_k(g,g') \otimes \mathbf{1}_{\cn_i} = \bigoplus_i \rho_1^i(g) \otimes \rho_2^i(g') \otimes \mathbf{1}_{\cn_i}\ ,
\end{align}
where $ \mathbf{1}_{\cn_i} $ is the identity on the multiplicity factor $\cn_i$ of each irreducible representation $W_i$.

Moreover the seed reference state $\ket{\phi(e)}_R$ decomposes as:
\begin{align}
  \ket{\phi(e)}_R = \sum_i \alpha_i \ket{\phi(e)}_i  \ , 
\end{align}
where each $\ket{\phi(e)}_i$ is necessarily invariant under $W_i(g,g)$.

Vectors in $\ch_1^i \otimes \ch_2^i \otimes \ch_{\cn_i}$ invariant under  $W_i(g,g)$ only exist when $\ch_1^i \cong \bar \ch_2^i$ by the above lemma. 

This shows that the general form of the representation is 
\begin{align}\label{eq-LR-rep2}
    W_R(g,k) \simeq \bigoplus_{i \in \ci} \rho_i(g) \otimes \bar \rho_i(k) \otimes \mathbf{1}_{\cn_i} \ ,
\end{align}
However we now show there cannot be any repeated representations in $\ci$, namely that each $\cn_i$ is one dimensional.

Observe that when decomposing the representation $U_R(g)$ one has
\begin{align}
    U_R(g) = \bigoplus_i \rho_i(g) \otimes \mathbf{1}_{\bar \ch_i} \otimes \mathbf{1}_{\cn_i} \ ,
\end{align}
where $\bar \ch_i \otimes \cn_i$ is the multiplicity space of the irreducible representation $\ch_i$.

By Example~\ref{Ex1Compact} one has that the dimension of $\ch_i$ must be greater or equal to the dimension of $\bar \ch_i \otimes \cn_i$. Since $\dim \ \ch_i= \dim \ \bar \ch_i$ this implies that $\cn_i$ is trivial (one dimensional).

Let us show that there given a representation of the form given in Equation~\eqref{eq-LR-rep2} and a seed state of the form 
\begin{align}
    \ket{\phi(e)}_R = \sum_{i \in \ci} \alpha_i \ket{\phi(e)}^i  \ ,
\end{align}
the requirement 
\begin{align}
    \int_{g \in G} \ketbra{\phi(g)}{\phi(g)}_R dg =  n \mathbf{1}_R \ ,
\end{align}
constrains the possible coefficients.
For arbitrary coefficients one has (following the proof of Example~\ref{Ex1Compact}):
\begin{align}
    \int_{g \in G} \ketbra{\phi(g)}{\phi(g)}_R dg = 
    \bigoplus \frac{1}{\dim(\ch_i \otimes \bar \ch_i)}|\alpha_i|^2 \mathbf{1}_{(\ch_i \otimes \bar \ch_i)} 
\end{align}
Requiring this to be equal to $ n \mathbf{1}_R $ constrains the possible values $\{\alpha_i\}_{i \in \ci }$ proving the last part of the lemma.

\section{Spin $j =1$ $\SU(2)$ coherent states }\label{app:j_1_coherent_states}

Let us consider a quantum system $\Cl^3$ acted on by $\SU(2)$. We choose as generators of $\su2$ the following:

\begin{align}
    \rho_{j=1}(\sigma_X) = 
    \begin{pmatrix}
        0 & \sqrt{2} & 0 \\
        \sqrt{2} & 0 & \sqrt{2} \\
        0 & \sqrt{2} & 0 
    \end{pmatrix} , \ 
      \rho_{j=1}(\sigma_Y) = i
    \begin{pmatrix}
        0 & -\sqrt{2} & 0 \\
        \sqrt{2} & 0 & -\sqrt{2} \\
        0 & \sqrt{2} & 0 
    \end{pmatrix}, \ 
     \rho_{j=1}(\sigma_Z) = 
     \begin{pmatrix}
         2 & 0 & 0 \\
         0 & 0 & 0 \\
         0 & 0 & -2
     \end{pmatrix} .
\end{align}

Let us find a ray $\ket \psi \in \Cl^3$ which is not annihilated by any element $c_X \rho_{j=1}(\sigma_X) + c_Y \rho_{j=1}(\sigma_Y) + c_Z \rho_{j=1}(\sigma_Z)$ of $\su2$. This state $\ket \psi$ will therefore have a trivial stabilizer subgroup under the action of $\SU(2)$ generated by this representation of $\su2$.

Consider the action of an arbitrary element $c_X \rho_{j=1}(\sigma_X) + c_Y \rho_{j=1}(\sigma_Y) + c_Z \rho_{j=1}(\sigma_Z)$ on the state $ \ket{-2} +  \ket 0 + \ket 2$:

\begin{align}
    & \ket{-2} +  \ket 0 + \ket 2 \mapsto -2 c_Z  \ket{-2}  + 2 c_Z \ket 2 \nn \\
    + & c_X \sqrt 2 \ket 0 + c_X \sqrt{2} (\ket{-2} + \ket 2 ) + c_X \sqrt 2 \ket 0 \nn \\
    + & i c_Y \sqrt 2 \ket 0 + i c_Y \sqrt{2} (-\ket{2} + \ket{-2}) -i c_Y \sqrt 2 \ket 0 \nn \\
    = & \ket{2} (2 c_Z + \sqrt 2 c_X - i \sqrt 2 c_Y) + \ket 0 (c_X \sqrt{2} ) + \ket{-2} (- 2 c_Z + c_X \sqrt{2} + i c_Y \sqrt{2}) \ . \label{eq:zero_state}
\end{align}

The vector $\ket{-2} +  \ket 0 + \ket 2$ has a trivial isotropy subgroup under the representation $\rho_{j=1}$ of $\su2$ if and only if it is stabilized by only the identity. Observe that a vector $v$ is invariant under $e^{i Xt}$ if and only if $X v = 0$. Therefore $\ket{-2} +  \ket 0 + \ket 2$ has a trivial isotropy subgroup if and only if the only  element $c_X \rho_{j=1}(\sigma_X) + c_Y \rho_{j=1}(\sigma_Y) + c_Z \rho_{j=1}(\sigma_Z)$ of the Lie algebra which maps it to $0$ is the element where $c_X =c_Y =c_Z = 0$.

Setting this state of equation~\eqref{eq:zero_state} to $0$ gives $c_X =0$ from the middle term. The first term gives $ c_Z = i \frac{\sqrt 2}{2} c_Y$ and the last term gives $c_Z = i  \frac{\sqrt 2}{2} c_Y$. These two together imply that $c_Z = c_Y = 0$ for $c_Z,c_Y,c_X \in \Rl$. We note that when  $c_Z,c_Y,c_X \in \Cl$ (i.e. the algebra is $\sl2c$ instead of $\mathfrak{su}_2$) then the condition $c_Z = i  \frac{\sqrt 2}{2} c c_Y$ can be met, showing that the vector has a one parameter isotropy subgroup for $\sl2c$, but a trivial isotropy subgroup for $\su2$.

Let us check that the action of the group on $ \ket{-2} +  \ket 0 + \ket 2$ leaves the ray invariant, and not just the vector. This amounts to showing that there is no set of coefficients $c_Z,c_Y,c_X $ (where at least one is non-zero) such that $c_X \rho_{j=1}(\sigma_X) + c_Y \rho_{j=1}(\sigma_Y) + c_Z \rho_{j=1}(\sigma_Z) (\ket{-2} +  \ket 0 + \ket 2) = C (\ket{-2} +  \ket 0 + \ket 2)$. This is equivalent to the following equalities:

\begin{align}
    2 c_Z + \sqrt 2 c_X - i \sqrt 2 c_Y = c_X \sqrt{2} \ , \\
    c_X \sqrt{2} = - 2 c_Z + c_X \sqrt{2} + i c_Y \sqrt{2} \ .
\end{align}

The first equality gives $ 2 c_Z = - i \sqrt 2 c_Y$ which never holds for real non-zero coefficients. Hence this implies that $c_Z = c_Y = 0$. This in turn implies $c_X = 0$. Therefore there is no non-zero element $c_X \rho_{j=1}(\sigma_X) + c_Y \rho_{j=1}(\sigma_Y) + c_Z \rho_{j=1}(\sigma_Z)$ which leaves the ray $(\ket{-2} +  \ket 0 + \ket 2)$ invariant. 

The ray $ \ket{-2} +  \ket 0 + \ket 2$ therefore can be used as a seed state to generate the manifold $\SU(2)$ embedded in $\Cl^3$. 
We can generate a system of left coherent states:

\begin{align}
    \ket{\phi(e)} &= \ket{-2} +  \ket 0 + \ket 2 \\
    \ket{\phi(g = e^{i (c_X \sigma_X + c_Y \sigma_Y + c_Z \sigma_Z)})} &= e^{i (c_X \rho_{j=1}(\sigma_X) + c_Y \rho_{j=1}(\sigma_Y) + c_Z \rho_{j=1}(\sigma_Z))} \ket{-2} +  \ket 0 + \ket 2
\end{align}

\section{Proofs of examples in Section~\ref{sec_examples}}\label{sec: app proofs of examples}

\subsection{Representation theory of $\su2$}\label{app:rep_theory_su2}

The Lie Algebra $\su2$ of $\SU(2)$ is spanned by the Pauli matrices:
\begin{align}
    \sigma_Z = \begin{pmatrix}
        1 & 0 \\
        0 & -1
    \end{pmatrix}  , \ 
     \sigma_X = \begin{pmatrix}
        0 & 1 \\
        1 & 0
    \end{pmatrix}  , \ 
    \sigma_Y = \begin{pmatrix}
        0 & -i \\
        i & 0
    \end{pmatrix}  .
\end{align}

We note that  $\su2$ is a real vector space. The complexification of $\su2$ is isomorphic to $\sl2c$, the Lie Algebra of the special linear group on $\Cl^2$: $\su2 \otimes \Cl \simeq \sl2c$. A basis for  $\sl2c$ (a complex vector space) is given by:
\begin{align}
    Z = \begin{pmatrix}
        1 & 0 \\
        0 & -1
    \end{pmatrix}  , \ 
     X = \begin{pmatrix}
        0 & 1 \\
        0 & 0
    \end{pmatrix}  , \ 
    Y = \begin{pmatrix}
        0 & 0 \\
        1 & 0
    \end{pmatrix}  ,
\end{align}

with commutation relations
\begin{align}
     [Z,X] = 2 X, \ [Z, Y] = -2 Y,  \  [X, Y] = Z .
\end{align}

The fundamental representation of $\sl2c$ is given by the action of these matrices on $\Cl^2$. This is the spin $j= \frac{1}{2}$ representation of $\su2$. The space $\Cl^2$ decomposes into two eigenspaces of $Z$: 
\begin{align}
    \Cl^2 \simeq \ch_1 \oplus \ch_{-1} \ .
\end{align}

Let us denote the eigenbasis by $\ket{1}$ and $\ket{-1}$. In this basis $X$ acts as a raising operator, and $Y$ as a lowering operator: 
\begin{align}
    X \ket{-1} = \ket{1} ,\ X \ket{1} = 0\ ,
\end{align}
\begin{align}
    Y \ket{-1} = 0,\ Y \ket{1} = \ket{-1}\ .
\end{align}

Let us consider the spin $j=1$ representation of $\su2$ generated by the following matrices: 
\begin{align}
    \rho_{j=1}(\sigma_X) = 
    \begin{pmatrix}
        0 & \sqrt{2} & 0 \\
        \sqrt{2} & 0 & \sqrt{2} \\
        0 & \sqrt{2} & 0 
    \end{pmatrix} , \ 
      \rho_{j=1}(\sigma_Y) = i
    \begin{pmatrix}
        0 & -\sqrt{2} & 0 \\
        \sqrt{2} & 0 & -\sqrt{2} \\
        0 & \sqrt{2} & 0 
    \end{pmatrix}, \ 
     \rho_{j=1}(\sigma_Z) = 
     \begin{pmatrix}
         2 & 0 & 0 \\
         0 & 0 & 0 \\
         0 & 0 & -2
     \end{pmatrix} .
\end{align}

We can define a basis for a representation of $\sl2c$ as: 
\begin{align}
    \rho_{j=1}(X) = 
    \begin{pmatrix}
        0 & \sqrt{2} & 0 \\
        0& 0 & \sqrt{2} \\
        0 & 0 & 0 
    \end{pmatrix} , \ 
      \rho_{j=1}(Y) = i
    \begin{pmatrix}
        0 & 0 & 0 \\
        \sqrt{2} & 0 & 0 \\
        0 & \sqrt{2} & 0 
    \end{pmatrix}, \ 
     \rho_{j=1}(\sigma_Z) = 
     \begin{pmatrix}
         2 & 0 & 0 \\
         0 & 0 & 0 \\
         0 & 0 & -2
     \end{pmatrix} ,
\end{align}
where $X$ and $Y$ are the raising and lowering operators respectively.
The space  $\Cl^3$ decomposes into the following $Z$ eigenspaces: 
\begin{align}
    \Cl^3  \simeq \ch_{-2} \oplus \ch_0 \oplus \ch_2\ .
\end{align}

The action of the above operators on the $Z$ eigenstates $\ket{-2},\ket{0},\ket{2}$ can be computed directly.

\subsection{Proof for three spin $j=1$ systems}\label{app:three_spin_1}

Let us consider three quantum systems $\Cl^3$ each carrying a $j=1$ irreducible representation of $\SU(2)$.
The physical subspace is given by the following constraint:
\begin{align}
    \rho_{j=1}(K) \otimes \I \otimes \I + \I \otimes \rho_{j=1}(K)+ \otimes \I + \I \otimes \I \otimes \rho_{j=1}(K) \ket{\psi_\phys} = 0 , \ \forall K \in \sl2c \ .
\end{align}

Since $X,Y,Z$ form a basis we just need to find the subspace invariant under these three elements of $\sl2c$. 
Let us find the subspace which is invariant under the $Z$ operator: $ \rho_{j=1}(Z) \otimes \I \otimes \I + \I \otimes \rho_{j=1}(Z)+ \otimes \I + \I \otimes \I \otimes \rho_{j=1}(Z)$. This is spanned by:
\begin{align}
    \ket{0,0,0}, \ \ket{0,-2,2} , \ \ket{0,2,-2} , \ \ket{-2,0,2} , \ \ket{2,0,-2} , \ \ket{-2,2,0} ,  \ \ket{2,-2,0} .
\end{align}

We need to find a vector in this subspace which is invariant under $ \rho_{j=1}(\sigma_X) \otimes \I \otimes \I + \I \otimes \rho_{j=1}(\sigma_X)+ \otimes \I + \I \otimes \I \otimes \rho_{j=1}(\sigma_X)$ and $ \rho_{j=1}(\sigma_Y) \otimes \I \otimes \I + \I \otimes \rho_{j=1}(\sigma_Y)+ \otimes \I + \I \otimes \I \otimes \rho_{j=1}(\sigma_Y)$ also.  
One can easily check that the following state is in the kernel of the three generators of the $\mathfrak{su}_2$ Lie algebra (not taking into account normalization here and in the following):
\begin{align}
    \ket{\psi_\phys} = ( \ket{0,-2,2} - \ket{2,-2,0}) + (\ket{2,0,-2} - \ket{0,2,-2} ) + (\ket{-2,2,0}-  \ket{-2,0,2} ) \ .
\end{align}

\subsubsection{Frame dependence of reduced system subspaces}

We can condition $\ket{\psi_\phys}$ on the state $\ket{\phi(e)}_\A = \ket{2}_\A + \ket{0}_\A + \ket{-2}_\A$ to obtain the reduced state $\ket{\psi(e)}_{\B\C|\A}$:
\begin{align}
    \ket{\psi(e)}_{\B\C|\A} = \ket{-2,2} - \ket{-2,0} + \ket{0,-2} - \ket{2,-2}  + \ket{2,0}-  \ket{0,2}  \ .
\end{align}

Let us define 
\begin{align}
    \ch_{\B\C|\A,G} = \spann(\{\ch_{\B\C|\A,g}|g \in \SU(2)\})\ .
\end{align}

The action of $\rho_{j=1}^\B(g) \otimes \rho_{j=1}^\C(g)$ decomposes into the following irreducible subspaces of $\SU(2)$: $0 \oplus 1 \oplus 2$. By finding which subspaces $\ket{\psi}_{\B\C|\A}$ has support in we can then determine $\ch_{\B\C|\A,G}$.

Let us decompose  $0 \oplus 1 \oplus 2$ into subspaces of $\rho_{j=1}^\B(Z) \otimes \I_\C + \I_\B \otimes \rho_{j=1}^\C(Z)$:
\begin{align}
    \ch_{-4} \oplus \ch_{-2} \oplus \ch_0 \oplus \ch_{2} \oplus \ch_{4} \ ,
\end{align}
where $\dim(\ch_{-4}) = \dim(\ch_{4}) = 1$, $\dim(\ch_{-2}) = \dim(\ch_{2}) = 2$ and $\dim(\ch_0) = 3$. 

\begin{align}
    \ch_{-4} &= \spann(\ket{-2,-2}) , \\
    \ch_{-2} &= \spann(\ket{-2,0}, \ket{0,-2}) ,  \\
    \ch_{0} &= \spann(\ket{0,0}, \ket{2,-2}, \ket{-2,2}) ,  \\
    \ch_{2} &= \spann(\ket{2,0}, \ket{0,2}) ,  \\
    \ch_{4} &= \spann(\ket{2,2}) ,
\end{align}

The trivial irrep $j = 0$ is spanned by a vector in $\ch_0$ which is annihilated by $\rho_{j=1}^\B(X) \otimes \I_\C + \I_\B \otimes \rho_{j=1}^\C(X)$. Hence it is spanned by:
\begin{align}
    \ket{Z=0, j=0} = \ket{0,0} - \ket{2,-2} - \ket{-2,2} \ .
\end{align}

The highest weight vector of the $j = 1$ irrep is the vector in $\ch_2$ which is annihilated by $\rho_{j=1}^\B(X) \otimes \I_\C + \I_\B \otimes \rho_{j=1}^\C(X)$ which is:
\begin{align}
    \ket{Z=2, j=1} = \ket{0,2} - \ket{2,0} \ .
\end{align}

We can obtain the other basis vectors for this representation by applying $\rho_{j=1}^\B(Y) \otimes \I_\C + \I_\B \otimes \rho_{j=1}^\C(Y)$:
\begin{align}
    \ket{Z=0, j=1} = \ket{-2,2} + \ket{0,0} - \ket{0,0} - \ket{2,-2} = \ket{-2,2} - \ket{2,-2} \ ,
\end{align}
which is orthogonal to $\ket{Z=0, j=0}$ as expected. Applying  $\rho_{j=1}^\B(Y) \otimes \I_\C + \I_\B \otimes \rho_{j=1}^\C(Y)$ once more gives:
\begin{align}
    \ket{Z=-2, j=1} =  \ket{-2,0} - \ket{0,-2} 
\end{align}

We therefore find the basis for the $j = 2$ representation:
\begin{align}
    \ket{Z=4, j=2} &= \ket{2,2}  \\
    \ket{Z=2, j=2} &= \ket{0,2} + \ket{2,0} \\
    \ket{Z=0, j=2} &= 2\ket{0,0} + \ket{2,-2} + \ket{-2,2} \\
    \ket{Z = -2, j = 2} &=\ket{-2,0} + \ket{0,-2} \\
    \ket{Z=-4, j=2} &= \ket{-2,-2} 
\end{align}

The state $\ket{\psi(e)}_{\B\C|\A} = \ket{-2,2} - \ket{-2,0} + \ket{0,-2} - \ket{2,-2}  + \ket{2,0}-  \ket{0,2}$ has no support on the $j=0$ irrep, it has support on the $j=1$ irrep and no support on the $j = 2$ irrep. Therefore:
\begin{align}
    \ch_{\B\C|\A,G} \cong \Cl^3 \ .
\end{align}

Since 
\begin{align}
    \ch_{\B\C|\A,G} \not \cong \ch_{\B\C|\A,g}\ ,
\end{align}
we have shown that the physical system subspace is dependent of the frame orientation. This is not surprising since the physical system Hilbert space is one dimensional, and conditioning with different $\bra{\phi(g)}$ gives different states.

We can re-write $\ket{\psi(e)}_{\B\C|\A}$ in the irrep labelled basis as:

\begin{align}
    \ket{\psi(e)}_{\B\C|\A} = \ket{Z=0, j=1} - \ket{Z=2, j=1} - \ket{Z=-2, j=1} 
\end{align}

Observe that by relabelling $\ket{2} \mapsto -\ket 2$ and $\ket{-2} \mapsto -\ket{-2}$ in Equation~\eqref{eq:zero_state} it follows that $\ket{\psi(e)}_{\B\C|\A}$ has a trivial stabilized under the action of $\SU(2)$ and therefore generates a set of $\SU(2)$ coherent states. This in turn implies

\begin{align}
    \ch_{\B\C|\A,g} \neq \ch_{\B\C|\A,k}, \forall g \neq k \in \SU(2) \ .
\end{align}

\subsection{Proof for four spin $j=1$ systems} \label{app: Four spin j=1 systems}
We consider four quantum systems $\Cl^3$, each carrying a $j=1$ irreducible representation of $SU(2)$. We find that the physical Hilbert space is three-dimensional and that all physical states satisfy the constraint
\begin{align}
    [ \rho_{j=1}(K) \otimes \I^{\otimes 3} + \I \otimes \rho_{j=1}(K) \otimes \I^{\otimes 2} + \I^{\otimes 2} \otimes \rho_{j=1}(K) \otimes \I + \I^{\otimes 3} \otimes \rho_{j=1}(K) ] \ket{\psi_\phys} = 0 , \ \forall K \in \sl2c \ .
\end{align}

The following three vectors form an orthogonal basis of the physical Hilbert space (not taking normalization into account here):
\begin{align}
\ket{v_1}=&\ket{2,-2,2,-2}-\ket{2,-2,0,0}+\ket{2,-2,-2,2}-\ket{0,0,2,-2}+\ket{0,0,0,0}-\ket{0,0,-2,2} \nonumber\\ &+\ket{-2,2,2,-2}-\ket{-2,2,0,0}+\ket{-2,2,-2,2}\ , \\
\ket{v_2}=&\ket{2,0,0,-2}-\ket{2,0,-2,0}-\ket{2,-2,2,-2}+\ket{2,-2,-2,2}-\ket{0,2,0,-2}+\ket{0,2,-2,0} \nonumber\\
&+\ket{0,-2,2,0}-\ket{0,-2,0,2}+\ket{-2,2,2,-2}-\ket{-2,2,-2,2}-\ket{-2,0,2,0}+\ket{-2,0,0,2}\ , \\
\ket{v_3}=&6\ket{2,2,-2,-2}-3\ket{2,0,0,-2}-3\ket{2,0,-2,0}+\ket{2,-2,2,-2}+2\ket{2,-2,0,0}+\ket{2,-2,-2,2}\nonumber\\ &+2\ket{0,0,2,-2}+4\ket{0,0,0,0}+2\ket{0,0,-2,2}-3\ket{0,-2,2,0} -3\ket{0,2,0,-2}-3\ket{0,2,-2,0} \nonumber\\ &-3\ket{0,-2,0,2}+\ket{-2,2,2,-2}+2\ket{-2,2,0,0}+\ket{-2,2,-2,2} -3\ket{-2,0,2,0}-3\ket{-2,0,0,2}\nonumber\\
&+6\ket{-2,-2,2,2} \ ,
\end{align}
such that every pure physical state is of the form 
\begin{align}
    \ket{\psi_\phys}=\alpha \ket{v_1}+\beta \ket{v_2}+\gamma \ket{v_3}\ ,\ \alpha, \beta, \gamma \in \Cl,\ |\alpha|^2+|\beta|^2+|\gamma|^2=1 \ .
\end{align}

Each system has coherent states $\{\ket{\phi(g)}|g \in \SU(2)\}$ with $\ket{\phi(e)} = \ket{-2} + \ket{0} + \ket{1}$. By conditioning on the first system $\A$ being in state $\ket{\phi(e)}_\A= \ket{-2} + \ket{0} + \ket{1}$, we get the conditional state
\begin{align}
    \ket{\psi(e)}_{\B\C\D|\A} = \alpha \ket{v_1'} + \beta \ket{v_2'} + \gamma \ket{v_3'} , 
\end{align}
where $\ket{v'_i}=\bra{\phi(e)}_\A \ket{v_i}$.
The conditional state $\ket{\psi(e)}_{\B\C\D|\A}$ is acted on by $U_\B(g)\otimes U_\C(g) \otimes U_\D(g)$ which decomposes as $1 \otimes 1 \otimes 1 = 0 \oplus 1^{\oplus 3} \oplus 2^{\oplus 2} \oplus 3$. As mentioned in the main text, we find a basis for each of these irreps in order to find the span of $U_\B(g)\otimes U_\C(g) \otimes U_\D(g)\ket{\psi(e)}_{\B\C\D|\A}$ for different $\ket{\psi(e)}_{\B\C\D|\A}$. \\

Let us start with the $j = 0$ irrep. As the following state is annihilated by 
$\rho_{j=1}^\B(K) \otimes \I_\C  \otimes \I_\D + \I_\B \otimes \rho_{j=1}^\C(K)  \otimes \I_\D + \I_\B \otimes  \I_\C \otimes \rho_{j=1}^\D(K)$ for $K=X,Y,Z$:
\begin{align}
    \ket{Z=0,j=0}=-\ket{2,0,-2}+\ket{2,-2,0}+\ket{0,2,-2}-\ket{0,-2,2}-\ket{-2,2,0}+\ket{-2,0,2} \ ,
\end{align}
it forms the basis vector for the one-dimensional invariant subspace. The conditional state $\ket{\psi(e)}_{\B\C\D|\A}$ has zero overlap with this state.\\

Now let us turn to the three $j = 1$ irreps.
The highest weight $Z=2$ subspace of the $j=1$ irrep is three-dimensional and spanned by the following three orthonormal basis vectors: 
\begin{align}
    \ket{1,Z=2,j=1}&= \ket{2,2,-2}-\ket{2,0,0}+\ket{2,-2,2} , \\
    \ket{2,Z=2,j=1}&= \ket{2,2,-2}-\ket{2,-2,2}-\ket{0,2,0}+\ket{0,0,2} , \\
    \ket{3,Z=2,j=1}&= \ket{2,2,-2}+2\ket{2,0,0}+\ket{2,-2,2}-3\ket{0,2,0}-3\ket{0,0,2}+6\ket{-2,2,2} \ .
\end{align}
By applying the lowering operator, we generate the full basis consisting of $9$ states. A straightforward calculation shows that the conditional state $\ket{\psi(e)}_{\B\C\D|\A}$ has non-zero overlap with all of these states. In addition, one can check that the overlap of the conditional state with all three highest-weight vectors is different. This means that applying $U_\B(g)\otimes U_\C(g) \otimes U_\D(g)$ to the conditional state does indeed span the full $9$-dimensional space.\\

Now, we consider the two $j = 2$ irreps. The $Z = 4$ subspace of the $j = 2$ irreps is two dimensional, and spanned by the following orthogonal basis vectors:
\begin{align}
\ket{1, Z=4, j=2}&=\ket{0,2,2} - \ket{2,2,0} \ , \\
\ket{2, Z=4, j=2}&=-\ket{0,2,2} - \ket{2,2,0} + 2\ket{2,0,2}\ .
\end{align}

By applying the lowering operator to each of these vectors, we can generate a full basis for the two $j = 2$ irreps which consists of $2\cdot(2j+1)=10$ vectors.
A straightforward calculation yields that the conditional state has non-vanishing, pairwise different overlap with three of these vectors. Applying the representation on it generates the full $10$-dimensional space.\\

Finally, we can find a basis of $2j+1=7$ vectors for the $j = 3$ irrep by taking the highest weight state $\ket{Z = 6, j = 3}=\ket{2,2,2}$ and applying the lowering operator $Y$ multiple times.
The conditional state has non-vanishing, pairwise different overlap with three of these vectors. Applying the representation on it generates the full $7$-dimensional space.\\

All in all, it turns out that the conditional state $\ket{\psi(e)}_{\B\C\D|\A}$ has support in the invariant subspace of the irreps $j=1$, $j=2$ and $j=3$ while it has no support in $j=0$. Now, let us see what this means in terms of the decomposition $1 \otimes 1 \otimes 1 = 0 \oplus 1^{\oplus 3} \oplus 2^{\oplus 2} \oplus 3$.

As there is only one copy of the three-dimensional irrep in the decomposition, it is easy to see that this gives rise to $2\cdot3+1=7$ dimensions of  $\ch_{\B\C\D|A,G}$ already. However, if there are several copies of the same irrep, the argument becomes sightly more involved. It is not immediately clear that every single vector contained in the basis contributes an additional dimension. Note that there are as many highest weight vectors as there are copies of the irrep in the decomposition. It is important to check whether the overlap that the conditional state has with the various highest weight vectors of the same irrep is different. Otherwise, if the overlap is the same for the highest weight vectors $u_1$ and $u_2$ of two copies of the same irrep, these could be rewritten in terms of $u_1+u_2$ and $u_1-u_2$. The first would still be a highest weight vector, and by applying $U_\B(g) \otimes U_\C(g) \otimes U_\D(g)$, one can still span the full invariant subspace. However, the conditional state would then have zero overlap with $u_1-u_2$. Then it follows that applying the representation will only cover the space of one copy of this irrep, giving rise to fewer additional dimensions of $\ch_{\B\C\D|A,G}$. 

Since the overlap with the three highest weight vectors of $j=1$ is pairwise different, just like the overlap with the two highest weight vectors of $j=2$, we get a contribution of $3\cdot (2\cdot1+1)=9$ additional dimensions due to $j=1$ and $2\cdot (2\cdot2+1)=10$ further ones due to $j=2$. Adding these to the $7$ dimensions we get from the $j=3$ irrep, we get $\dim \ch_{\B\C\D|A,G}=26$ in total, which is only one dimension short of the full kinematical system Hilbert space $\ch_\B \otimes \ch_\C \otimes \ch_\D$ (which is $3^3=27$ dimensional).

\section{Remarks on Group Averaging}\label{app_groupaverage}

A rigorous account of the integrals \eqref{proj} and \eqref{Gtwirl} presents certain mathematical difficulties which we address in brief in this appendix, focussing on the setting $G=\mathbb{R}$, referring to the main text and \cite{Giulini:1998kf} for discussion of the general case. We wish to provide some further intuition to the group averaging procedures, particularly the coherent group averaging and the associated map to, and construction of, the physical Hilbert space, highlighting some concrete aspects of the construction. The first task is to give meaning to the integral expressions. One possibility is to attempt to view both as Bochner integrals (i.e., with values in the Banach space of bounded operators on $\mathcal{H})$. However, this fails immediately for both \eqref{proj} and  \eqref{Gtwirl} if $G$ is not compact: as a function  $G \to B(\mathcal{H})$, $U$ is not Bochner integrable, since $f: G \to B(\mathcal{H})$ is Bochner integrable if and only if $\int ||f|| dg $ is finite, and $||U(g)|| =1$. Similarly, $||U(g)AU(g)^{\dagger}|| = ||A||$. As we shall see, even the milder
weak interpretation of the integrals, i.e., $\int \ip{\varphi}{U(g) \phi} dg$ do not converge in general, and even when they do, they do not necessarily give rise to a well defined Hilbert space operator. Instead, the setting of distributions can be used, which we now discuss for the case $G=R$.

\subsection{Coherent Group Averaging}

We consider the expression $\int_GU(g)dg$, where $U$ is a (strongly continuous) true (i.e., not projective) unitary representation in $\mathcal{H}$ of the locally compact unimodular second countable group $G$ and $dg$ represents Haar measure. If $G$ is compact, then we can define the (bounded) operator $\Pi$ through
\begin{equation}\label{eq:avapp}
    \ip{\varphi}{\Pi \phi} = \int_G \ip{\varphi}{U(g) \phi}dg,
\end{equation}
holding for all of $\varphi, \phi \in \mathcal{H}$.\footnote{Note that since $\mathcal{H}$ is complex we need actually only consider the case $\varphi = \phi$.} The function $f_{\varphi, \phi}:G \to \mathbb{C}$ defined by $f_{\varphi, \phi}(g) = \ip{\varphi}{U(g) \phi}$ is integrable since $|f_{\varphi, \phi}(g)| \leq ||\varphi||||\phi||$ and $\int_G dg$ is finite. For the same reason the sequilinear form $\Omega (\varphi,\phi)$ defined as the right hand side of \eqref{eq:avapp} is bounded and therefore defines a unique bounded operator (i.e., $\Pi$) in $\mathcal{H}$. If $G$ is not compact but $X \in \mathcal{B}(G)$ is a compact (Borel) set, then since Haar measure is a Radon measure the expression $\int_X U(g)dg$ can be understood in exact analogy to above and the integral converges.
In both cases, the resulting operator is then a well-defined bounded self-adjoint linear mapping  $\mathcal{H} \to \mathcal{H}$, and, for compact $G$, $\Pi$ in~(\ref{eq:avapp}) coincides with $\Pi_{\rm phys}$ in the main text. Moreover, as can be easily confirmed, any $\varphi$ in the image of $\Pi$ is invariant under all the $U(g)$. The idempotence is most easily confirmed by utilising the invariance of the vectors in the range of $\Pi$; given that $\Pi(\varphi)$ is invariant it is immediate that $\Pi (\Pi (\varphi)) = \Pi(\varphi)$. The physical Hilbert space can then be then identified with the image of $\Pi$, which is  closed by the boundedness and idempotence of $\Pi$, with the inherited inner product. 

In case $G$ is not compact, various difficulties arise. One is the identification, if one exists, of a set of vectors in $\mathcal{H}\times \mathcal{H}$ for which \eqref{eq:avapp} converges; another is finding a (dense) domain $\mathcal{D}(\Pi) \subset \mathcal{H}$ on which $\Pi$ is defined as a possibly unbounded operator. Both possibilities depend on $U$ --- for instance,  if $U$ is such that $U(g) = \mathbf{1}$ for all $g$, then the integral converges only if $\langle\varphi|\phi\rangle=0$, and no non-trivial subspace of vectors $S$ can be found such the integral converges for all ${\varphi},{\psi}\in S$. Perhaps surprisingly, even if $U$ is non-trivial, in general there is no dense domain on which $\Pi$ can be defined as a Hilbert space operator, as will be seen below from the simple analysis of the regular representation of the translation group on $L^2(\mathbb{R})$, and therefore there is no hope that the methods of unbounded operators can be used there or in the general case. Instead, as we will see, the theory of distributions provides the flexibility and generality required to both make sense of $\Pi$ as a map whose codomain is larger than $\mathcal{H}$, and for the existence of invariant states serving to define the physical Hilbert space.

\subsection{$G=\mathbb{R}$, single particle}\label{ss:gr}

Much of the general idea for how to avoid the obstacles described above can already be seen in the setting $G=\mathbb{R}$, realised as translations in $L^2(\mathbb{R})$, i.e., $(U(y)\varphi)(x)= \varphi (x-y)$. If $\mathbb{R}$ is the (gauge) symmetry group with $U$ the corresponding representation, physical states should be invariant under all the $U(y)$. The following standard fact indicates that for such invariance to be possible, we need more general objects that those afforded by the Hilbert space framework.

\begin{proposition}\label{prop:four}
Let $U:\mathbb{R}\to B(L^2(\mathbb{R}))$ be a strongly continuous unitary representation of $\mathbb{R}$, defined by $(U(y)\varphi)(x)= \varphi (x-y)$. Then if $y \neq 0$ $U(y)$ has no eigenvectors in $L^2$.
\end{proposition}
\begin{proof}
Suppose $\varphi \in L^2$ and satisfies $(U(y)\varphi)(x)= \varphi (x-y) \equiv \varphi_y(x) = \alpha \varphi (x)$ for some $\alpha \in \mathbb{C}$. Writing $\hat{\varphi}$ for the Fourier transform of $\varphi$, we find that $\hat{\varphi}_y(p)=e^{ipy}\hat{\varphi}(p) = \alpha \hat{\varphi}(p)$. Since $\alpha \in \mathbb{C}$ is fixed, $\hat{\varphi}$ is only non-zero at countably many points, and since it must be $L^2$, we must have $\varphi$ vanishing.
\end{proof}
However, as is well known to all physicists, there are functions (necessarily outside $L^2$) which do satisfy the eigenvalue equation, e.g., the functions $e^{ipx}$. Viewed as continuous linear functionals on the Schwartz space $\mathcal{S}$ (of smooth functions of rapid decrease), identified with a dense subspace of $L^2$, the functions
$e^{ipx}$ can be seen as {\it generalized eigenvectors} of the translation operator $U(y)$.
These are defined to be any linear functional $\Lambda$ on $\mathcal{S}$ for which
$\Lambda (U(y) \varphi) = \lambda \Lambda (\varphi)$ for all $\varphi \in \mathcal{S}$; $\lambda$ is the generalized eigenvalue. The functional $\Lambda$ corresponding to $e^{ipx}$ is $\Lambda(\varphi) = \int e^{ipx} \varphi (x) dx = \hat{\varphi}(p)$, and therefore $\Lambda(U_y\varphi) = \int e^{ipx} \varphi (x-y) dx = e^{ipy}\hat{\varphi}(p) = e^{ipy} \Lambda (\varphi)$. Note also that the Fourier transform is the arena for realising the regular representation as a direct integral decomposition of irreducible "components" $e^{ipx}$, which as we have noted do not lie in $L^2$ but play a major role in the harmonic analysis of $\mathbb{R}$ as a locally compact group.

These observations point directly to the theory of distributions as the natural setting
for solving the problem of finding physical "states" under the constraints of symmetry.  We therefore consider both Schwartz space $\mathcal{S}$ and its topological dual---the space of tempered distributions---which is large enough to contain both the invariant "states" and the delta "functions". As discussed in the main text, a concrete approach to constructing the physical Hilbert space of invariant states is through the coherent group averaging construction, i.e., the integral \eqref{eq:avapp}.

We must first address the issues of the convergence of the integral and the domain of $U$. It is clear from considering 
\begin{equation}\label{eq:gar}
    \int \ip{\varphi}{e^{ixP }\phi} dx = \int \left(\int \overline{\varphi (y)} \phi(y-x)dy \right)dx = \left(\int \overline{\varphi(x)}dx\right)\left(\int \phi(x)dx\right),
\end{equation} 
that the group average of $U(x)$, interpreted weakly, does not converge everywhere, e.g., for 
$\varphi(x)=\frac{1}{1+|x|} \in L^2(\mathbb{R}) \setminus L^1 (\mathbb{R})$. On the other hand, the right hand side shows that the integral does exist for any  $\varphi,\phi \in L^2(\mathbb{R})\cap L^1(\mathbb{R})$, which contains $\mathcal{S}$.
However, this is not enough to define $\Pi$ as a (possibly unbounded) linear operator acting in a dense domain of $L^2$. 

In fact, consider $\phi \in L^2 \cap L^1$ and suppose $\Pi \phi \in L^2$. Then $\Pi \phi$ defines the continuous (bounded) linear functional $\omega_{\Pi \phi}:(L^2)^* \to \mathbb{C}$ defined by the right hand side of \eqref{eq:gar}, holding for any $L^2$ function $\varphi$. However, this map is clearly not bounded, and therefore there is no
$\Pi \phi \in L^2$ corresponding to the functional. Thus we must either restrict further the domain of $\Pi$, or expand the possibilities for the range of $\Pi$, i.e., make precise the notion that $\Pi \phi \notin L^2$. Since $\varphi \in L^2$ is arbitrary the linear functional $\omega_{\Pi \phi}$ is only bounded when 
$\int \phi(x)dx = 0$ (on which $\Pi \phi = 0$), and such a collection of functions $\phi$ is not dense
in $L^2$. Another avenue is to not demand that there be an operator corresponding to 
$\Pi$: view $\omega_{\Pi \phi}$ as a bounded linear functional on some dense domain
of $L^2$, we will take the Schwartz functions, which is well defined as a functional
by the right hand side of $\eqref{eq:gar}$. Thus we can think of $\Pi \phi$ as living in $\mathcal{S}'$, i.e., $\Pi : \mathcal{S} \to \mathcal{S}'$ as an antilinear map. The resulting triple $(S,\mathcal{H},S')$ is of course an example of a rigged Hilbert space. The space $S'$ contains the familiar objects $\bra{x}, \bra{p}$, and $S^{\times}$ $\ket{x},\ket{p}$. Note that the space of tempered distributions is invariant under the appropriately extended Fourier transform. Note also that whereas evaluation at a point of $L^2$ functions is nonsense, it is perfectly sensible in $\mathcal{S}$. The elements of $S^{\times}$ can be viewed as generalized states. 

The extension $\tilde{U}$ of $U$ to $S'$ is given as $(\tilde{U}(x)\phi)(f) = \phi (U(-x)f)$, with those $\phi$ for which $\tilde{U}(x) \phi = \phi$ being in bijective correspondence with the generalized states which satisfy the invariance condition. The existence of such $\phi$ is readily established by noting that any polynomially bounded Riemann integrable (on $[-\ell, \ell]$ for all $\ell > 0$) function $\phi$ is a tempered distribution through the association $\phi(f) = \int \phi(x)f(x)dx$, and therefore we have $(\tilde{U}(x)\phi)(f) = \int \phi(y-x)f(y)dy$, and so setting $\phi (x-y)=K$ (as a function) satisfies the condition.

To summarise, we observe that $\Pi$ is ${\Pi}_{\rm phys}: S \to S'$ from the main text, which is exactly the rigging map in this example, defined by 
\begin{equation}
    ({\Pi}_{\rm phys} (f_1))(f_2):= \int dx \ip{f_1}{U(x)f_2},
\end{equation}
which converges because $f_{1,2} \in \mathcal{S}$. That $({\Pi}_{\rm phys} (f_1))$ is invariant under $\tilde{U}(x)$ follows directly from writing 
$(\tilde{U}(x)(f_1))(f_2) = \Pi(f_1)(U(-x)f_2)$ and exploiting the invariance of $dx$.

Note that the question of idempotence of ${\Pi}_{\rm phys}$ is not well posed. We also comment that, from a physical point of view, we should deal directly with $S^{\times}$, which is the actual home of the physical states, but all statements can be translated through the obvious bijection $\mathcal{S}' \to \mathcal{S}^{\times}$.  In the image of ${{\Pi}}_{\rm phys}$ we can introduce the inner product $({{\Pi}}_{\rm phys}(f_1),{{\Pi}}_{\rm phys}(f_2))= {{\Pi}}_{\rm phys}(f_1)(f_2)$. This is positive, i.e., $(\tilde{{\Pi}}_{\rm phys}(f),{{\Pi}}_{\rm phys}(f)) \geq 0$,
since, by Fubini's theorem, \eqref{eq:gar} can be written $\vert \int dx f(x) \vert ^2 \geq 0$. Up to the distinction between $\mathcal{S}'$ and $\mathcal{S}^{\times}$, this corresponds to the physical inner product. The closure of the image of ${{\Pi}}_{\rm phys}$ is then a Hilbert space corresponding (again, up to bijection) to the space of physical states. As a vector space the collection of invariant tempered distributions is easily seen to be one dimensional, which is not very interesting from a physical point of view. The two particle case, however, which follows the same constructions as above but with some superficial complications, gives rise to richer possibilities, pointing to the importance of
including quantum reference frames in the physical description.

\subsection{Incoherent Group Averaging/G-twirl}

The incoherent group averaging, or $G$-twirl given in \eqref{Gtwirl}, is a well defined mapping $B(\mathcal{H}) \to B(\mathcal{H})$ when $G$ is compact, since
$$ \left \vert \int_G \ip{ \varphi }{ U(g)AU(g)^{\dag} \varphi} dg \right \vert \leq \int_G |\ip{ \varphi }{ U(g)AU(g)^{\dag} \varphi}| dg \leq \int_G ||A|| ||\varphi||^2 dg < \infty,$$ by the boundedness of $A$ and the finiteness of Haar measure for compact groups. For unbounded (but self-adjoint) $A$ one can work directly with the spectral measure. If $G$ is non-compact however, the integral clearly does not converge in general (for instance, when $A$ commutes with all the $U(g)$.) Whenever $f^{\varphi}_A(g):= \langle \varphi |U(g) A U(g)^{\dag} \varphi \rangle$ for $f_A^{\varphi}:G \to \mathbb{R}$ has compact support for all $|\varphi \rangle$, the integral converges since Haar measure is a Radon measure and is therefore finite on compact sets. 

There are again various options regarding how to proceed. The first is, once $A$ is given, to hope to find some dense set of vectors in $\mathcal{H}$ on which \eqref{Gtwirl} does converge, define the resulting operator to be an unbounded operator with the same domain, and proceed in a similar way to above. However, it is not known for which $A$ such a procedure is feasible, and in the setting of this paper, the following strategy suffices. Informally write $A=\ket{\phi(g)}\bra{\phi(g)} \otimes f_S$; then $$\cg(A) = \int_G dg \ket{\phi(g'g)}\bra{\phi(g'g)} \otimes U_S(g')f_SU_S(g')^{\dagger}.$$ Now we use the fact that 
$X \mapsto \int_X \ket{\phi(g)}\bra{\phi(g)} dg \equiv E^G(X)$ is a (possibly not normalised, covariant) POVM, and therefore we can do integration with respect to $E^G$. Specifically, with, $g=e$, we note that the above integral can be written
$$ \int_G E(dg) \otimes  U_S(g)f_SU_S(g)^{\dagger}.$$
From here, the integral is defined on a dense subset of bounded operators on
$\mathcal{H}_S$, and on the whole space if $G$ is abelian - see Sec. 4.1 of \cite{loveridgeSymmetryReferenceFrames2018a}. The case $g \neq e$ causes no additional difficulties.

\end{document}